\documentclass[submission,copyright,creativecommons]{eptcs}
\providecommand{\event}{QPL 2024}
\pdfoutput=1

\usepackage{amsthm, amsmath, amssymb, dsfont, physics, textcomp, tikz}
\usepackage{braket}
\usepackage{tikz-cd}
\usepackage{hyperref}
\usepackage{stmaryrd}
\usepackage[capitalize,nameinlink,sort,noabbrev]{cleveref}
\Crefname{equation}{Relation}{Relations}
\usepackage{algorithm}
\usepackage{algpseudocode}
\usepackage{booktabs}
\usepackage{subcaption}
\usepackage{xspace}
\usepackage{quiver}
\usepackage{underoverlap}
\usepackage{pdflscape}
\usepackage{multicol}
\usepackage{xurl}

\newcommand{\Z}{\mathbb{Z}}
\newcommand{\N}{\mathbb{N}}

\newcommand{\D}{\mathbb{D}}
\newcommand{\Q}{\mathbb{Q}}
\newcommand{\R}{\mathbb{R}}
\newcommand{\interp}[1]{[\![#1]\!]}

\newcommand{\tool}{\textsc{Tietze}\xspace}
\newcommand{\toolurl}{https://github.com/meamy/tietze}

\newcommand{\GateAt}[2]{{#1}_{#2}}
\newcommand{\CGate}[3]{\GateAt{C{#1}}{{#2},{#3}}}
\newcommand{\CCGate}[3]{\GateAt{CC{#1}}{{#2},{#3}}}

\allowdisplaybreaks
\usepackage{tabularx}

\newtheoremstyle{break}
  {}
  {}
  {\itshape}
  {}
  {\bfseries}
  {.}
  {\newline}
  {}

\theoremstyle{plain}
\newtheorem{theorem}{Theorem}[section]
\newtheorem{lemma}[theorem]{Lemma}

\newtheorem*{proposition*}{Proposition}
\newtheorem{corollary}[theorem]{Corollary}
\theoremstyle{break}

\theoremstyle{definition}
\newtheorem{definition}[theorem]{Definition}

\newtheorem{example}[theorem]{Example}
\theoremstyle{remark}

\usepackage{xcolor}
\hypersetup{
    colorlinks,
    linkcolor={red!80!black},
    citecolor={green!80!black},
    urlcolor={blue!80!black}
}

\begin{document}

\newcommand{\papertitle}
           {A Sound and Complete Equational Theory for 3-Qubit
           Toffoli-Hadamard Circuits}   

\title{\papertitle}
\author{Matthew Amy
        \institute{Simon Fraser University \\ Burnaby, Canada}
        \email{matt\_amy@sfu.ca}
        \and
        Neil J. Ross
        \institute{Dalhousie University \\ Halifax, Canada}
        \email{neil.jr.ross@dal.ca}
        \and
        Scott Wesley
        \institute{Dalhousie University \\ Halifax, Canada}
        \email{scott.wesley@dal.ca}}

\newcommand{\titlerunning}{\papertitle}
\newcommand{\authorrunning}{M. Amy, N. J. Ross \& S. Wesley}

\maketitle

\begin{abstract}
We give a sound and complete equational theory for $3$-qubit quantum circuits over the Toffoli-Hadamard gate set $\{ X, CX, CCX, H\}$.
That is, we introduce a collection of true equations among Toffoli-Hadamard circuits on three qubits that is sufficient to derive any other true equation between such circuits.
To obtain this equational theory, we first consider circuits over the Toffoli-$K$ gate set $\{ X, CX, CCX, K\}$, where $K=H\otimes H$.
The Toffoli-Hadamard and Toffoli-$K$ gate sets appear similar, but they are crucially different on exactly three qubits.
Indeed, in this case, the former generates an infinite group of operators, while the latter generates the finite group of automorphisms of the well-known $E_8$ lattice.
We take advantage of this fact, and of the theory of automorphism groups of lattices, to obtain a sound and complete collection of equations for Toffoli-$K$ circuits.
We then extend this equational theory to one for Toffoli-Hadamard circuits by leveraging prior work of Li \textit{et al.} on Toffoli-Hadamard operators.
\end{abstract}

\section{Introduction}

The \emph{Toffoli-Hadamard} gate set is obtained by extending the
classical reversible gate set $\{ X, CX, CCX\}$ with the Hadamard gate
$H$. The addition of the Hadamard gate promotes the gate set $\{ X,
CX, CCX\}$ from one that is universal for classical reversible
computation to one that is universal for quantum computation
\cite{aharonov2003,Shi2003}. Because the Hadamard gate can introduce
phases of $(-1)$ and produce superpositions, one can think of the
addition of the Hadamard gate as a simple way to augment classical
reversible computation with these typically quantum features. In turn, this motivates the study of
Toffoli-Hadamard circuits
\cite{aharonov2003,AGLR2023,AmyGlaudell2020,dalla2013toffoli,LiRoss2021,Shi2003,vilmart2018zx,vilmart2023completeness}.

In recent years, an important effort has been made to understand
quantum circuits equationally. If $G$ is a set of quantum gates, an
\emph{equational theory} for $G$ is given by a set of \emph{equations} (or \emph{relations}) among the circuits over
$G$. The equational
theory is \emph{sound} if it only equates circuits that correspond to
the same operator, and \emph{complete} if it always equates circuits
that correspond to the same operator. Equational theories can be used
to optimize and verify quantum circuits in practice, but, more
fundamentally, they can illuminate the mathematical structure
underlying the gate set $G$. Sound and complete equational theories
have been found for several important gate sets
\cite{Amy2016AFP,BianSelinger2022,BS2023-cliffordcs3,Clement2022,MakaryRoss2021,Sel2014-clifford}.

In this paper, we give a sound and complete equational theory for
3-qubit Toffoli-Hadamard circuits. A presentation for the group of
Toffoli-Hadamard operators was given in \cite{LiRoss2021}, but the
presentation uses 1-, 2-, and 4-level operators
as generators. While these operators can be represented by Toffoli-Hadamard
circuits, this leads to an unnatural presentation, from the
perspective of quantum circuits. What is more, the presentation of
\cite{LiRoss2021} contains over 2000 relations, even when restricted to 3-qubit operators. Many of these
relations can be presented concisely as relation schemas in the
language of operators, but these relations need to be expanded to be
stated in the language of circuits. In contrast, our presentation
contains only 65 relations, most of which are natural from the
perspective of quantum circuits.

To obtain our presentation, we first consider circuits over the
Toffoli-$K$ gate set $\{ X, CX, CCX, K\}$, where $K=H\otimes H$.  The
Toffoli-Hadamard and Toffoli-$K$ gate sets appear similar, but they
are crucially different on exactly three qubits.  Indeed, in this
case, the former generates an infinite group of operators, while the
latter generates the finite group of automorphisms of the well-known $E_8$ lattice. The correspondence between 3-qubit Toffoli-$K$ circuits
and the automorphisms of the $E_8$ lattice was previously known (see
\cite{Goucher2020,Planat2023}). We take advantage of this correspondence, and of the theory of automorphism groups of lattices, to obtain a sound and
complete collection of equations for Toffoli-$K$ circuits. The
automorphism group of the $E_8$ lattice admits a finite Coxeter
presentation, which enjoys many geometric and combinatorial
properties, and we use \emph{Tietze transformations} to turn the Coxeter
presentation of the group of Toffoli-$K$ operators into a concise
circuit presentation. We then extend this equational theory to one for
Toffoli-Hadamard circuits by building upon \cite{LiRoss2021}. Our paper
therefore regards the group of $3$-qubit Toffoli-Hadamard circuits as an extension of the automorphism group of the $E_8$ lattice in order to elucidate its underlying
mathematical structure.

The paper is organized as follows.
In \cref{Sec:Groups}, we define three groups of interest.
In \cref{Sec:Prelim}, we recall prior results on finite group presentations and we review Tietze transformations.
In \cref{Sec:E8}, we use the theory of Coxeter groups to obtain a presentation for the group of $3$-qubit Toffoli-K circuits using a minimal number of generators. We moreover show that every operator in this group can be represented by a circuit of Toffoli-count at most $120$.
In \cref{Sec:O8D,Sec:TofH}, the results of Li \emph{et al.} \cite{LiRoss2021} are used to extend this presentation to a presentation for $3$-qubit circuits over the gate set $\{ X, CX, CCX, K, CCZ\}$, and then to one for 3-qubit Toffoli-Hadamard circuits. Our approach relies on a large number of derivations and intricate rewriting proofs, which we relegate to several appendices and a supplement~\cite{supp}.

\section{Three Groups and Their Generators}
\label{Sec:Groups}

Let $\Z$ denote the ring of integers.
The \emph{half-integers} $\Z+1/2$ are defined as $\Z+1/2 = \{a+1/2 \mid a \in \Z\}$ and the ring of \emph{dyadic fractions} $\D$ is defined as $\D = \Z[1/2] = \{a/2^k \mid a \in \Z \mbox{ and } k\in \N \}$.
Equivalently, $\D$ is the smallest subring of $\Q$ that contains both $\Z$ and $1/2$.
The \emph{$E_8$ lattice} $\Gamma_8$ is the following collection of 8-dimensional vectors,
\[
\Gamma_8 = \left\{ x \in \Z^8 \cup (\Z+1/2)^8 \;\middle\vert\; \sum x_i \equiv 0 \pmod{2} \right\}.
\]
In other words, $\Gamma_8$ consists of the vectors in $\R^8$ whose components sum to an even integer and are either all integers or all half-integers.
The $E_8$ lattice is well-studied because it enjoys many remarkable properties \cite{ConwaySloane1987}; in particular, it provides the densest sphere packing in dimension 8 \cite{packing}.

We now introduce the three groups that will be the focus of this paper.
Let $R$ be a ring.
For each $n \in \N$, let $\mathrm{GL}( n, R )$ denote the general linear group over $R$ in dimension $n$ and let $\mathrm{O}( n, R )$ denote the orthogonal group over $R$ in dimension $n$.
Define $W(E_8)$ to be the subgroup of $\mathrm{O}(8,\D)$ consisting of the elements of $\mathrm{O}(8,\D)$ that fix the $E_8$ lattice.
Define $\textsf{TofH}(n)$ to be the subgroup of $\mathrm{O}(2^n,\Z[1/\sqrt{2}])$ consisting of matrices $M / \sqrt{2}{}^k$, where $M$ is an integer matrix and $k \in \N$.
We will be interested in the groups $W(E_8)$, $\mathrm{O}(8,\D)$, and $\textsf{TofH}(3)$. Note that we have $W(E_8) \le \mathrm{O}(8,\D) \le \textsf{TofH}(3)$.

The above three groups are generated by well-known quantum gates.
Let $I$ denote the $2\times 2$  identity matrix and $\otimes$ denote the Kronecker tensor product.
Given a dimension $2$ matrix $M$, define $\GateAt{M}{0} = M \otimes I \otimes I$, $\GateAt{M}{1} = I \otimes M \otimes I$, and $\GateAt{M}{2} = I \otimes I \otimes M$.
That is, $\GateAt{M}{j}$ applies operator $M$ to the $j$-th qubit.
Furthermore, define $\CGate{M}{j}{k}$ to be the operator that sends each standard basis state $\ket{x_0 x_1 x_2}$ to $(\GateAt{M}{k})^{x_j} \ket{x_0 x_1 x_2}$.
That is, $\CGate{M}{j}{k}$ applies operator $M$ to the $k$-th qubit whenever the $j$-th qubit is in the basis state $\ket{1}$.
Likewise, define $\CCGate{M}{j}{k}$ to be the operator that sends each standard basis state $\ket{x_0 x_1 x_2}$ to $(\GateAt{M}{l})^{x_j x_k} \ket{x_0 x_1 x_2}$ for $l \in \{ 1, 2, 3 \} \setminus \{ j, k \}$.
That is, $\CCGate{M}{j}{k}$ applies operator $M$ to the $l$-th qubit whenever the $j$-th and $k$-th qubits are both in the basis state $\ket{1}$.
The operators $\CGate{M}{j}{k}$ and $\CCGate{M}{j}{k}$ denote the usual \emph{controlled-$M$ gate} and \emph{doubly-controlled-$M$ gate}, respectively.
Now recall the \emph{Pauli $X$}, \emph{Pauli $Z$}, \emph{Hadamard}, and $K$ matrices,
\[
    X = \begin{bmatrix}
        0 & 1 \\
        1 & 0
    \end{bmatrix},
    \qquad
    Z = \begin{bmatrix}
        1 & 0 \\
        0 & -1
    \end{bmatrix},
    \qquad
    H =
    \frac{1}{\sqrt{2}}
    \begin{bmatrix}
        1 & 1 \\
        1 & -1
    \end{bmatrix},
    \qquad
    \mbox{and}
    \qquad
    K_{j,k} = H_j \circ H_k,
\]
where $(\circ)$ denotes matrix multiplication.
Then $\CCGate{X}{j}{k}$ denotes the \emph{Toffoli} gate.
Note that the matrices of the form $\GateAt{X}{j}$, $\GateAt{K}{{j,k}}$, $\CGate{X}{j}{k}$, $\CCGate{X}{j}{k}$, and $\CCGate{Z}{j}{k}$ belong to $\mathrm{O}(8,\D)$.
It is known that $\{\GateAt{X}{j},\CGate{X}{j}{k},\CCGate{X}{j}{k},\GateAt{K}{{j,k}}\}$ is a generating set for $W(E_8)$ \cite{ConwaySloane1987}.
Similarly it is known that $\{\GateAt{X}{j},\CGate{X}{j}{k},\CCGate{X}{j}{k},\GateAt{K}{{j,k}}, \CCGate{Z}{j}{k} \}$ and $\{\GateAt{X}{j},\CGate{X}{j}{k},\CCGate{X}{j}{k},\GateAt{H}{j}\}$ are generating sets for $\mathrm{O}(8,\D)$ and $\textsf{TofH}(3)$, respectively \cite{AmyGlaudell2020}\footnote{The generator $CCZ_{j,k}$ is necessary to apply these results to the ancilla-free three-qubit case.}.

\section{Presentations and Tietze Transformations}
\label{Sec:Prelim}

We now recall key results from combinatorial group theory. In particular, we discuss monoid presentations, as well as \emph{Tietze transformations}, which will play an important role in the rest of the paper.

\subsection{Presentations}

Let $\Sigma$ be an alphabet (i.e., a set of symbols).
Then $\Sigma^*$ is the \emph{free monoid on $\Sigma$}.
The elements of $\Sigma^*$ are the words over $\Sigma$, the monoid operation is string concatenation, which we denote by $(\cdot)$, and the identity element in $\Sigma^*$ is the empty word, which we denote by $\epsilon$.

If $G$ is a monoid and $Q$ is a quotient of $G$, then we write $\pi_Q: G \twoheadrightarrow Q$ to denote the canonical projection of $G$ onto $Q$.
Given a subset $R$ of $\Sigma^{*} \times \Sigma^{*}$, we write $Q = \langle \Sigma \mid R \rangle$ to denote the largest quotient of $G = \Sigma^*$ such that $\pi_Q( q ) = \pi_Q( r )$ for all $( q, r ) \in R$.
If $M \cong \langle \Sigma \mid R \rangle$, then we say that $\langle \Sigma \mid R \rangle$ is a \emph{presentation of $M$} and write $q \approx_R$ for each $( q, r ) \in R$.
The elements of $\Sigma$ are called \emph{generators} and the elements of $R$ are called \emph{relations}.
If, for each $x \in \Sigma$, there exists a $w \in \Sigma^*$ such that $\pi_Q(x \cdot w) = \pi_Q(\epsilon)$, then $M$ is a group and $\langle \Sigma \mid R \rangle$ is a \emph{monoid presentation for the group} $M$.
In either case, if $\Sigma$ and $R$ are finite, then $\langle \Sigma \mid R \rangle$ is a \emph{finite presentation}.
We distinguish between the presentations $\langle \Sigma \mid R \rangle$ and $\langle \Sigma \mid R' \rangle$ whenever $R \ne R'$, even if $R$ and $R'$ generate the same quotient.

Certain aspects of presentations can be conveniently expressed in the language of string rewriting.
Let $\Sigma$ be an alphabet and $R \subseteq \Sigma^* \times \Sigma^*$.
Fix some $u \in \Sigma^*$ and $v \in \Sigma^*$.
If there exists some $( q, r ) \in R$ and $s, t \in \Sigma^*$ such that $u = s \cdot q \cdot t$ and $v = s \cdot r \cdot t$, then we write,
\[
u \xrightarrow{R} v. 
\]
If either $u \xrightarrow{R} v$ or $u \xleftarrow{R} v$, then we write $u \stackrel{R}{\leftrightarrow} v$.
We say that \emph{$u$ rewrites to $v$}, denoted $u \sim_R v$, if either $u = v$ or there exists a finite sequence,
\[
u
\stackrel{R}{\leftrightarrow}
w_1
\stackrel{R}{\leftrightarrow}
w_2
\stackrel{R}{\leftrightarrow}
\cdots
\stackrel{R}{\leftrightarrow}
w_n
\stackrel{R}{\leftrightarrow}
v.
\]
That is, $( \sim_R )$ is the symmetric, transitive and reflexive closure of $\xrightarrow{R}$. Importantly, $\pi_Q( u ) = \pi_Q( v )$ in $Q = \langle \Sigma, R \rangle$ if and only if $u \sim_R v$ \cite[Ch.~7]{BookOtto1993}.
That is, two words $u$ and $v$ represent the same element in $G$ if and only if the relations in $R$ suffice to rewrite $u$ into $v$.
In this sense, the relations in $R$ define a \emph{complete equational theory} for the monoid $Q$ with respect to the generators $\Sigma$.
For further information on presentations and on rewriting, the reader is encouraged to consult \cite{Johnson1990} and \cite{BookOtto1993}, respectively.

\subsection{Tietze Transformations}

The transformations, which we state formally below, allow one to add a generator, remove a generator, add a relation, and remove a relation.
Let $\Sigma$ be an alphabet, $R \subseteq \Sigma^{*} \times \Sigma^{*}$, and $G = \langle \Sigma \mid R \rangle$ be a monoid.
\begin{itemize}
    \item[--] \textbf{Gen($+$)}.
          Let $x$ be a symbol.
          If $x \not \in \Sigma$ and $w \in \Sigma^*$, then $G \cong \langle \Sigma \cup \{ x \} \mid R \cup \{ x \approx w \} \rangle$.
    \item[--] \textbf{Gen($-$)}.
          Let $x \in \Sigma$, $x \approx_R w$, $\Pi = \Sigma \setminus \{ x \}$, and $Q = R \setminus \{ x \approx w \}$.
          If $Q \subseteq \Pi^* \times \Pi^*$, then $G \cong \langle \Pi \mid Q \rangle$.
    \item[--] \textbf{Rel($+$)}.
          If $q \sim_R r$, then $G \cong \langle \Sigma \mid R \cup \{ q \approx r \} \rangle$.
    \item[--] \textbf{Rel($-$)}.
          Let $q \approx_R r$ and $Q = R \setminus \{ q \approx r \}$.
          If $q \sim_Q r$, then $G \cong \langle \Sigma \mid Q \rangle$.
\end{itemize}
The \textbf{Gen($+$)} rule states that one can add a generator if one also adds a relation defining it in terms of the other generators. The \textbf{Gen($-$)} rule states that a generator can be removed if it is defined in terms of the other generators, and does not appear in any of the other relations. The \textbf{Rel($+$)} rule states that if a relation can be derived from the existing ones, then it can be added to the set of relations. Finally, the \textbf{Rel($-$)} rule conversely states that if a relation can be derived from other relations in the presentation, it is redundant and can be removed.

Tietze transformations are sound and complete for the isomorphism of finite monoid presentations. That is, two presentations $\langle \Sigma \mid R \rangle$ and $\langle \Pi \mid Q \rangle$ specify the same monoid if and only if $\langle \Sigma \mid R \rangle$ can be obtained from $\langle \Pi \mid Q \rangle$ through a finite sequence of Tietze transformations~\cite[Section~1]{HenryMimram2022}.

The goal of this paper is to find presentations for groups of quantum operators in which each generator corresponds to a specific quantum gate.
More explicitly, given a group $G$, a generating set $\Sigma$, and a \emph{semantic interpretation} $\interp{\cdot}_{\Sigma}: \Sigma \to G$, our goal is to find a set of relations $R \subseteq \Sigma^{*} \times \Sigma^{*}$ such that $\interp{\cdot}$ induces an isomorphism between $\langle \Sigma \mid R \rangle$ and $G$.
In what follows, we start from a known presentation $\langle \Pi \mid Q \rangle$ over different generators $\Pi$ with a semantic interpretation $\interp{\cdot}_\Pi: \Pi \to G$, and obtain $\langle \Sigma \mid R \rangle$ via a sequence of Tietze transformations.
As these Tietze transformations act on the abstract group $\langle \Pi \mid Q \rangle$, one must ensure that the transformations respect the intended interpretation $\interp{\cdot}_{\Sigma}$ of the new generators in $\Sigma$, as discussed further in \cref{Appendix:Tietze}.

\section{From Coxeter to Circuit Presentations of \texorpdfstring{$W(E_8)$}{W(E8)}}
\label{Sec:E8}

A \emph{Coxeter group} is a group $G$ which admits a group presentation of the form $\langle r_1, \ldots, r_n \mid (r_j r_k)^{N_{j,k}} \approx \epsilon\rangle$,
where $N$ is an $n \times n$ matrix over $\N \cup \{ \infty \}$ such that $N_{j,j} = 1$ and $N_{j,k} > 1$ for all $j \ne k$~\cite{Humphreys1990}.
The matrix $N$ is known as the \emph{Coxeter matrix} of $G$.
Note that since $\pi_G( r_j \cdot r_j ) = \epsilon$ for each $r_j$, then every Coxeter presentation is automatically a monoid presentation for a group.
Coxeter groups are an abstraction of reflection groups and, in particular, for every finite Coxeter group $G$, there is a faithful group representation $G \to \mathrm{O}(n )$ that maps each $r_j$ to a  reflection in $\R^n$~\cite{Humphreys1990}.
Recall that a \emph{Householder transformation} is a reflection about the hyperplane normal to some vector $\alpha \in \R^n$ defined by $v \mapsto v - 2 \frac{\langle v, \alpha \rangle}{\langle \alpha, \alpha \rangle} \alpha$~\cite{Humphreys1990}.
If $r$ is the reflection about the hyperplane normal to $\alpha \in \R^n$ and $M \in \mathrm{O}( n )$, then $M \circ r \circ M^{-1}$ is the reflection about the hyperplane normal to $M \alpha$~\cite[Prop.~1.2]{Humphreys1990}.
As a special case, $v$ and $-v$ define the same reflection.

The goal of this section is to construct a presentation for the Weyl group of the $E_8$ lattice in terms of Toffoli-K gates. 
Recall that the Weyl group for any lattice $L \subseteq \R^n$ is the finite reflection group generated by reflections about the roots of $L$ (see~\cite[Sec.~2.9]{Humphreys1990}).
That is, given a root system $\Phi$ for $L$, the group $W( L )$ is generated by $\{ r_\alpha : \alpha \in \Phi \}$ where $r_\alpha$ is the reflection through the hyperplane normal to $\alpha$.
Consequently, $W( E_8 )$ is a Coxeter group.
A root system and the corresponding Coxeter matrix for $W( E_8 )$ are given in \cref{Fig:E8Desc}.

\begin{figure}[t]
  \centering
  \begin{subfigure}[t]{0.58\textwidth}
    \centering
    \small
    \input{e8_roots}
    \subcaption{$E_8$ Root System.}
  \end{subfigure}
  \hfill
  \begin{subfigure}[t]{0.38\textwidth}
    \centering
    \small
    \input{e8_coxeter}
    \subcaption{$W( E_8 )$ Coxeter Matrix.}
  \end{subfigure}
  \caption{The root system and Coxeter matrix for $W( E_8 )$. Note that the root system consists of $8$ vectors and are presented as the columns of an $8 \times 8$ matrix.}
  \label{Fig:E8Desc}
\end{figure}

\begin{figure}[t]
  \centering
  {\small\begin{tikzcd}[column sep=4em, row sep=4em]
  \langle \Sigma_{E8} \mid R_{E8} \rangle
  \arrow["\textbf{Gen(}+\textbf{)}"]{r}
  \arrow[dotted, "\cong"]{dd}
  &
  \langle \Sigma_{E8} \cup \Sigma_{D} \mid R_{E8} \cup R_{D(E8)} \rangle
  \arrow["\textbf{Rel(}+\textbf{)}"]{r}
  &
  \langle \Sigma_{E8} \cup \Sigma_{D} \mid R_{E8} \cup R_{D(E8)} \cup R_{E8(D)} \rangle
  \arrow["\textbf{Rel(}+\textbf{)}"]{d}
  \\
  &&
  \langle \Sigma_{E8} \cup \Sigma_{D} \mid R_{E8} \cup R_{D(E8)} \cup R_{E8(D)} \cup R_{D} \rangle
  \arrow["\textbf{Rel(}-\textbf{)}"]{d}
  \\
  \langle \Sigma_{D} \mid R_{D} \rangle
  &
  \langle \Sigma_{E8} \cup \Sigma_{D} \mid R_{E8(D)} \cup R_{D} \rangle
  \arrow["\textbf{Gen(}-\textbf{)}"]{l}
  &
  \langle \Sigma_{E8} \cup \Sigma_{D} \mid R_{E8}  \cup R_{E8(D)} \cup R_{D} \rangle
  \arrow["\textbf{Rel(}-\textbf{)}"]{l}
\end{tikzcd}}%
  \caption{A diagrammatic summary of the Tietze transformations used to obtain a presentation for $W(E_8)$. Note that in this diagram $\Sigma_{E8}$ denotes the Coxeter generators, $\Sigma_{D}$ denotes the dyadic Toffoli-Hadamard gates, $R_{D(E8)}$ expresses the $\Sigma_{D}$ in terms of $\Sigma_{E8}$, and $R_{E8(D)}$ expresses $\Sigma_{E8}$ in terms of $\Sigma_{D}$.}
  \label{Fig:E8ProofSum}
\end{figure}

To obtain a presentation in terms of Toffoli-K gates, we begin with the Coxeter presentation of $W( E_8 )$.
The desired presentation is then obtained through a sequence of Tietze transformations.
First, the \textbf{Gen($+$)} rule is used to introduce the dyadic Toffoli-K gates with their intended semantics.
Second, the \textbf{Rel($+$)} rule is used to rewrite the Coxeter generators in terms of Toffoli-K gates (call these relations $R_{E8(D)}$).
Third, the \textbf{Rel($+$)} rule is used to introduce well-known relations satisfied by the Toffoli-K gates (see, e.g.,~\cite{LiRoss2021,MakaryRoss2021}).
Given these new relations, the \textbf{Rel($-$)} rule is used to eliminate all defining relations for the Toffoli-K gates.
In a similar fashion, the \textbf{Rel($-$)} rule is also used to eliminate the Coxeter relations of $W( E_8 )$.
At this point, the Coxeter generators only appear in $R_{E8(D)}$, and can be eliminated using the \textbf{Gen($-$)} rule.
What remains is a presentation of $W(E_8)$ in terms of Toffoli-K gates.
All steps of this proof are summarized in \cref{Fig:E8ProofSum}.

Each step of this proof requires numerous applications of the corresponding Tietze transformation.
To establish that each \textbf{Gen($+$)} and \textbf{Rel($+$)} transformation holds, an equation of $8 \times 8$ matrices must be validated.
To establish that each \textbf{Rel($-$)} transformation holds, a derivational proof must be validated.
In both cases, the proof obligation is computational in nature.
The validity of our Tietze transformations have been machine-verified by the software package \tool\footnote{Available at: \url{\toolurl}.}.

\subsection{Introducing the Toffoli-K Gates}
\label{Sec:E8:DIntro}

The generators of $W( E_8 )$ can be written as follows.
{\scriptsize\begin{align*}
    r_1 &= \GateAt{X}{0} \circ \GateAt{X}{1} \circ \CCGate{X}{0}{1} \circ \GateAt{X}{1} \circ \GateAt{X}{0}
    &
    r_2 &= \GateAt{X}{0}  \circ \CGate{X}{2}{1} \circ \CCGate{X}{0}{1} \circ \CGate{X}{2}{1} \circ \GateAt{X}{0} 
    \\
    r_3 &= \GateAt{X}{0} \circ \CCGate{X}{0}{1} \circ \GateAt{X}{0}
    &
    r_4 &= \CGate{X}{0}{1} \circ \CGate{X}{0}{2} \circ \CCGate{X}{1}{2} \circ \CGate{X}{0}{2} \circ \CGate{X}{0}{1}
    \\
    r_5 &= \GateAt{X}{1} \circ \CCGate{X}{0}{1} \circ \GateAt{X}{1}
    &
    r_6 &= \CGate{X}{2}{1} \circ \CCGate{X}{0}{1} \circ \CGate{X}{2}{1}
    \\
    r_7 &= \CGate{Z}{0}{1} \circ \CGate{X}{2}{1} \circ \CCGate{X}{0}{1} \circ \CGate{X}{2}{1} \circ \CGate{Z}{0}{1}
    &
    r_8 &= K_{1,2} \circ \GateAt{X}{1} \circ \GateAt{X}{2} \circ \CGate{Z}{0}{2} \circ \CCGate{X}{1}{2} \circ \CGate{Z}{0}{2} \circ \GateAt{X}{2} \circ \GateAt{X}{1} \circ K_{1,2}
\end{align*}}%
These equations can be derived from the geometry of $\mathbb{R}^8$.
First, note that $\CCGate{X}{0}{1}$ is a reflection about the hyperplane normal to $\ket{\widehat{b}} = \ket{1} \otimes \ket{1} \otimes \ket{-}$ where $\ket{\pm} = (\ket{0} \pm \ket{1}) / \sqrt{2}$.
Then for each generator $r_j$ with normal vector $\ket{b_j}$, it suffices to find an element $M \in W(E_8)$ such that $M \ket{\widehat{b}} = \ket{b_j}$.
The corresponding circuit would be $M \circ \CCGate{X}{1}{2} \circ M^{-1}$.
As an example of this technique, consider the Coxeter generator $r_3$ defined by the normal vector $\ket{b_3} = \ket{0} \otimes \ket{1} \otimes \ket{-}$.
Since $( \GateAt{X}{0} ) \ket{\widehat{b}} = \ket{b_3}$ with $\GateAt{X}{0}$ self-inverse, then $r_3 = \GateAt{X}{0} \circ \CCGate{X}{0}{1} \circ \GateAt{X}{0}$.
The remaining cases are established in \cref{Appendix:E8Decomp}.

Next, the Toffoli-K gates are introduced.
For simplicity of presentation, we first introduce the swap matrices $\sigma_{j,k} = \CGate{X}{j}{k} \circ \CGate{X}{k}{j} \circ \CGate{X}{j}{k}$ where $\sigma_{j,k}$ permutes the $j$-th qubit with the $k$-th qubit.
Recall that the Toffoli-K gates correspond to the following matrices\footnote{For simplicity, we assume that all $K$ gates are applied to adjacent qubits. This is sufficient, since $K_{0,2} = K_{0,1} \circ K_{1,2}$.}:
{\small\begin{equation*}
    \Sigma_D := \left\{ \GateAt{X}{j}, \CGate{X}{k}{l}, \CCGate{X}{j}{k}, \GateAt{Z}{j}, \CGate{Z}{j}{k}, K_{j,j+1}, \sigma_{j,k} \mid j,k,l \in \{ 0, 1, 2 \}, j < k, j \ne l \ne k \right\}.
\end{equation*}}%
It turns out that all Toffoli-K gates are generated by $\GateAt{X}{0}$, $\CGate{X}{1}{0}$, $\CCGate{X}{1}{2}$, and $K_{1,2}$.
To derive these primitive gates, it helps to first derive several diagonal matrices over $(\pm 1)$.
These are then used to derive the $\CCGate{X}{0}{1}$ and $\GateAt{X}{0}$ gates.
From this, the swap matrices can be derived, after which, it is relatively straightforward to construct the $\GateAt{X}{0}$, $\GateAt{K}{1,2}$, and $\CGate{X}{0}{1}$ gates.
This yields four words $w_X$, $w_K$, $w_{CX}$, and $w_{CCX}$, such that $\interp{w_X}^{*}_{E8} = \GateAt{X}{0}$, $\interp{w_{CX}}^{*}_{E8} = \CGate{X}{1}{0}$, $\interp{w_{CCX}}^{*}_{E8} = \CCGate{X}{1}{2}$, and $\interp{w_K}^{*}_{E8} = \GateAt{K}{1,2}$, as outlined in \cref{Appendix:E8Construct}. 

The remaining Toffoli-K gates are derived in terms of $K_{1,2}$, $\CCGate{X}{1}{2}$, $\GateAt{X}{0}$, and $\CGate{X}{0}{1}$.
To simplify this process, we note that once a gate has been derived, it may then be used to derive other gates.
This is analogous to how the generator $\GateAt{X}{0}$ appears in the defining relation for $\CGate{X}{0}{1}$.
Given a set of defining relations, if the dependencies between the generators defined by the relations form an acyclic digraph, then the defining relations arise from a valid sequence of \textbf{Gen($+$)} transformations (see~\cref{Appendix:DerivedGens}).
Likewise, the derived generators can be eliminated by a valid sequence of Tietze transformation.
The defining relations for the remaining $19$ gates are found in \cref{Appendix:DerivedGens}.
Since the dependencies among these generators are acyclic, then they must arise from $19$ valid applications of the \textbf{Gen($+$)} rules.
Let $R_{D(E8)}$ denote all $23$ relations.
Then $W( E_8 ) \cong \langle \Sigma_{E8} \cup \Sigma_D \mid R_{E8} \cup R_{D(E8)} \rangle$.

\subsection{Deriving the \texorpdfstring{$W(E_8)$}{W(E8)} Coxeter Generators}

Recall the circuit definitions for the $W(E_8)$ generators from \cref{Sec:E8:DIntro}.
Let $R_{E8(D)}$ denote the set of corresponding relations.
For example, the relation corresponding to $r_3$ is $r_3 \approx \GateAt{X}{0} \cdot \CCGate{X}{0}{1} \cdot \GateAt{X}{0}$.
Since these relations hold by definition, then they may be introduced via $8$ applications of \textbf{Rel($+$)} and consequently $W( E_8 ) \cong \langle \Sigma_{E8} \cup \Sigma_D \mid R_{E8} \cup R_{D(E8)} \cup R_{E8(D)} \rangle$.

\subsection{Elimination of the Coxeter Generators}
\label{Sec:E8:DyadRel}

\begin{figure}[t]
  \centering
  \begin{subfigure}[t]{0.49\textwidth}
    \centering
    \scriptsize
    \begin{align*}
      \mathbf{r_1} \cdot r_1
        &\; \rightarrow \; \GateAt{X}{0} \cdot \GateAt{X}{1} \cdot \CCGate{X}{0}{1} \cdot \GateAt{X}{1} \cdot \GateAt{X}{0} \cdot \mathbf{r_1} \\
        &\; \rightarrow \; \GateAt{X}{0} \cdot \GateAt{X}{1} \cdot \CCGate{X}{0}{1} \cdot \GateAt{X}{1} \cdot \mathbf{\GateAt{X}{0} \cdot \GateAt{X}{0}} \cdot \GateAt{X}{1} \cdot \CCGate{X}{0}{1} \cdot \GateAt{X}{1} \cdot \GateAt{X}{0} \\
        &\; \rightarrow \; \GateAt{X}{0} \cdot \GateAt{X}{1} \cdot \CCGate{X}{0}{1} \cdot \mathbf{\GateAt{X}{1} \cdot \GateAt{X}{1}} \cdot \CCGate{X}{0}{1} \cdot \GateAt{X}{1} \cdot \GateAt{X}{0} \\
        &\; \rightarrow \; \GateAt{X}{0} \cdot \GateAt{X}{1} \cdot \mathbf{\CCGate{X}{0}{1} \cdot \CCGate{X}{0}{1}} \cdot \GateAt{X}{1} \cdot \GateAt{X}{0} \\
        &\; \rightarrow \; \GateAt{X}{0} \cdot \mathbf{\GateAt{X}{1} \cdot \GateAt{X}{1}} \cdot \GateAt{X}{0} \\
        &\; \rightarrow \; \mathbf{\GateAt{X}{0} \cdot \GateAt{X}{0}} \\
        &\; \rightarrow \; \epsilon
    \end{align*}
    \subcaption{Deriving $\epsilon$ from $r_1 \cdot r_1$.}
    \label{Fig:Derivs:R1R1}
  \end{subfigure}
  \hfill
  \begin{subfigure}[t]{0.49\textwidth}
    \centering
    \scriptsize
    \begin{align*}
      \mathbf{\CGate{X}{0}{1}} \cdot \GateAt{X}{1}
        &\; \rightarrow \; \sigma_{0,1} \cdot \CGate{X}{1}{0} \cdot \sigma_{0,1} \cdot \mathbf{\GateAt{X}{1}} \\
        &\; \rightarrow \; \sigma_{0,1} \cdot \CGate{X}{1}{0} \cdot \mathbf{\sigma_{0,1} \cdot \sigma_{0,1}} \cdot \GateAt{X}{0} \sigma_{0,1} \\
        &\; \rightarrow \; \sigma_{0,1} \cdot \mathbf{\CGate{X}{1}{0} \cdot \GateAt{X}{0}} \cdot \sigma_{0,1} \\
        &\; \rightarrow \; \sigma_{0,1} \cdot \GateAt{X}{0} \mathbf{\cdot} \CGate{X}{1}{0} \sigma_{0,1} \\
        &\; \rightarrow \; \mathbf{\sigma_{0,1} \cdot \GateAt{X}{0} \cdot \sigma_{0,1}} \cdot \sigma_{0,1} \cdot \CGate{X}{1}{0} \cdot \sigma_{0,1} \\
        &\; \rightarrow \; \GateAt{X}{1} \cdot \mathbf{\sigma_{0,1} \cdot \CGate{X}{1}{0} \cdot \sigma_{0,1}} \\
        &\; \rightarrow \; \GateAt{X}{1} \cdot \CGate{X}{0}{1}
    \end{align*}
    \subcaption{Deriving $\GateAt{X}{1} \cdot \CGate{X}{0}{1}$ from $\CGate{X}{0}{1} \cdot \GateAt{X}{1}$.}
    \label{Fig:Derivs:Comm}
  \end{subfigure}
  \caption{Examples of derivations proofs which appear in the proof that $W( E_8 ) \cong \langle \Sigma_D \mid R_0 \rangle$.}
  \label{Fig:Derivs}
\end{figure}

In this section, the relations in $R_{E8}$ and $R_{D(E8)}$ are eliminated.
To do this, some additional $(\Sigma_D)$-relations are required.
For the remainder of this section, let $M_{(x_0,x_1,\ldots,x_k)}$ denote a gate $M$ applied to the qubits $x_0$ through to $x_k$.
For example, if $M$ is a doubly-controlled $X$ gate, then $M_{(1,2,0)}$ would correspond to $\CCGate{X}{1}{2}$.
Using this notation, we introduce the following families of relations, denoted $R_{D}$.
\begin{itemize}
\item[--] \textbf{Bifunctoriality}.
      $M_S \cdot N_T = N_T \cdot M_S$, for all $M_S, N_T \in \Sigma_{D}$ with $S \cap T = \varnothing$.
\item[--] \textbf{Symmetry}.
      $\sigma_{i,j} \cdot M_S \cdot \sigma_{i,j} = M_{\sigma_{i,j}( S )}$, for all $M_S \in \Sigma_D$ and integers $0 \le i < j \le 3$.
\item[--] \textbf{Order}.
      $M \cdot M = \epsilon$, for all $M \in \Sigma_D$.
\item[--] \textbf{Commutators}.
      $M_S \cdot N_T = N_T \cdot w$, for all $M_S, N_T \in \Sigma_{D}$ with $S \cap T \ne \varnothing$ and $w  \in (\Sigma_D)^{*}$ minimal.
\end{itemize}
The relation $\sigma_{1,2} \,\approx\, \CGate{Z}{1}{2} \cdot K_{1,2} \cdot \CGate{Z}{1}{2} \cdot K_{1,2} \cdot \CGate{Z}{1}{2} \cdot K_{1,2}$ from~\cite{BianSelinger2022} is also included for simplicity.
From these relations, all elements of $R_{D(E8)}$ and $R_{E8}$ can be derived.
Since all elements of $\Sigma_{E8}$ are self-inverse, then it suffices to consider only the upper half of the Coxeter matrix for $W( E_8 )$.
As an example, consider the relation $r_1 \cdot r_1 \approx \epsilon$ in $R_{E8}$.
The proof proceeds as in \cref{Fig:Derivs:R1R1}.
Then $\epsilon$ can be derived from $r_1 \cdot r_1$ using the relations in $R_D$.
Similar methods can be used to eliminate the remaining $(R_{E8})$-relations.
All derivations, for both $R_{D(E8)}$ and $R_{E8}$ can be found in the supplement to this paper~\cite{supp}.

\begin{theorem}
    \label{Thm:G0_Is_E8}
    $W( E_8 ) \cong \langle \Sigma_D \mid R_D \rangle$
\end{theorem}

It follows immediately from \cref{Thm:G0_Is_E8} that given any circuit $C$ over $\Sigma_D$, there exists a minimal word $w$ over the alphabet $\{ r_1, r_2, \ldots, r_8 \}$ such $\interp{C}_D^{*} = \interp{w}_{E8}^{*}$.
Then by the decompositions of \cref{Sec:E8:DIntro}, there exists a circuit $C'$ over $\Sigma_D$ such that $\interp{w}_{E8}^{*} = \interp{C'}_{D}^*$ such that $C'$ contains exactly $|w|$ Toffoli gates.
By \cite[Thm.~1.8]{Humphreys1990}, every minimal word in $W( E_8 )$ has length at most $n$, where $n$ is the cardinality of the positive root system associated with $W( E_8 )$.
By \cite{Humphreys1990}, the positive root system associated with $W( E_8 )$ has cardinality $120$.
Therefore, $C'$ contains at most $120$ Toffoli gates.
This provides an upper-bound on the Toffoli count for circuits over $\Sigma_D$, which can be thought of as a measure of computational complexity for these three-qubit circuits.

\begin{corollary}
    If $C \in \Sigma_D^*$, then there exists $C' \in \Sigma_D$ with Toffoli count at most $120$ such that $\interp{C}_D^* = \interp{C'}_D^*$.
\end{corollary}

\subsection{A Reduced Set of Relations for \texorpdfstring{$W(E_8)$}{W(E8)}}
\label{Sec:E8:RedRels}

\begin{figure}[t]
  {\scriptsize\begin{minipage}{0.47\textwidth}
    \begin{align}
    \CGate{Z}{0}{1} &\,\approx \, K_{1,2} \cdot \CGate{X}{0}{1} \cdot K_{1,2} \label{Rel:DGen:Start} \\
    \GateAt{X}{1} &\,\approx \, \CGate{X}{0}{1} \cdot \GateAt{X}{0} \cdot \CGate{X}{0}{1} \cdot \GateAt{X}{0} \\
    \GateAt{Z}{0} &\,\approx \, \CGate{Z}{0}{1} \cdot \CGate{X}{0}{1} \cdot \CGate{Z}{0}{1} \cdot \CGate{X}{0}{1} \\
    \GateAt{Z}{1} &\,\approx \, K_{1,2} \cdot \GateAt{X}{1} \cdot K_{1,2} \\
    \CGate{X}{2}{0} &\,\approx \, \GateAt{X}{1} \cdot \CCGate{X}{1}{2} \cdot \GateAt{X}{1} \cdot \CCGate{X}{1}{2} \\
    \CGate{X}{2}{1} &\,\approx \, \CGate{X}{2}{0} \cdot \CGate{X}{0}{1} \cdot \CGate{X}{2}{0} \cdot \CGate{X}{0}{1} \\
    \CGate{X}{1}{2} &\,\approx \, K_{1,2} \cdot \CGate{X}{2}{1} \cdot K_{1,2} \\
    \sigma_{1,2} &\,\approx \, \CGate{X}{1}{2} \cdot \CGate{X}{2}{1} \cdot \CGate{X}{1}{2} \\
    \CGate{X}{0}{2} &\,\approx \, \sigma_{1,2} \cdot \CGate{X}{0}{1} \cdot \sigma_{1,2} \\
    \sigma_{0,2} &\,\approx \, \CGate{X}{0}{2} \cdot \CGate{X}{2}{0} \cdot \CGate{X}{0}{2} \\
    K_{0,1} &\,\approx \, \sigma_{0,2} \cdot K_{1,2} \cdot \sigma_{0,2} \\
    \CGate{X}{1}{0} &\,\approx \, K_{0,1} \cdot \CGate{X}{0}{1} \cdot K_{0,1} \\
    \sigma_{0,1} &\,\approx \, \CGate{X}{0}{1} \cdot \CGate{X}{1}{0} \cdot \CGate{X}{0}{1} \\
    \CCGate{X}{0}{2} &\,\approx \, \sigma_{0,1} \cdot \CCGate{X}{1}{2} \cdot \sigma_{0,1} \\
    \GateAt{X}{2} &\,\approx \, \sigma_{0,2} \cdot \GateAt{X}{0} \cdot \sigma_{0,2} \\
    \GateAt{Z}{2} &\,\approx \, \sigma_{0,2} \cdot \GateAt{Z}{0} \cdot \sigma_{0,2} \\
    \CCGate{X}{0}{1} &\,\approx \, K_{1,2} \cdot \CCGate{X}{0}{2} \cdot K_{1,2} \label{Rel:DGen:DefCXX1} \\
    \CGate{Z}{0}{2} &\,\approx \, \sigma_{1,2} \cdot \CGate{Z}{0}{1} \cdot \sigma_{1,2} \\
    \CGate{Z}{1}{2} &\,\approx \, \sigma_{0,1} \cdot \CGate{Z}{0}{2} \cdot \sigma_{0,1}
    \label{Rel:DGen:End} \\
    \sigma_{1,2} &\,\approx \, \CGate{Z}{1}{2} \cdot K_{1,2} \cdot \CGate{Z}{1}{2} \cdot K_{1,2} \cdot \CGate{Z}{1}{2} \cdot K_{1,2} \\
    \GateAt{X}{0} \cdot \GateAt{X}{0} &\,\approx \, \epsilon \\
    \CGate{X}{0}{1} \cdot \CGate{X}{0}{1} &\,\approx \, \epsilon \\
    K_{1,2} \cdot K_{1,2} &\,\approx \, \epsilon \label{Rel:DGen:ElimK12}
    \end{align}
  \end{minipage}
  \begin{minipage}{0.51\textwidth}
    \begin{align}
    \CCGate{X}{1}{2} \cdot \CCGate{X}{1}{2} &\,\approx \, \epsilon \\
    K_{0,1} \cdot K_{0,1} &\,\approx \, \epsilon \\
    \CGate{X}{1}{2} \cdot \GateAt{X}{0}  &\,\approx \, \GateAt{X}{0} \cdot \CGate{X}{1}{2} \\
    \GateAt{X}{0} \cdot K_{1,2}  &\,\approx \, K_{1,2} \cdot \GateAt{X}{0} \\
    \GateAt{X}{1}  &\,\approx \, \sigma_{0,1} \cdot \GateAt{X}{0} \cdot \sigma_{0,1} \\
    \CGate{X}{2}{0}  &\,\approx \, \sigma_{0,2} \cdot \CGate{X}{0}{2} \cdot \sigma_{0,2} \\
    \CGate{X}{1}{2}  &\,\approx \, \sigma_{0,1} \cdot \CGate{X}{0}{2} \cdot \sigma_{0,1} \\
    \CGate{X}{2}{1}  &\,\approx \, \sigma_{0,1} \cdot \CGate{X}{2}{0} \cdot \sigma_{0,1} \\
    \CCGate{X}{0}{1} &\,\approx \, \sigma_{0,2} \cdot \CCGate{X}{1}{2} \cdot \sigma_{0,2} \\
    \CCGate{X}{0}{1} &\,\approx \, \sigma_{1,2} \cdot \CCGate{X}{0}{2} \cdot \sigma_{0,2} \\
    \GateAt{Z}{1} &\,\approx \, \sigma_{0,1} \cdot \GateAt{Z}{0} \cdot \sigma_{0,1} \\
    K_{0,1}  &\,\approx \, \sigma_{0,1} \cdot K_{0,1} \cdot \sigma_{0,1} \\
    \CCGate{X}{1}{2} \cdot \CGate{X}{1}{0} &\,\approx \, \CGate{X}{1}{0} \cdot \CCGate{X}{1}{2} \\
    \GateAt{X}{0} \cdot \CCGate{X}{1}{2}  &\,\approx \, \CCGate{X}{1}{2} \cdot \GateAt{X}{0} \\
    \GateAt{X}{0} \cdot \CGate{X}{1}{0}  &\,\approx \, \CGate{X}{1}{0} \cdot \GateAt{X}{0} \\
    K_{0,1} \cdot K_{1,2}  &\,\approx \, K_{1,2} \cdot K_{0,1} \\
    \CGate{Z}{0}{1} \cdot \CGate{Z}{1}{2}  &\,\approx \, \CGate{Z}{1}{2} \cdot \CGate{Z}{0}{1} \\
    K_{0,1} \cdot \GateAt{Z}{0}  &\,\approx \, \GateAt{X}{0} \cdot K_{0,1} \\
    \GateAt{X}{0} \cdot \CCGate{X}{0}{1}  &\,\approx \, \CCGate{X}{0}{1} \cdot \CGate{X}{1}{2} \cdot \GateAt{X}{0} \\
    \CGate{X}{0}{1} \cdot \CGate{Z}{1}{2}  &\,\approx \, \CGate{Z}{1}{2} \cdot \CGate{Z}{0}{2} \cdot \CGate{X}{0}{1} \\
    \CGate{X}{1}{2} \cdot \CCGate{X}{1}{2}  &\,\approx \, \CCGate{X}{1}{2} \cdot \CGate{X}{1}{0} \cdot \CGate{X}{1}{2} \\
    \CCGate{X}{1}{2} \cdot \CGate{X}{0}{1}  &\,\approx \, \CGate{X}{0}{1} \cdot \CCGate{X}{0}{2} \cdot \CCGate{X}{1}{2} \cdot \CCGate{X}{0}{2} \\
    \CCGate{X}{0}{1} \cdot \CCGate{X}{0}{2} &\,\approx \, \CCGate{X}{0}{2} \cdot \CCGate{X}{0}{1} \cdot \CCGate{X}{0}{2} \cdot \CCGate{X}{0}{1}
    \end{align}
  \end{minipage}}
  \caption{Relations for $W(E_8)$, denoted $R_0$.\label{Fig:G0Rels}}
  \hfill\rule{0.95\textwidth}{0.1mm}\hfill\hbox{}
\end{figure}

The relations $R_D$ from \cref{Sec:E8:DyadRel} are far from minimal.
For example, the family of commutator relations contains all relations of the form $\CGate{X}{j}{k} \cdot \GateAt{X}{k} \approx \GateAt{X}{k} \cdot \CGate{X}{j}{k}$.
However, given all symmetry relations, it suffices to include only $\CGate{X}{1}{0} \cdot \GateAt{X}{0} \approx \GateAt{X}{0} \cdot \CGate{X}{1}{0}$.
The remaining commutator relations can be derived, as illustrated in \cref{Fig:Derivs:Comm}.
Furthermore, many of the relations in $R_D$ do not appear in any derivations of the supplement.
For example, the relation $\CGate{X}{2}{0} \cdot \GateAt{X}{0} \approx \GateAt{X}{0} \cdot \CGate{X}{2}{0}$ does not appear, and therefore \cref{Thm:G0_Is_E8} holds with respect to the relation set $R_{D} \setminus \{ \CGate{X}{2}{0} \cdot \GateAt{X}{0} \approx \GateAt{X}{0} \cdot \CGate{X}{2}{0} \}$.
Using both techniques, a new relation set $R_0$ is obtained, as illustrated in \cref{Fig:G0Rels}.
All derivations can be found in the supplement to this paper~\cite{supp}.

\begin{corollary}
    $W( E_8 ) \cong \langle \Sigma_D \mid R_0 \rangle$
\end{corollary}

\subsection{A Minimal Generating Set for \texorpdfstring{$W(E_8)$}{W(E8)}}
\label{Sec:E8:MinGens}

Define $\Sigma_0 = \left\{ \GateAt{X}{0}, \CGate{X}{0}{1}, \CCGate{X}{1}{2}, K_{1,2} \right\}$.
From \cref{Sec:E8:DIntro}, it is clear that $\Sigma_0$ generates $W(E_8)$.
In fact, $\Sigma_0$ is minimal in the sense that every proper subset of $\Sigma_0$ generates a proper subgroup of $W(E_8)$.
In other words, no proper subset of $\Sigma_0$ generates $W( E_8 )$.
To show that $\Sigma_0$ is a minimal generating set for $W( E_8 )$, it suffices to show that for every maximal proper subset $\Sigma'$ of $\Sigma_0$, there exists some $8 \times 8$ dyadic matrix $M$ such that $M$ commutes with the elements of $\Sigma'$ but does not commute with the elements of $\Sigma_0$.
Intuitively, the subgroup generated by $\Sigma'$ commutes with $M$, whereas the subgroup generated by $\Sigma_0$ does not commute with $M$.
This claim is proven in \cref{Appendix:E8MinGens}, and the matrices are constructed.

\begin{theorem}
    \label{Thm:MinGens}
    $\Sigma_0$ is a minimal generating set for $W( E_8 )$.
\end{theorem}

\newcommand{\TLOGen}[2]{\mathcal{G}_{#1}^{#2}}
\newcommand{\TLORel}[2]{\mathcal{R}_{#1}^{#2}}
\newcommand{\TLO}[2]{{#1}{}_{[#2]}}

\section{Extending to a Presentation of \texorpdfstring{$\mathrm{O}(8, \D)$}{O(8,D)}}
\label{Sec:O8D}

\begin{figure}[t]
  {\scriptsize\begin{minipage}[t]{0.33\textwidth}
    \begin{align}
    \TLO{X}{a,b}{}^2 &\,\approx \, \epsilon \label{Eq:O8DRels:Orig:Perm1} \\
    \TLO{(-1)}{a}{}^2 &\,\approx \, \epsilon \label{EQ:O8DRels:Orig:Rep1} \\
    \TLO{K}{a,b,c,d}{}^2 &\,\approx \, \epsilon \label{EQ:O8DRels:Orig:Rep2} \\
    \TLO{X}{a,b} \cdot \TLO{X}{c,d} &\,\approx \, \TLO{X}{c,d} \cdot \TLO{X}{a,b} \label{Eq:O8DRels:Orig:Perm2} \\
    \TLO{X}{a,b} \cdot \TLO{(-1)}{c} &\,\approx \, \TLO{(-1)}{c} \cdot \TLO{X}{a,b} \label{Eq:O8DRels:Orig:Perm3} \\
    \TLO{X}{a,b} \cdot \TLO{K}{c,d,e,f} &\,\approx \, \TLO{K}{c,d,e,f} \cdot \TLO{X}{a,b} \label{Eq:O8DRels:Orig:Perm4} \\
    \TLO{(-1)}{a} \cdot \TLO{K}{b,c,d,e} &\,\approx \, \TLO{K}{b,c,d,e} \cdot \TLO{(-1)}{a}  \label{EQ:O8DRels:Orig:Rep3} \\
    \TLO{(-1)}{a} \cdot \TLO{(-1)}{b} &\,\approx \, \TLO{(-1)}{b} \cdot \TLO{(-1)}{a} \label{Eq:O8DRels:Orig:ZCom1} \\
    \TLO{K}{a,b,c,d} \cdot \TLO{K}{e,f,g,h} &\,\approx \, \TLO{K}{e,f,g,h} \cdot \TLO{K}{a,b,c,d} \label{Eq:O8DRels:Orig:KCom1} \\
    \TLO{X}{a,c} \cdot \TLO{X}{a,b} &\,\approx \, \TLO{X}{c,b} \cdot \TLO{X}{a,c} \label{Eq:O8DRels:Orig:Perm5}
    \end{align}
  \end{minipage}
  \begin{minipage}[t]{0.65\textwidth}
    \begin{align}
    \TLO{X}{b,c} \cdot \TLO{X}{a,b} &\,\approx \, \TLO{X}{a,c} \cdot \TLO{X}{b,c} \label{Eq:O8DRels:Orig:Perm6} \\
    \TLO{X}{a,b} \cdot \TLO{(-1)}{a} &\,\approx \, \TLO{(-1)}{b} \cdot \TLO{X}{a,b} \label{Eq:O8DRels:Orig:Perm7} \\
    \TLO{X}{a,e} \cdot \TLO{K}{a,b,c,d} &\,\approx \, \TLO{K}{e,b,c,d} \cdot \TLO{X}{a,e}  \label{Eq:O8DRels:Orig:Perm8} \\
    \TLO{X}{b,e} \cdot \TLO{K}{a,b,c,d} &\,\approx \, \TLO{K}{a,e,c,d} \cdot \TLO{X}{b,e}  \label{Eq:O8DRels:Orig:Perm9} \\
    \TLO{X}{c,e} \cdot \TLO{K}{a,b,c,d} &\,\approx \, \TLO{K}{a,b,e,d} \cdot \TLO{X}{c,e}  \label{Eq:O8DRels:Orig:Perm10} \\
    \TLO{X}{d,e} \cdot \TLO{K}{a,b,c,d} &\,\approx \, \TLO{K}{a,b,c,e} \cdot \TLO{X}{d,e}  \label{Eq:O8DRels:Orig:Perm11} \\
    \TLO{X}{a,b} \cdot \TLO{K}{a,b,c,d} &\,\approx \, \TLO{K}{a,b,c,d} \cdot \TLO{X}{a,b} \cdot \TLO{(-1)}{b} \cdot \TLO{(-1)}{d}  \label{EQ:O8DRels:Orig:Rep5a} \\
    \TLO{X}{b,c} \cdot \TLO{K}{a,b,c,d} &\,\approx \, \TLO{(-1)}{a} \cdot \TLO{K}{a,b,c,d} \cdot \TLO{(-1)}{a} \cdot \TLO{K}{a,b,c,d} \cdot \TLO{(-1)}{a}  \label{EQ:O8DRels:Orig:Rep6} \\
    \TLO{X}{c,d} \cdot \TLO{K}{a,b,c,d} &\,\approx \, \TLO{K}{a,b,c,d} \cdot \TLO{X}{b,d}  \label{EQ:O8DRels:Orig:Rep7} \\
    \TLO{K}{a,b,c,d} \cdot \TLO{K}{b,d,e,f} &\,\approx \, \TLO{K}{b,d,e,f} \cdot \TLO{K}{a,b,c,d}  \label{EQ:O8DRels:Orig:Rep8} \\
    \TLO{(-1)}{a} \cdot \TLO{(-1)}{e} \cdot \TLO{X}{a,e} \cdot \rho_{a,b,c,d,e,f,g,h} &\,\approx\, \rho_{a,b,c,d,e,f,g,h} \cdot \TLO{X}{a,e} \cdot \TLO{(-1)}{e} \cdot \TLO{(-1)}{a} \label{EQ:O8DRels:Orig:Rep9}
    \end{align}
  \end{minipage}}
  \caption{The relations in $\TLORel{n}{}$ from~\cite{LiRoss2021}, for all valid choices of $a, b, c, d, e, f, g, h \in \mathbb{Z}$. We write $\rho_{a,b,c,d,e,f,g,h}$ for the substring $\TLO{K}{e,f,g,h} \cdot \TLO{K}{a,b,c,d} \cdot \TLO{X}{d,e} \cdot \TLO{K}{a,b,c,d} \cdot \TLO{K}{e,f,g,h}$.}
  \label{Fig:O8DRels:Orig}
  \hfill\rule{0.95\textwidth}{0.1mm}\hfill\hbox{}
\end{figure}

Li \textit{et al.}~\cite{LiRoss2021} introduced a presentation for $\mathrm{O}(8, \D)$ using $m$-level operators.
Let $n > 0$, $I$ be the $n \times n$ identity matrix, and $[m] = \{ 0, 1, \ldots, m - 1 \}$.
Then given an $m \times m$ matrix $M$ with $m < n$, and a strictly increasing sequence $( a_0, \ldots, a_{m-1} )$ over $[m]$, define $\TLO{M}{a_0,\ldots,a_{m-1}}$ to be the $n \times n$ matrix such that:
\begin{enumerate}
\item For each pair of elements $( j, k )$ over $[m]$, the component $( a_j, a_k )$ of $\TLO{M}{a_0,\ldots,a_{m-1}}$ is equal to the component $( a_j, a_k )$ of $M$;
\item For each pair of elements $( j, k )$ over $[n] \setminus \{ a_0, a_1, \ldots, a_m \}$, the component $( j, k )$ of $\TLO{M}{a_0,a_1,\ldots,a_{m-1}}$ is equal to the component $( a_j, a_k )$ of $I$.
\end{enumerate}
We say that $\TLO{M}{a_0,\ldots,a_{m-1}}$ is an \emph{$m$-level operator of type $M$}.
When $n=8$ for example, $\CCGate{X}{0}{1} = \TLO{X}{6,7}$, $\CCGate{Z}{0}{1} = \TLO{(-1)}{7}$, and $\TLO{K}{4,5,6,7}$ is a controlled $K$-gate.
Define the following for $n > 3$.
\begin{equation*}
    \TLOGen{n}{} = \left\{
        \TLO{(-1)}{a},
        \TLO{X}{a,b},
        \TLO{K}{a,b,c,d}
        \mid
        a, b, c, d \in \Z
        \text{ and }
        0 \le a < b < c < d < n
    \right\}
\end{equation*}
It was shown in \cite{LiRoss2021} that $\mathrm{O}( n, \D ) \cong \langle \TLOGen{n}{} \mid \TLORel{n}{} \rangle$, where $\TLORel{n}{}$ is given in \cref{Fig:O8DRels:Orig}.
The goal of this section is to construct a sequence of Tietze transformations, starting from $\langle \TLOGen{8}{} \mid \TLORel{8}{} \rangle$, such that the generators and relations describing the subgroup $W( E_8 )$ are replaced by $\Sigma_D$ and $R_0$, respectively.
This process follows similarly to \cref{Sec:E8}.
However, one should note that $|\TLORel{8}{}| = 2113$ (see~\cref{Appendix:RelCount}).
Inspection of $\TLORel{n}{}$ reveals that many of these relations are either definitional, or obtained through permutations of indices.
For this reason, $\TLORel{n}{}$ is partially reduced before carrying out the aforementioned Tietze transformations.
First, the permutations are eliminated via a sequence of \textbf{Rel($-$)} transformations to obtain $\TLORel{n}{1}$.
Next, some redundant commutator relations are eliminated via a sequence of \textbf{Rel($-$)} transformations to obtain $\TLORel{n}{2}$.
Finally, the derived generators are eliminated to obtain $\TLORel{n}{3}$.
All proofs can be found in \cref{Appendix:O8D}.

\subsection{Permutation Groups and Reindexing}
\label{Sec:O8D:Reindex}

Let $[n] = \{ 0, 1, \ldots, n - 1 \}$ and $\mathrm{S}(n)$ denote the group of permutations on $[n]$.
For $j, k \in [n]$, let $\tau_{j,k}$ denote the permutation that swaps $j$ and $k$.
For example, $\tau_{0,1}(0) = 1$, $\tau_{0,1}(1) = 0$, and $\tau_{0,1}(2) = 2$.
The group $\mathrm{S}(n)$ is a finite reflection group generated by the \emph{transpositions} $\{ \tau_{j,j+1} \mid j \in [n] \}$~(see~\cite{Humphreys1990}).
The \emph{braiding relations}, which state that $\tau_{j,j+1} \circ \tau_{j+1,j+2} \circ \tau_{j,j+1} = \tau_{j+1,j+2} \circ \tau_{j,j+1} \circ \tau_{j+1,j+2}$ for all $j \in [n-2]$, together with the order relations are sound and complete for $\mathrm{S}(n)$~(see~\cite{Johnson1990}).
The standard representation of $\mathrm{S}(n)$ as a reflection group sends each $\tau_{j,k}$ to $X_{[j,k]}$.
This means that every two-level operator of type $X$ can be decomposed into sequence of transpositions.
Intuitively, each $X_{[j,k]}$ acts by permuting the standard basis vectors $\ket{j}$ and $\ket{k}$, which can be achieved through a sequence of transpositions of basis vectors.
Clearly, $\mathrm{S}(8) \hookrightarrow W(E_8) \le O(8, \D)$.

Many relations in $\TLORel{n}{}$ are related via permutation of indices.
The \emph{formal application of $\sigma$} to a word over $\TLOGen{n}{}$ is defined inductively as follows.
{\scriptsize\begin{align*}
    \sigma(\epsilon) &= \epsilon
    &
    \sigma(\TLO{X}{a,b} \cdot w) &= \TLO{X}{\sigma(a), \sigma(b)} \cdot \sigma(w)
    \\
    \sigma(\TLO{(-1)}{a} \cdot w) &= \TLO{(-1)}{\sigma(a)} \cdot \sigma(w)
    &
    \sigma(\TLO{K}{a,b,c,d} \cdot w) &= \TLO{K}{\sigma(a),\sigma(b),\sigma(c),\sigma(d)} \cdot \sigma(w)
\end{align*}}%
Note that $\sigma( w )$ may yield $m$-level operators with invalid indices.
For example, $\tau_{1,2}( \TLO{X}{1,2} )$ yields $\TLO{X}{2,1}$, which is not a valid two-level operator since $2 > 1$.
The permutation $\sigma$ is a \emph{valid reindexing for $w$} if all symbols in $\sigma(w)$ are well-formed multi-level operators.
If $\sigma$ is valid for $v$ and $w$, then $\sigma$ is valid for $v \cdot w$.
Conversely, if $\sigma$ is valid for $w$, then $\sigma$ is valid for all subwords in $w$.
Consider, for example, the word $w = \TLO{K}{2,3,4,5} \cdot \TLO{K}{3,5,6,7}$ which appears on the left-hand side of an instance of \cref{Eq:O8DRels:Orig:KCom1}.
Let $\sigma \in \mathrm{S}(8)$ be the cyclic permutation $7 \mapsto 5 \mapsto 3 \mapsto 1 \mapsto 6 \mapsto 4 \mapsto 2 \mapsto 0 \mapsto 7$.
Then $\sigma$ is a valid reindexing for $w$ since $\sigma( w ) = \TLO{K}{0,1,2,5} \cdot \TLO{K}{1,3,4,5}$.
In \cref{Appendix:O8D:Reindex}, we show that all valid reindexings are derivable using only the relations in $\TLORel{\sigma}{} = \{ \text{\cref{Eq:O8DRels:Orig:Perm1,Eq:O8DRels:Orig:Perm3,Eq:O8DRels:Orig:Perm4,Eq:O8DRels:Orig:Perm5,Eq:O8DRels:Orig:Perm6,Eq:O8DRels:Orig:Perm7,Eq:O8DRels:Orig:Perm8,Eq:O8DRels:Orig:Perm9,Eq:O8DRels:Orig:Perm10,Eq:O8DRels:Orig:Perm11}} \}$.

\subsection{Selecting Representative Relations for \texorpdfstring{$\mathrm{O}(n, \D)$}{O(n,D)}}
\label{Sec:O8D:RepRels}

As a consequence of \cref{Sec:O8D:Reindex}, many relations in $\TLORel{n}{}$ can be replaced by representative instances.
For example, let $r$ denote instance $\TLO{(-1)}{6} \cdot \TLO{(-1)}{7} \approx \TLO{(-1)}{7} \cdot \TLO{(-1)}{6}$ of \cref{Eq:O8DRels:Orig:ZCom1}.
Clearly $\sigma = \tau_{0,6} \circ \tau_{1,7}$ is a valid reindexing for $r$, where $\sigma(r)$ is $\TLO{(-1)}{0} \cdot \TLO{(-1)}{1} \approx \TLO{(-1)}{1} \cdot \TLO{(-1)}{0}$
Then by \cref{Appendix:O8D:Reindex}, it is possible to derive $\sigma(r)$ from $r$ using $\TLORel{n}{} \setminus \{ \sigma(r) \}$.
Then $\langle \TLOGen{n}{} \mid \TLORel{n}{} \rangle \cong \langle \TLOGen{n}{} \mid \TLORel{n}{} \setminus \{ \sigma(r) \} \rangle$ by \textbf{Rel($-$)}.

This process can be repeated, until all instances of \cref{Eq:O8DRels:Orig:ZCom1} have been eliminated, except for the representative relation $r$.
In a similar fashion, \cref{EQ:O8DRels:Orig:Rep1,EQ:O8DRels:Orig:Rep2,EQ:O8DRels:Orig:Rep3,Eq:O8DRels:Orig:KCom1,EQ:O8DRels:Orig:Rep5a,EQ:O8DRels:Orig:Rep6,EQ:O8DRels:Orig:Rep7,EQ:O8DRels:Orig:Rep8,EQ:O8DRels:Orig:Rep9} can be eliminated, since these relations do not appear in $\TLORel{\sigma}{}$.
Then $\mathrm{O}(n,\D) \cong \langle \TLOGen{n}{} \mid \TLORel{n}{1} \rangle$ where $\TLORel{n}{1}$ is the set of representative relations (see \cref{Appendix:O8D:Reps}).

\subsection{Selecting Representative Generators for \texorpdfstring{$\mathrm{O}(n, \D)$}{O(n,D)}}

Define the new generator set,
\begin{equation*}
    \TLOGen{n}{1} = \left\{
        \TLO{X}{a,b},
        \mid
        a, b, \in \Z
        \text{ and }
        0 \le a < b < n
    \right\} \cup \left\{
        \TLO{K}{0,1,2,3}
    \right\} \cup \left\{
        \TLO{(-1)}{0}
    \right\}.
\end{equation*}
Many of the generators in $\TLOGen{n}{}$ are redundant in the sense that they may be constructed using only the generators in $\TLOGen{n}{1}$.
This is because $\mathrm{S}(n) \hookrightarrow \TLOGen{n}{1}$, with $\TLOGen{n}{1} \setminus \TLOGen{n}{}$ consisting of valid indexings of either $\TLO{K}{0,1,2,3}$ or $\TLO{(-1)}{0}$.
Furthermore, these reindexings follow from relations in $\TLORel{n}{1}$.
As an example, consider the instance $\TLO{X}{0,7} \cdot \TLO{(-1)}{0} \approx \TLO{(-1)}{7} \cdot \TLO{X}{0,7}$ of \cref{Eq:O8DRels:Orig:Perm7}.
Using the order relation for $\TLO{X}{0,7}$, the following derivation holds.
{\scriptsize\begin{equation*}
    \TLO{(-1)}{7} \;\leftarrow\; \TLO{(-1)}{7} \cdot \TLO{X}{0,7}{}^2 \;\leftarrow\; \TLO{X}{0,7} \cdot \TLO{(-1)}{0} \cdot \TLO{X}{0,7}
\end{equation*}}%
Similarly, the original relation can be obtained from this new relation using the order relation for $\TLO{X}{0,7}$.
Then through a \textbf{Rel($+$)} transformation followed by a \textbf{Rel($-$)} transformation, the commutator relation $\TLO{X}{0,7} \cdot \TLO{(-1)}{0} \approx \TLO{(-1)}{7} \cdot \TLO{X}{0,7}$ can be replaced by the definitional relation $\TLO{(-1)}{7} = \TLO{X}{0,7} \cdot \TLO{(-1)}{0} \cdot \TLO{X}{0,7}$.
This process can be repeated for all instances of \cref{Eq:O8DRels:Orig:Perm7}.

To derive the four-level operators of type $K$, it suffices to note that the following family of relations are valid with respect to $\interp{\cdot}_O^*$.
{\scriptsize\begin{equation*}
    \TLO{K}{a,b,c,d} \approx \TLO{X}{0,a} \cdot \TLO{X}{1,b} \cdot \TLO{X}{2,c} \cdot \TLO{X}{3,d} \cdot \TLO{K}{0,1,2,3} \cdot \TLO{X}{3,d} \cdot \TLO{X}{2,c} \cdot \TLO{X}{1, b} \cdot \TLO{X}{0,a}
\end{equation*}}%
The cases where $\{ a, b, c, d \} \cap \{ 0, 1, 2, 3 \} \ne \varnothing$ can be handled using the techniques of \cref{Appendix:O8D:Reindex}.
These relations are introduced using a sequence of \textbf{Rel($+$)} relations to obtain a new relation set $R$.
In this relation set, all multi-level operators of type $(-1)$ and $K$ are defined in terms of $\TLO{(-1)}{0}$ and $\TLO{K}{0,1,2,3}$, respectively.
As outlined in \cref{Appendix:DerivedGens}, these defining relations can be used to eliminate all generators in $\TLOGen{n}{1} \setminus \TLORel{n}{0}$ via a finite sequence of Tietze transformations.

The elimination process works as follows.
Let $M \in \TLOGen{n}{1} \setminus \TLOGen{n}{}$.
Then $M$ appears in some defining relation $M \approx w$.
If $M$ appears in some relation $r \in R$, then every instance of $M$ will be replaced by $w$.
For example, \cref{Eq:O8DRels:Orig:ZCom1} will be replaced by the following relation.
{\scriptsize\begin{equation*}
    \TLO{(-1)}{0} \cdot \TLO{X}{0,5} \cdot \TLO{(-1)}{0} \cdot \TLO{X}{0,5} \approx \TLO{X}{0,5} \cdot \TLO{(-1)}{0} \cdot \TLO{X}{0,5} \cdot \TLO{(-1)}{0}
\end{equation*}}%
We introduce the following abbreviations for simplicity of presentation.
{\scriptsize\begin{align*}
    \TLO{(-1)}{c} &= \TLO{X}{0,c} \cdot \TLO{(-1)}{0} \cdot \TLO{X}{0,c}
    &
    \TLO{K}{0,1,2,d} &= \TLO{X}{3,d} \cdot \TLO{K}{0,1,2,3} \cdot \TLO{X}{3,d}
    &
    \TLO{K}{0,1,c,d} &= \TLO{X}{2,c} \cdot \TLO{K}{0,1,2,d} \cdot \TLO{X}{2,c}
    \\
    \TLO{K}{0,b,c,d} &= \TLO{X}{1,b} \cdot \TLO{K}{0,1,c,d} \cdot \TLO{X}{1,b}
    &
    \TLO{K}{a,b,c,d} &= \TLO{X}{0,a} \cdot \TLO{K}{0,b,c,d} \cdot \TLO{X}{0,a}
\end{align*}}%
Denote this new set of relations $\TLORel{n}{2}$.
Then $\mathrm{O}(n,\D) \cong \langle \TLOGen{n}{1} \mid \TLORel{n}{2} \rangle$.

\subsection{Eliminating Redundant Relations}

\begin{figure}[t]
  {\scriptsize\begin{minipage}[t]{0.35\textwidth}
    \begin{align}
    \TLO{X}{a,a+1}{}^2 &\,\approx \, \epsilon \\
    \TLO{(-1)}{0}{}^2 &\,\approx \, \epsilon \\
    \TLO{K}{0,1,2,3}{}^2 &\,\approx \, \epsilon \\
    \TLO{X}{b,b+1} \cdot \TLO{(-1)}{0} &\,\approx \, \TLO{(-1)}{0} \cdot \TLO{X}{b,b+1} \\
    \TLO{X}{c,c+1} \cdot \TLO{K}{0,1,2,3} &\,\approx \, \TLO{K}{0,1,2,3} \cdot \TLO{X}{c,c+1} \\
    \TLO{(-1)}{4} \cdot \TLO{K}{0,1,2,3} &\,\approx \, \TLO{K}{0,1,2,3} \cdot \TLO{(-1)}{4} \\
    \TLO{(-1)}{0} \cdot \TLO{(-1)}{4} &\,\approx \, \TLO{(-1)}{4} \cdot \TLO{(-1)}{0}
    \end{align}
  \end{minipage}
  \begin{minipage}[t]{0.63\textwidth}
    \begin{align}
    \TLO{K}{0,1,2,3} \cdot \TLO{K}{4,5,6,7} &\,\approx \, \TLO{K}{4,5,6,7} \cdot \TLO{K}{0,1,2,3} \\
    \TLO{X}{a,a+1} \cdot \TLO{X}{a,a+2} &\,\approx \, \TLO{X}{a+1,a+2} \cdot \TLO{X}{a,a+1} \\
    \TLO{X}{a+1,b} \cdot \TLO{X}{a,a+1} &\,\approx \, \TLO{X}{a,b} \cdot \TLO{X}{a+1,b} \\
    \TLO{X}{0,1} \cdot \TLO{K}{0,1,2,3} &\,\approx \, \TLO{K}{0,1,2,3} \cdot \TLO{X}{0,1} \cdot \TLO{(-1)}{1} \cdot \TLO{(-1)}{3} \\
    \TLO{X}{1,2} \cdot \TLO{K}{0,1,2,3} &\,\approx \, \TLO{(-1)}{0} \cdot \TLO{K}{0,1,2,3} \cdot \TLO{(-1)}{0} \cdot \TLO{K}{0,1,2,3} \cdot \TLO{(-1)}{0} \\
    \TLO{X}{2,3} \cdot \TLO{K}{0,1,2,3} &\,\approx \, \TLO{K}{0,1,2,3} \cdot \TLO{X}{1,3} \\
    \TLO{K}{0,1,2,3} \cdot \TLO{K}{1,3,4,5} &\,\approx \, \TLO{K}{1,3,4,5} \cdot \TLO{K}{0,1,2,3} \\
    \TLO{(-1)}{0} \cdot \TLO{(-1)}{4} \cdot \TLO{X}{0,4} \cdot \rho &\,\approx\, \rho \cdot \TLO{X}{0,4} \cdot \TLO{(-1)}{4} \cdot \TLO{(-1)}{0}
  \end{align}
  \end{minipage}}
  \caption{The reduced relations in $\TLORel{n}{3}$, for all valid choices of $a, b, c \in \mathbb{Z}$ where $b > 0$ and $c > 3$. We write $\rho$ for the substring $\TLO{K}{4,5,6,7} \cdot \TLO{K}{0,1,2,3} \cdot \TLO{X}{3,4} \cdot \TLO{K}{0,1,2,3} \cdot \TLO{K}{4,5,6,7}$.}
  \label{Fig:O8DRels:Small}
  \hfill\rule{0.95\textwidth}{0.1mm}\hfill\hbox{}
\end{figure}

It will now shown that many relations in $\TLORel{n}{2}$ are redundant.
First, the braiding relations and order relations are used according to \cref{Appendix:O8D:Reindex} to eliminate all other relations over the two-level operators of type $X$.
This reduced relation set is then used to show that all instances of \cref{Eq:O8DRels:Orig:Perm3,Eq:O8DRels:Orig:Perm4} can be derived using transpositions in place of swaps.
Finally, it is shown that the relations \cref{Eq:O8DRels:Orig:Perm8,Eq:O8DRels:Orig:Perm9,Eq:O8DRels:Orig:Perm10,Eq:O8DRels:Orig:Perm11} are entirely redundant.
All derivations can be found in \cref{Appendix:O8D:Reduce}.
This new set of relations is denoted $\TLORel{n}{3}$, and can be found in \cref{Fig:O8DRels:Small}.
Then via a sequence of \textbf{Rel($-$)} transformations, the following presentation is obtained.

\begin{theorem}
    $\mathrm{O}(n,\D) \cong \langle \TLOGen{n}{1} \mid \TLORel{n}{3} \rangle$.
\end{theorem}

\subsection{Introducing the \texorpdfstring{$W(E_8)$}{W(E8)} Generators}

It this section, the circuit generators and relations for $W(E_8)$ are introduced.
Since $\CCGate{X}{1}{2} = \TLO{X}{6,7}$, then without loss of generality, every instance of $\TLO{X}{6,7}$ in $\TLORel{8}{3}$ can be replaced by $\CCGate{X}{1}{2}$.
Next, the generators $\GateAt{X}{0}$ and $\CGate{X}{0}{1}$ are introduced.
This yields the following relations.
{\scriptsize\begin{align*}
    (r_X): \GateAt{X}{0} &\approx \TLO{X}{0,4} \cdot \TLO{X}{1,5} \cdot \TLO{X}{2,6} \cdot \TLO{X}{3,7}
    &
    (r_{CX}): \CGate{X}{0}{1} &\approx \TLO{X}{2,6} \cdot \TLO{X}{3,7}
\end{align*}}%
It turns out that the $\GateAt{K}{1,2}$ gate decomposes into a word over $\GateAt{X}{0}$ and $\TLO{K}{0,1,2,3}$.
This is because $\TLO{K}{0,1,2,3}$ is a $\GateAt{K}{1,2}$ gate which is applied when qubit $0$ is in state $\ket{1}$, and $\GateAt{X}{0} \circ \TLO{K}{0,1,2,3} \circ \GateAt{X}{0}$ is a $\GateAt{K}{1,2}$ gate which is applied when qubit $0$ is in state $\ket{0}$.
Together, these two words compose to a $\GateAt{K}{1,2}$ gate without any controls.
This yields the following relation.
{\scriptsize\begin{align*}
    (r_K): \GateAt{K}{1,2} &\approx \TLO{K}{4,5,6,7} \cdot \GateAt{X}{0} \cdot \TLO{K}{4,5,6,7} \cdot \GateAt{X}{0}
\end{align*}}%
The relations $r_X$, $r_{CX}$, and $r_K$ can be validated with respect to $\interp{\cdot}_O$.
These relations do not depend on one-another, so the generators in $\Sigma_0$ may be introduced via a sequence of \textbf{Gen($+$)} transformations, as outlined in \cref{Appendix:DerivedGens}.
Likewise, the derived generators in $\Sigma_D \setminus \Sigma_0$ may be introduced via a sequence of \textbf{Gen($+$)} transformations, as outlined in \cref{Appendix:DerivedGens}.
Finally, the relations in $R_0$ may be introduced, since \cref{Sec:E8:DIntro} established the validity of these relations in $W( E_8 )$, which is a subgroup of $O(8, \D)$.
This sequence of transformations yields $\mathrm{O}(8,\D) \cong \left\langle \TLOGen{8}{1} \cup \Sigma_D \mid \TLORel{8}{3} \cup R_0 \cup \{ r_X, r_{CX}, r_K \}  \right\rangle$.

\subsection{Eliminating the Multi-Level Operators}

Using the generators in $\Sigma_D$ and the relations in $R_0$, it is possible to eliminate all two-level operators of type $X$.
As a first step, it must be shown that the two-level operators can be decomposed into circuits over $\Sigma_D$.
This follows from the fact that $\Sigma_D$ generates $W(E_8)$, and $\mathrm{S}(n) \hookrightarrow W(E_8)$.
{\scriptsize\begin{align*}
    \TLO{X}{0,1} &= \GateAt{X}{0} \circ \GateAt{X}{1} \circ \CCGate{X}{0}{1} \circ \GateAt{X}{1} \circ \GateAt{X}{0}
    &
    \TLO{X}{1,2} &= \GateAt{X}{0} \circ \CCGate{X}{0}{1} \circ \CCGate{X}{0}{2} \circ \CCGate{X}{0}{1} \circ \GateAt{X}{0}
    \\
    \TLO{X}{2,3} &= \GateAt{X}{0} \circ \CCGate{X}{0}{1} \circ \GateAt{X}{0}
    &
    \TLO{X}{3,4} &= \GateAt{X}{0} \circ \GateAt{X}{2} \circ \CCGate{X}{0}{1} \circ \GateAt{X}{0} \circ \CCGate{X}{1}{2} \circ \CCGate{X}{0}{2} \circ \CCGate{X}{1}{2} \circ \GateAt{X}{0} \circ \CCGate{X}{0}{1} \circ \GateAt{X}{2} \circ \GateAt{X}{0}
    \\
    \TLO{X}{4,5} &= \GateAt{X}{1} \circ \CCGate{X}{0}{1} \circ \GateAt{X}{1}
    &
    \TLO{X}{5,6} &= \CCGate{X}{0}{1} \circ \CCGate{X}{0}{2} \circ \CCGate{X}{0}{1}
\end{align*}}%
These relations can be validated with respect to $\interp{\cdot}_O$, and consequently introduced via a sequence of \textbf{Gen($+$)} operations.
These definitional relations can then be used to eliminate the two-level operators of type $X$, as outlined in \cref{Appendix:DerivedGens}.

\begin{figure}[t]
  {\scriptsize\begin{minipage}[t]{0.35\textwidth}
    \begin{align}
    \TLO{(-1)}{0}{}^2 &\,\approx \, \epsilon \\
    \TLO{K}{0,1,2,3}{}^2 &\,\approx \, \epsilon \\
    \TLO{X}{1,2} \cdot \TLO{(-1)}{0} &\,\approx \, \TLO{(-1)}{0} \cdot \TLO{X}{1,2} \\
    \TLO{X}{2,3} \cdot \TLO{(-1)}{0} &\,\approx \, \TLO{(-1)}{0} \cdot \TLO{X}{2,3} \\
    \TLO{X}{3,4} \cdot \TLO{(-1)}{0} &\,\approx \, \TLO{(-1)}{0} \cdot \TLO{X}{3,4} \\
    \TLO{X}{4,5} \cdot \TLO{(-1)}{0} &\,\approx \, \TLO{(-1)}{0} \cdot \TLO{X}{4,5} \\
    \TLO{X}{5,6} \cdot \TLO{(-1)}{0} &\,\approx \, \TLO{(-1)}{0} \cdot \TLO{X}{5,6} \\
    \TLO{X}{6,7} \cdot \TLO{(-1)}{0} &\,\approx \, \TLO{(-1)}{0} \cdot \TLO{X}{6,7} \\
    \TLO{X}{4,5} \cdot \TLO{K}{0,1,2,3} &\,\approx \, \TLO{K}{0,1,2,3} \cdot \TLO{X}{4,5}
    \end{align}
  \end{minipage}
  \begin{minipage}[t]{0.63\textwidth}
    \begin{align}
    \TLO{X}{5,6} \cdot \TLO{K}{0,1,2,3} &\,\approx \, \TLO{K}{0,1,2,3} \cdot \TLO{X}{5,6} \\
    \TLO{X}{6,7} \cdot \TLO{K}{0,1,2,3} &\,\approx \, \TLO{K}{0,1,2,3} \cdot \TLO{X}{6,7} \\
    \TLO{(-1)}{4} \cdot \TLO{K}{0,1,2,3} &\,\approx \, \TLO{K}{0,1,2,3} \cdot \TLO{(-1)}{4} \\
    \TLO{(-1)}{0} \cdot \TLO{(-1)}{4} &\,\approx \, \TLO{(-1)}{4} \cdot \TLO{(-1)}{0} \\
    \TLO{K}{0,1,2,3} \cdot \TLO{K}{4,5,6,7} &\,\approx \, \TLO{K}{4,5,6,7} \cdot \TLO{K}{0,1,2,3} \\
    \TLO{X}{0,1} \cdot \TLO{K}{0,1,2,3} &\,\approx \, \TLO{K}{0,1,2,3} \cdot \TLO{X}{0,1} \cdot \TLO{(-1)}{1} \cdot \TLO{(-1)}{3} \\
    \TLO{X}{1,2} \cdot \TLO{K}{0,1,2,3} &\,\approx \, \TLO{(-1)}{0} \cdot \TLO{K}{0,1,2,3} \cdot \TLO{(-1)}{0} \cdot \TLO{K}{0,1,2,3} \cdot \TLO{(-1)}{0} \\
    \TLO{X}{2,3} \cdot \TLO{K}{0,1,2,3} &\,\approx \, \TLO{K}{0,1,2,3} \cdot \TLO{X}{1,3} \\
    \TLO{K}{0,1,2,3} \cdot \TLO{K}{1,3,4,5} &\,\approx \, \TLO{K}{1,3,4,5} \cdot \TLO{K}{0,1,2,3} \\
    \TLO{(-1)}{0} \cdot \TLO{(-1)}{4} \cdot \TLO{X}{0,4} \cdot \rho &\,\approx\, \rho \cdot \TLO{X}{0,4} \cdot \TLO{(-1)}{4} \cdot \TLO{(-1)}{0}
  \end{align}
  \end{minipage}}
  \caption{The relations in $\TLORel{8}{4}$, sufficient to extend from $W( E_8 )$ to $O( 8, \D )$.}
  \label{Fig:O8DRels:Dyadic}
  \hfill\rule{0.95\textwidth}{0.1mm}\hfill\hbox{}
\end{figure}

In this new presentation, then relations \cref{Eq:O8DRels:Orig:Perm1,Eq:O8DRels:Orig:Perm5,Eq:O8DRels:Orig:Perm6} are replaced by relations over $\Sigma_D$.
Since $R_0$ is complete for $W( E_8 )$, then these relations can be derived from $R_0$.
Consequently, these relations can be eliminated with a sequence of \textbf{Rel($-$)} transformations.
This yields a new set of relations, denoted $\TLORel{8}{4}$, which can be found in \cref{Fig:O8DRels:Dyadic}.
For simplicity of presentation, we used $\TLO{X}{0,1}$ through to $\TLO{X}{6,7}$ as abbreviations for the circuits given above.
Furthermore, we take $CCZ = \TLO{(-1)}{7}$ to be a generator with $\TLO{(-1)}{0}$ an alias for $\GateAt{X}{0} \cdot \GateAt{X}{1} \cdot CCZ \cdot \GateAt{X}{1} \cdot \GateAt{X}{0}$.
Then define $\Sigma_1 = \Sigma_D \cup \{ \TLO{K}{0,1,2,3}, CCZ \}$ and $R_1 = R_0 \cup \TLORel{8}{4}$, where $\TLO{K}{0,1,2,3}$ corresponds to a negatively controlled $K$ gate.

\begin{theorem}
    $\mathrm{O}(8,\D) \cong \langle \Sigma_1 \mid R_1 \rangle$.
\end{theorem}

\subsection{A Minimal Generating Set for \texorpdfstring{$\mathrm{O}(8, \D)$}{O(8,D)}}

It turns out that $CCZ$ and $\TLO{K}{0,1,2,3}$ can be defined in terms of one-another, given that generators in $\Sigma_D$.
The decompositions are as follows.
{\scriptsize\begin{align*}
    CCZ &= \GateAt{K}{1,2} \circ \CGate{Z}{1}{2} \circ \GateAt{X}{0} \circ \TLO{K}{0,1,2,3} \circ \GateAt{X}{0} \circ \CGate{Z}{1}{2} \circ \TLO{K}{0,1,2,3} \circ \TLO{X}{5,6}
    &
    \TLO{K}{0,1,2,3} &= ( \GateAt{K}{1,2} \circ CCZ )^3 \circ \TLO{X}{5,6}
\end{align*}}%
Given this observation, it seems natural to eliminate the $\TLO{K}{0,1,2,3}$, given that it is not a common generator in quantum computation.
However, the set $\Sigma_D \cup \{ CCZ \}$ is minimal, whereas the set $\Sigma_D \cup \{ \TLO{K}{0,1,2,3} \}$ is not.
In other words, choosing the generator $\TLO{K}{0,1,2,3}$ enables a smaller generating set, whereas choosing the generator $CCZ$ allows for more conventional circuit decompositions.
For this reason, we choose to keep both $\TLO{K}{0,1,2,3}$ and $CCZ$ in our presentation.
The minimality of these generating sets are proven in \cref{Appendix:E8MinGens}, using the same techniques as in \cref{Sec:E8:MinGens}.

\begin{theorem}
    \label{Thm:O8DMin}
    The following generating sets are minimal for $\mathrm{O}(8,\D)$.
    \begin{enumerate}
    \item $\Sigma_K = \left\{ \GateAt{X}{0}, \CGate{X}{0}{1}, \CCGate{X}{1}{2}, \TLO{K}{0,1,2,3} \right\}$.
    \item $\Sigma_Z = \left\{ \GateAt{X}{0}, \CGate{X}{0}{1}, \CCGate{X}{1}{2}, \GateAt{K}{1,2}, CCZ \right\}$
    \end{enumerate}
\end{theorem}

\section{Extending to the 3-Qubit Toffoli-Hadamard Circuits}
\label{Sec:TofH}

We now give a presentation of $\textsf{TofH}(3)$, by leveraging the presentation of $\mathrm{O}(8, \D)$ found in \cref{Sec:O8D}.
The argument in this section closely follows \cite[Section~5]{LiRoss2021}.
From \cite{AmyGlaudell2020}, it is known that $\textsf{TofH}(3)$ is obtained by adding the generator $\GateAt{H}{2}$ to $\mathrm{O}(8, \D)$.
Let $\Sigma_2 = \Sigma_1 \cup \{ \GateAt{H}{2} \}$ and $R_2$ extend the set $R_1$ with all relations found in \cref{Fig:G2Rels}.
Using the relations in \cref{Fig:G2Rels}, the generator $\GateAt{H}{2}$ can be moved from the left-hand side to the right-hand side of any word over $\Sigma_1$.
Since $\GateAt{H}{2}$ is self-inverse, this is sufficient to decide equality in $\textsf{TofH}(3)$.

\begin{figure}[t]
  {\scriptsize\begin{minipage}{0.55\textwidth}
    \begin{align}
    \GateAt{H}{2} \cdot \GateAt{X}{0}
        &\,\approx \,
        \GateAt{X}{0} \cdot \GateAt{H}{2}
        \label{eq:g21}
    \\
    \GateAt{H}{2} \cdot \CGate{X}{0}{1}
        &\,\approx \,
        \CGate{X}{0}{1} \cdot \GateAt{H}{2}
        \label{eq:g22}
    \\
    \GateAt{H}{2} \cdot \CCGate{X}{1}{2}
        &\,\approx
        \, K_{0,1} \cdot K_{1,2} \cdot CCZ \cdot K_{1,2} \cdot K_{0,1} \cdot \GateAt{H}{2}
        \label{eq:g23}
    \end{align}
  \end{minipage}
  \begin{minipage}{0.43\textwidth}
    \begin{align}
    \GateAt{H}{2} \cdot CCZ
        &\,\approx \,
        \CCGate{X}{0}{1} \cdot \GateAt{H}{2}
        \label{eq:g24}
    \\
    \GateAt{H}{2} \cdot K_{1,2}
        &\,\approx \,
        K_{1,2} \cdot \GateAt{H}{2}
        \label{eq:g25}
    \\
    \GateAt{H}{2} \cdot \TLO{K}{0,1,2,3}
        &\,\approx \,
        \TLO{K}{0,1,2,3} \cdot \GateAt{H}{2}
        \label{eq:g25a}
    \\
    \GateAt{H}{2} \cdot \GateAt{H}{2}
        &\,\approx \,
        \epsilon
        \label{eq:g26}
    \end{align}
  \end{minipage}}%
  \caption{Additional relations for $\textsf{TofH}(3)$.}
  \label{Fig:G2Rels}
  \hfill\rule{0.95\textwidth}{0.1mm}\hfill\hbox{}
\end{figure}

The proof proceeds as follows.
In \cref{lem:lcommute}, it is shown that $\GateAt{H}{2}$ commutes with every word in over $\Sigma_1$ using only the relations in $R_2$.
This is used in \cref{lem:commute}, to show that every word over $\Sigma_2$ can be rewritten as a word over $\Sigma_1$, followed by at most one $\GateAt{H}{2}$ gate.
Since $R_1 \subseteq R_2$ is a complete equational theory for $\mathrm{O}( 8, \D )$ with every element of $\textsf{TofH}(3)$ of the form described in \cref{lem:commute}, it follows that $R_2$ is a complete equational theory for $\textsf{TofH}(3)$ (see \cref{Thm:App:TofHComplete}).

\begin{lemma}
    \label{lem:lcommute}
    If $w \in \Sigma_1^*$, then there exists a $w' \in \Sigma_1^*$ such that $\GateAt{H}{2} \cdot w \,\sim_{R2}\, w' \cdot \GateAt{H}{2}$.
\end{lemma}

\begin{proof}
    The proof follows by induction on $|w|$.
    \begin{enumerate}
    \item[--] \textbf{Base Case}.
          If $|w| = 0$, then $\GateAt{H}{2} \cdot w = w \cdot \GateAt{H}{2}$.
          Then $\GateAt{H}{2} \cdot w \,\sim_{R2}\, w \cdot \GateAt{H}{2}$ by the transitivity of $( \sim_{R2 })$.
    \item[--] \textbf{Inductive Case}.
          Assume that for some $k \in \mathbb{N}$, if $|w| = k$, then there exists a $w' \in \Sigma_1^*$ such that $\GateAt{H}{2} \cdot w \,\sim_{R2}\, w' \cdot \GateAt{H}{2}$.
    \item[--] \textbf{Inductive Step}.
          Assume that $|w| = k + 1$.
          Then there exists a $u \in \Sigma_1^*$ and $x \in \Sigma_1$ such that $x \cdot u = w$.
          It follows by one of \cref{eq:g21,eq:g22,eq:g23,eq:g24,eq:g25,eq:g25a}, that there exits a $u' \in \Sigma_1^*$ such that $\GateAt{H}{2} \cdot w \,\sim_{R2}\, u' \cdot \GateAt{H}{2} \cdot u$.
          Since $|u| = k$, then by the inductive hypothesis, there exists a $w' \in \Sigma_1^*$ such that $\GateAt{H}{2} \cdot u \,\sim_{R2}\, w' \cdot \GateAt{H}{2}$.
          Then $u' \cdot \GateAt{H}{2} \cdot u \,\sim_{R2}\, u' \cdot w' \cdot \GateAt{H}{2}$.
          Then $\GateAt{H}{2} \cdot u \,\sim_{R2}\, u' \cdot w' \cdot \GateAt{H}{2}$ by the transitivity of $( \sim_{R2} )$, and the inductive case holds.
    \end{enumerate}
    Then by the principle of induction, there exists a $w' \in \Sigma_1^*$ such that $\GateAt{H}{2} \cdot w \,\sim_{R2}\, w' \cdot \GateAt{H}{2}$.
 \end{proof}

\begin{lemma}
  \label{lem:commute}
  If $w \in \Sigma_2^*$, then there exists some $w' \in \Sigma_1^*$ and $\ell \in \{ 0, 1 \}$ such that $w \,\sim_{R_2}\, w' \cdot \GateAt{H}{2}^\ell$.
\end{lemma}

\begin{proof}
  Let $f: \Sigma_2^* \to \N$ map each word $w \in \Sigma_2^*$ to the number of $\GateAt{H}{2}$ symbols in $w$.
  The proof follows by induction on $f( w )$.
  \begin{enumerate}
  \item[--] \textbf{Base Case}.
        If $f( w ) = 0$, then $w \in \Sigma_1^{*}$.
        Then $w \,\sim_{R2}\, w \cdot \GateAt{H}{2}^0$ by the reflexivity of $( \sim_{R2} )$
  \item[--] \textbf{Inductive Hypothesis}.
        Assume that for some $k \in \mathbb{N}$, if $f( w ) = k$, then there exists some $w' \in \Sigma_1^*$ and $\ell \in \{ 0, 1 \}$ such that $w \,\sim_{R2}\, w' \cdot \GateAt{H}{2}^\ell$.
  \item[--] \textbf{Inductive Step}.
        Assume that $f( w ) = k + 1$.
        Then there exists $w_1 \in \Sigma_2^{*}$ and $w_2 \in \Sigma_1^{*}$ such that $w = w_1 \cdot \GateAt{H}{2} \cdot w_2$ with $f( w_1 ) = f( w ) - 1$.
        Then by \cref{lem:lcommute}, $w \,\sim_{R2}\, w_1 \cdot w_2' \cdot \GateAt{H}{2}$ for some $w_2' \in \Sigma_1^{*}$.
        Then $f( w_1 \cdot w_2' ) = f( w_1 ) = f( w ) - 1$.
        By the inductive hypothesis, there exists some $w_3 \in \Sigma_1^*$ and $\ell \in \{ 0, 1 \}$ such that $w_1 \cdot w_2' \,\sim_{R2}\, w_3 \cdot \GateAt{H}{2}^\ell$.
        Then $w \,\sim_{R2}\, w_3 \cdot \GateAt{H}{2}^{\ell+1}$.
        If $\ell = 0$, then $w \,\sim_{R2}\, w_3 \cdot \GateAt{H}{2}$ and we are done.
        Otherwise, if $\ell = 1$, then $w \,\sim_{R2}\, w_3$ by \cref{eq:g26}.
        In either case, there exists an $\ell' \in \{ 0, 1 \}$ such that $w \,\sim_{R2}\, w_3 \cdot \GateAt{H}{2}^{\ell'}$ and the inductive step holds.
  \end{enumerate}
  Then by the principle of induction, there exists some $w' \in \Sigma_1^*$ and $\ell \in \{ 0, 1 \}$ such that $w \,\sim_{R_2}\, w' \cdot \GateAt{H}{2}^\ell$.
\end{proof}

\begin{theorem}
    \label{Thm:App:TofHComplete}
    For all $w_1, w_2 \in \Sigma_2^*$, $\interp{w_1}_H^* = \interp{w_2}_H^*$ if and only if $w_1 \sim_{R2} w_2$.
\end{theorem}

\begin{proof}
  It follows by matrix multiplication that the relations in \cref{Fig:G2Rels} are sound.
  It remains to be shown that the relations in \cref{Fig:G2Rels} are complete.
  Let $w_1 \in \Sigma_2^*$ and $w_2 \in \Sigma_2^*$ such that $\interp{w_1}_H^* = \interp{w_2}_H^*$.
  By \cref{lem:commute}, there exists $\ell_1, \ell_2 \in \{ 0, 1 \}$ and $w_1', w_2' \in \Sigma_1^*$ such that $w_1 \,\sim_{R2}\, w_1' \cdot \GateAt{H}{2}^{\ell_1}$ and $w_2 \,\sim_{R2}\, w_2' \cdot \GateAt{H}{2}^{\ell_2}$.
  Since $\interp{w_1}_H^* \in \mathrm{O}(8, \D) \cong \langle \Sigma_1, R_1 \rangle$ with $R_1 \subseteq R_2$, then $w_1' \cdot \GateAt{H}{2}^{\ell_1} \,\sim_{R2}\, w_2' \cdot \GateAt{H}{2}^{\ell_1}$.
  Assume for the intent of contradiction that $\ell_1 \ne \ell_2$.
  Then $\interp{w_1'}_H^* = \interp{w_2'}_H^* \circ \interp{\GateAt{H}{2}}_H$.
  Then $\interp{w_2'}_H^* \circ \interp{\GateAt{H}{2}}_H^* \in \mathrm{O}(8, \D)$.
  However, $\interp{w_2'}_H^* \circ \interp{\GateAt{H}{2}}_H$ has a denominator of the form $1 / (2^k \sqrt{2})$, and therefore $\interp{w_2'}_H^* \circ \interp{\GateAt{H}{2}}_H \not \in \mathrm{O}(8, \D)$.
  By contradiction, $\ell_1 = \ell_2$.
  Since $\ell_1 = \ell_2$, then $w_1 \,\sim_{R2}\, w_2$ by the transitivity and symmetry of $(\sim_{R2})$.
  Since $w_1$ and $w_2$ were arbitrary, then the relations in \cref{Fig:G2Rels} are complete.
\end{proof}
\section{Conclusion}

We used the geometry of $W( E_8 )$ to obtain a circuit  presentation for the $3$-qubit Toffoli-K circuits, and then leveraged \cite{LiRoss2021} to obtain a finite presentation of $3$-qubit Toffoli-Hadamard circuits. Our presentation contains 65 relations, compared to the 2113 relations of \cite{LiRoss2021}.
There are several directions for future work. We hope to simplify our presentation by further reducing the number of relations. In addition, we plan to explore the structural properties of the group of 3-qubit Toffoli-Hadamard circuits and of its subgroups. In particular, it is known that the group $\mathrm{O}(8, \D)$ is generated by reflections, but it is not known whether this group can be presented as an (infinite) Coexeter group.
From an applied perspective, we also hope to explore applications of these presentations to circuit optimization and equivalence checking.

\bibliographystyle{eptcs}
\bibliography{cox}

\newpage
\appendix
\section{Semantic Tietze Transformations}
\label{Appendix:Tietze}

This section recalls what it means for a function to induce a monoid homomorphism.
This is then used to prove the soundness and completeness of Tietze transformations with respect to semantic interpretations.
The majority of this section is dedicated to proving that all induced homomorphisms are isomorphisms, and in the case of \textbf{Gen($+$)}, the extension is unique.
The uniqueness of this extension is necessary to prove that each generator in $\Sigma_D$ has the intended matrix semantics.

\subsection{Induced Monoid Homomorphisms}

Let $\Sigma$ be an alphabet and $M$ a monoid.
For each function $f: \Sigma \to M$, define the function $f^*: \Sigma^* \to M$ such that $f^*( x_1 \cdot x_2 \cdots x_n ) = 1_M \circ f( x_1 ) \circ f( x_2 ) \circ \cdots \circ f( x_n )$ for all $x_1 \cdot x_2 \cdots x_n \in \Sigma^*$.
It can then be shown that $f^*$ is the unique monoid homomorphism such that $f^*( x ) = f( x )$ for all $x \in \Sigma$~\cite{Johnson1990}.
Given a set of relations $R \subseteq \Sigma^* \times \Sigma^*$, it can then be asked whether $f$ induces a monoid homomorphism between $G = \langle \Sigma \mid R \rangle$ and $M$.
This question is answered by the following theorem.

\begin{theorem}[\cite{BookOtto1993}]
    \label{Thm:App:IndHom}
    Let $M$ and $G = \langle \Sigma \mid R \rangle$ be monoids with $f: \Sigma \to M$.
    There there exists a unique monoid homomorphism $\varphi: G \to M$ such that $f^* = \varphi \circ \pi_G $ if and only if $f^*( q ) = f^*( r )$ for all $q \approx_R r$.
    In this case, $\Im( \varphi ) = \langle \varphi( \pi_G( \Sigma ) ) \rangle$.
\end{theorem}

\cref{Thm:App:IndHom} characterizes when $f$ induces a monoid homomorphism, and how to construct this induced homomorphism $\varphi$.
It can then be asked how the construction of $\varphi$ interacts with the introduction or elimination of generators.
As outlined by the following theorems, the elimination of generators corresponds to certain restrictions of $\varphi$, whereas the introduction of generators corresponds to certain unique extensions of $\varphi$.
 
\begin{lemma}
    \label{Lemma:App:ElimDGen0}
    Let $\Sigma$ be an alphabet, $x \in \Sigma$, $\Sigma' = \Sigma \setminus \{ x \}$, and $D = \{ x \approx w \}$ for some $w \in (\Sigma')^*$.
    If $q \in \Sigma^*$, then there exists a $q' \in ( \Sigma' )^*$ such that $q \sim_{D} q'$.
\end{lemma}

\begin{proof}
    Let $f: \Sigma^* \to \mathbb{N}$ count the $x$ symbols in a word.
    Then the proof follows by induction on $f( q )$.
    \begin{enumerate}
    \item[--] \textbf{Base Case}.
          If $f( q ) = 0$, then $q \in (\Sigma')^*$.
    \item[--] \textbf{Inductive Hypothesis}.
          Assume that for some $k \in \mathbb{N}$, if $f( q ) = k$, then there exists some $q' \in (\Sigma')^*$ such that $q \sim_D q'$.
    \item[--] \textbf{Inductive Step}.
          Assume $f( q ) = k + 1$.
          Since $f( q ) > 1$, then there exists some $u, v \in \Sigma^*$ such that $q = u \cdot x \cdot v$.
          Since $k + 1 = f( q ) = f( u \cdot x \cdot v ) = f( u ) + f( x ) + f( v ) = f( u ) + f( v ) + 1$, then $k = f( u ) + f( v )$.
          Since $x \approx_D w$, then $q \sim_D u \cdot w \cdot v$ with $f( u \cdot w \cdot v ) = f( u ) + f( v ) = k$.
          Then by the inductive hypothesis, there exists some $q' \in ( \Sigma' )^*$ such that $u \cdot w \cdot v \sim_D q'$.
          Then $q \sim q'$.
    \end{enumerate}
    By the principle of induction, there exists some $q' \in (\Sigma')^*$ such that $q \sim_D q'$.
\end{proof}

\begin{theorem}
    \label{Cor:App:IndExt}
    Let $M$ and $G = \langle \Sigma \mid R \rangle$ be monoids with $f: \Sigma \to M$, and $H = \langle \Sigma \cup \{ x \} \mid R \cup \{ x \approx w \} \rangle$ be a monoid for some $x \not \in \Sigma$ and $w \in \Sigma^*$.
    Define $g: \Sigma \cup \{ x \} \to M$ such that $g|_{\Sigma} = f$ and $g: x \mapsto f^*( w )$.
    If $f$ induces a monoid homomorphism from $G$ to $M$, then $g$ is the unique extension of $f$ to induce a monoid homomorphism from $H$ to $M$.
    Furthermore, if $f$ induces an injection (resp. surjection) from $G$ to $M$, then $g$ induces an injection (resp. surjection) from $H$ to $M$.
\end{theorem}

\begin{proof}
    Let $\Pi = \Sigma \cup \{ x \}$ and $Q = R \cup \{ x \approx w \}$.
    Assume that $f$ induces a homomorphism from $G$ to $M$.
    \begin{itemize}
        \item[--] (\textbf{Induced Hom}).
              Let $q \approx_Q r$.
              Then either $q \approx_R r$ or $( q, r ) = ( x, w )$.
              First, assume that $q \approx_R r$.
              Then $f^*( q ) = f^*( r )$ by \cref{Thm:App:IndHom}.
              Then $g^*( q ) = (g|_\Sigma)^*( q ) = f^*( q ) = f^*( r ) = (g|_\Sigma)^*( r ) = g^*( r )$.
              Next, assume that $( q, r ) = ( x, w )$.
              Then $g^*( x ) = g( x ) = f^*( w ) = (g|_\Sigma)^*( w ) = g^*( w )$.
              In either case $g^*( q ) = g^*( r )$.
              Since $q \approx_R r$ was arbitrary, then $g^*( q ) = g^*( r )$ for all $q \approx_R r$.
              Then $g$ induce a monoid homomorphism from $H$ to $M$ by \cref{Thm:App:IndHom}.
        \item[--] (\textbf{Uniqueness}).
              Assume that $k: \Sigma \cup \{ x \} \to M$ is an extension of $f$ which induces a monoid homomorphism.
              Since $k$ induces a monoid homomorphism and $x \approx_Q w$, then $k^*( x ) = k^*( w )$ by \cref{Thm:App:IndHom}.
              Then $k( x ) = k^*( x ) = k^*( w ) = (k|_\Sigma)^*( w ) = f^*( w ) = g( x )$ by construction of $g$.
              Moreover, since $k|_\Sigma = f = g|_\Sigma$, then $k = g$.
              Since $k$ was arbitrary, then $g$ is unique.
        \item[--] (\textbf{Injectivity}).
              Let $f$ induce $\varphi$ and $g$ induce $\rho$.
              Assume that $\varphi$ is injective.
              Let $q, r \in \Pi^*$ such that $\rho( \pi_H( q ) ) = \rho( \pi_H( r ) )$.
              Since $x \approx_Q w$, then by \cref{Lemma:App:ElimDGen0} there exists $q', r' \in \Sigma^*$ such that $q \sim_Q q'$ and $r \sim_Q r'$.
              Then $\pi_H( q ) = \pi_H( q' )$ and $\pi_H( r ) = \pi_H( r' )$.
              Then $\rho( \pi_H( q' ) ) = \rho( \pi_H( r' ) )$.
              Then $g^*(q') = g^*(r')$.
              Since $q', r' \in \Sigma^*$, then $f^*(q') = f^*(r')$.              Then $\varphi( \pi_G( q' ) ) = \varphi( \pi_G( r' ) )$.
              Since $\varphi$ is injective, then $\pi_G( q' ) = \pi_G( r' )$.
              Then $q' \sim_R r'$.
              Since $R \subseteq Q$, then $q' \sim_Q r'$.
              Then $\pi_H( q' ) = \pi_H( r' )$.
              Then $\pi_H( q ) = \pi_H( r )$.
              Since $q$ and $r$ were arbitrary, then $\rho$ is injective.
        \item[--] (\textbf{Surjectivity}).
              Let $f$ induce $\varphi$ and $g$ induce $\rho$.
              Assume that $\varphi$ is surjective.
              Since $\varphi( \pi_G( y ) ) = f^*( y ) = g^*( y ) = \rho( \pi_H( y ) )$ for each $y \in \Sigma$, then $\varphi( \pi_G( \Sigma ) ) \subseteq \Im( \rho )$.
              Since $\varphi$ is surjective, then $H = \Im( \varphi ) = \langle \varphi( \pi_G( \Sigma ) ) \rangle \le \Im( \rho ) \le H$ and $\rho$ is surjective.
    \end{itemize}
    Therefore, $g$ is the unique extension of $f$ to induce a monoid homomorphism from $H$ to $M$, with $g$ inducing an injection (resp. surjection) whenever $f$ induces an injection (resp. surjection).
\end{proof}

\begin{theorem}
    \label{Cor:App:IndRes}
    Let $M$ and $G = \langle \Sigma \mid R \rangle$ be monoids with $f: \Sigma \to M$ and $H = \langle \Pi \mid Q \rangle$ where $\Pi = \Sigma \setminus \{ x \}$ for some $x \in \Sigma$ and $Q = R \setminus \{ x \approx w \}$ for some $x \approx_R w$.
    If $Q \subseteq \Pi^* \times \Pi^*$ and $f$ induces a monoid homomorphism from $G$ to $M$, then $f|_{\Pi}$ induces a monoid homomorphism from $H$ to $M$.
    Furthermore, if $f$ induces an injection (resp. surjection) from $G$ to $M$, then $f|_{\Pi}$ induces an injection (resp. surjection) from $H$ to $M$.
\end{theorem}

\begin{proof}
    Assume that $f$ induces a homomorphism from $G$ to $M$.
    \begin{itemize}
        \item[--] (\textbf{Induced Hom}).
              Since $f$ induces a monoid homomorphism, then $f^*(q) = f^*(r)$ for all $q \approx_R r$ by \cref{Thm:App:IndHom}.
              Since $Q \subseteq R$ and $Q \subseteq \Pi^* \times \Pi^*$, then $f|_{\Pi}^*(q) = f^*(q) = f^*(r) = f|_{\Pi}^*(r)$ for all $q \approx_Q r$.
              Then $f|_{\Pi}$ induces a monoid homomorphism from $H$ to $M$ by \cref{Thm:App:IndHom}.
        \item[--] (\textbf{Injectivity}).
              Let $f$ induce $\varphi$ and $f|_{\Pi}$ induce $\rho$.
              Assume that $\varphi$ is injective.
              Let $q, r \in \Pi^*$ such that $\rho( \pi_H( q ) ) = \rho( \pi_H( r ) )$.
              Then $f|_{\Pi}^*(q) = f|_{\Pi}^*(r)$.
              Then $f^*(q) = f^*(r)$.
              Then $\varphi( \pi_G( q ) ) = \varphi( \pi_G( r ) )$.
              Since $\varphi$ is an injective, then $\pi_G( q ) = \pi_G( r )$.
              Then $q \sim_R r$.
              Since $Q = R \setminus \{ x \approx w \}$, $Q \subseteq \Pi^* \times \Pi^*$, and $q, r \in \Pi^*$, then $q \sim_Q r$.
              Then $\pi_H( q ) = \pi_H( r )$.
              Since $q$ and $r$ were arbitrary, then $\rho$ is injective.
        \item[--] (\textbf{Surjectivity}).
              Let $f$ induce $\varphi$ and $f|_{\Pi}$ induce $\rho$.
              Assume that $\varphi$ is surjective.
              Since $\varphi( \pi_G( y ) ) = f^*( y ) = f|_{\Pi}^*( y ) = \rho( \pi( y ) ) \in \Im( \rho )$ for each $y \in \Pi$, then $\varphi( \pi_G( \Pi ) ) \subseteq \Im( \rho )$.
              Since $x \approx_R w$, then $\pi_G( x ) = \pi_G( w )$, and consequently $\varphi( \pi_G( x ) ) = \varphi( \pi_G( w ) ) = f^*( w ) = f|_{\Pi}^*( w ) = \rho( \pi_H( w ) ) \in \Im( \rho )$.
              Then $\varphi( \pi_G( \Sigma ) ) \subseteq \Im( \rho )$.
              Since $\varphi$ is surjective, then $H = \Im( \varphi ) \langle \varphi( \pi_G( \Sigma ) ) \le \Im( \rho ) \le H$ and $\rho$ is surjective.
    \end{itemize}
    Therefore, $f|_{\Pi}$ induces a monoid homomorphism from $H$ to $M$ with $f|_{\Pi}$ inducing an injection (resp. surjection) whenever $f$ induces an injection (resp. surjection).
\end{proof}

\subsection{Semantic Interpretations and Relations}

Let $G = \langle \Sigma \mid R \rangle$ be a monoid presentation with an interpretation $\interp{\cdot}_{G}: \Sigma \mapsto H$.
The \textbf{Rel($+$)} transformation states that if $r \in \Sigma^*$ and $q \in \Sigma^*$ with $r \sim_R q$, then $\langle \Sigma \mid R \rangle \cong \langle \Sigma \mid R \cup \{ r \} \rangle$.
In practice, deriving $q$ from $r$ can be challenging, and on a theoretical level, this is known to be undecidable~\cite{EndrullisGeuvers2011}.
However, it is rarely the case that one would try to prove $r \sim_R q$ without some intuition that $\pi_G( r ) = \pi_G( q )$.
In the case of this paper, this intuition comes from knowledge about operators in $O( 8, \D )$.
For example, if $M \circ N = A \circ B$, then for any complete set of relations $R$, it must be the case that $M \cdot N \sim_R A \cdot B$.
More generally, if $\interp{r}_\Sigma^* = \interp{q}_\Sigma^*$ with $\interp{\cdot}_\Sigma^*$ inducing an injection, then $r \sim_R q$.
This claim is established by the following theorem, and used freely throughout the paper to simplify derivations.

\begin{definition}[Valid Semantic Interpretation]
    A semantic interpretation $\interp{\cdot}_{\Sigma}: G \to H$ for a presentation $G = \langle \Sigma \mid R \rangle$ is valid if $\interp{q}_{\Sigma}^* = \interp{r}_{\Sigma}^*$ for all $r \approx_R q$.
\end{definition}

\begin{theorem}
    \label{Thm:SemRewrite}
    Let $G = \langle \Sigma \mid R \rangle$ be a presentation with a valid semantic interpretation $\interp{\cdot}_G: \Sigma \mapsto H$.
    If $\interp{\cdot}_G$ is injective and $\interp{q}_\Sigma^* = \interp{r}_\Sigma^*$, then $q \sim_R r$.
\end{theorem}

\begin{proof}
    Since $\interp{\cdot}_\Sigma$ is valid, then by \cref{Thm:App:IndHom}, $\interp{\cdot}_\Sigma$ induces a monoid homomorphism $\varphi: G \to H$ such that $\interp{\cdot}_\Sigma^* = \varphi \circ \pi_G$.
    Assume that $\interp{q}_\Sigma^* = \interp{r}_\Sigma^*$.
    Then $\varphi(\pi_G(q)) = \interp{q}_\Sigma^* = \interp{r}_\Sigma^* = \varphi(\pi_G(r))$.
    Since $\varphi$ is injective, then $\pi_G(q) = \pi_G(r)$.
    Then $q \sim_R r$.
\end{proof}

\subsection{Semantics and Generator Introduction}

In the previous section, it was assumed that $\interp{\cdot}_{\Sigma}$ inducted an injection.
This is a reasonable assumption.
For example, if $G = \langle \Sigma \mid R \rangle$ is a presentation for $H$, then there exists an isomorphism $G \cong H$ from which $\interp{\cdot}_{\Sigma}$ can be extracted.

This can become problematic when trying to translate a known presentation to a desired generator set.
Assume that $\langle \Pi \mid Q \rangle$ is a known presentation with a semantic interpretation $\interp{\cdot}_\Pi: \Pi \to H$, from with a presentation $\langle \Sigma \mid R \rangle$ is derivable via a sequence of Tietze transformations.
One would hope that after each Tietze transformation, the semantic interpretation continues to induce an injection, so that \cref{Thm:SemRewrite} continues to hold.
Furthermore, one would hope that after all of the Tietze transformations, $\interp{\cdot}_\Sigma$ is a valid semantic interpretation.

It will be shown that under reasonable assumptions, all Tietze transformations satisfy these assumptions.
The first concern is answered by \cref{Cor:App:IndExt} and \cref{Cor:App:IndRes}, which state that after each \textbf{Gen($+$)} or \textbf{Gen($-$)} transformation, $\interp{\cdot}_\Pi$ continues to induce a monoid homomorphism (resp. injection, surjection, isomorphism).
The only time injectivity might fail is after a new generator is added.
In \textbf{Gen($+$)}, the symbol $x \in \Sigma$ becomes an alias for $w \in \Pi^{*}$.
Then it suffices to check that $\interp{x}_\Sigma = \interp{w}^*_{\Pi}$.

\newpage
\section{From Derivations to Tietze Transformations}

This section describes higher-level structures to reason about Tietze transformations.
It is proven that each structure corresponds to a valid sequence of Tietze transformations, and is therefore sound for the isomorphism of finite monoid presentations.
The structures described in this section are used throughout \cref{Sec:E8} and \cref{Sec:O8D}.

\subsection{Digraphs and Termination}

This section reviews the basics of directed graphs.
A directed graph is a tuple $( V, E )$ such that $V \subseteq E \times E$.
A vertex $v$ is a \emph{child} of $u$ if $( u, v ) \in E$.
A vertex $v$ is a \emph{parent} of $u$ if $( v, u ) \in E$.
A \emph{path} in $( V, E )$ is some sequence $( u_0, u_1, \ldots, u_n )$ over $V$ such that $( u_{k-1}, u_k ) \in E$ for all $k \in \{ 1, 2, \ldots, n \}$.
A directed graph $( V, E )$ is \emph{acyclic} if $u_0 \ne u_n$ for all paths $( u_0, u_1, \ldots, u_n )$ in $( V, E )$.

\begin{lemma}[{\cite[Sect.~2.2]{BaaderNipkow1998}}]
    \label{Lemma:App:Acyc}
    If a digraph $( V, E )$ is finite and acyclic, then for every vertex $v \in V$, there exists some path of length $n$ ending (resp. starting) at $v$ such that every path ending (resp. starting) at $v$ has length at most $n$.
    In particular, there exists a vertex $v \in V$ such that $v$ has no children (resp. parents) in $( V, E )$.
\end{lemma}

Note that the original statement of \cref{Lemma:App:Acyc} concerned the termination of abstract rewriting systems (which can be modelled using paths through digraphs).
In the proofs that follow, this intuition is useful to keep in mind.
Indeed, \cref{Lemma:App:Acyc} is used to argue that certain rewriting procedures terminate, though the rewriting systems are never stated explicitly for simplicity.

\subsection{Derived Generators and Tietze Transformations}
\label{Appendix:DerivedGens}

\begin{figure}[t]
  \centering
\begin{tikzcd}
	&& {\GateAt{X}{1}} &&& {\CGate{Z}{0}{1}} \\
	& {\GateAt{Z}{1}} & {\CGate{X}{2}{0}} & {\CGate{X}{2}{1}} & {\CGate{X}{1}{2}} & {\CGate{Z}{0}{2}} & {\GateAt{Z}{0}} \\
	{\CGate{X}{1}{0}} & {K_{0,1}} & {\sigma_{0,2}} & {\CGate{X}{0}{2}} & {\sigma_{1,2}} & {\CGate{Z}{1}{2}} & {\GateAt{Z}{2}} \\
	{\sigma_{0,1}} & {\CCGate{X}{0}{2}} & {\CCGate{X}{0}{1}} & {\GateAt{X}{2}}
	\arrow[from=3-5, to=2-4]
	\arrow[from=3-4, to=3-5]
	\arrow[from=3-3, to=3-4]
	\arrow[from=3-2, to=3-3]
	\arrow[from=3-1, to=3-2]
	\arrow[from=3-5, to=2-5]
	\arrow[from=2-5, to=2-4]
	\arrow[from=2-6, to=3-5]
	\arrow[from=3-6, to=2-6]
	\arrow[curve={height=-12pt}, from=3-6, to=3-3]
	\arrow[curve={height=-6pt}, from=4-4, to=3-3]
	\arrow[from=4-3, to=4-2]
	\arrow[from=4-2, to=4-1]
	\arrow[from=4-1, to=3-1]
	\arrow[from=3-7, to=2-7]
	\arrow[curve={height=-30pt}, from=3-7, to=3-3]
	\arrow[from=3-3, to=2-3]
	\arrow[from=2-4, to=2-3]
	\arrow[from=2-6, to=1-6]
	\arrow[from=2-7, to=1-6]
	\arrow[from=2-3, to=1-3]
	\arrow[from=2-2, to=1-3]
\end{tikzcd}
  \caption{The derived generator graph for the defining relations in $R_0$.}
  \label{Fig:App:Deps}
\end{figure}

In \cref{Sec:E8:DIntro}, it was claimed that $19$ of the generators in $\Sigma_D$ could be introduced freely, because their defining relations formed an acyclic dependency graph.
Later, in \cref{Appendix:E8MinGens}, the same argument was used to remove these $19$ generators from the generating set for $W( E_8 )$.
Similar arguments appear throughout \cref{Sec:O8D}.
The goal of this section is to establish both directions rigorously.
First, formal definitions are given for defining relations and derived generator graphs (i.e., the graph of dependencies between the relations).
Then, \cref{Thm:App:IntroDGens} and \cref{Thm:App:ElimDGens} are established to justify the claims of \cref{Sec:E8:DIntro} and \cref{Sec:E8:MinGens} respectively.
Finally, it is shown that these theorems apply to the generators in \cref{Sec:E8}.

\begin{definition}[Defining Relations]
    Let $\Sigma$ be an alphabet.
    A \emph{$\Sigma$-defining relation} for $x \in \Sigma$ is a relation of the form $x \approx w$ where $w \in \Sigma^*$.
    A \emph{family of $\Sigma$-defining relations for $\Pi \subseteq \Sigma$} is a set $\{ r_x \mid x \in \Pi \}$ such that for each $x \in \Pi$, $r_x$ is a $\Sigma$-defining relation for $x$.
\end{definition}

\begin{definition}[Derived Generator Graph]
    Let $D$ be a family of $\Sigma$-defining relations for $\Pi \subseteq \Sigma$.
    The \emph{derived generator graph for $D$} is the digraph $\Gamma_D( D ) := ( \Pi, E )$ such that $( x, y ) \in E$ if and only if there exists $w, w' \in \Sigma^*$ such that $x \approx w \cdot y \cdot w'$ is a relation in $D$.
\end{definition}

The proofs of \cref{Thm:App:IntroDGens} and \cref{Thm:App:ElimDGens} both rely heavily on derived generator graphs.
Intuitively, a generator can be introduced (resp. eliminated) if it has no dependencies (resp. dependants) in the derived generator graph.
It is always possible to find such a generator, provided the graph is acyclic.
However, once a generator has been removed from the defining relations, it is important that the new derived generator graph is also acyclic.
In fact, the new derived generator graph is always a subgraph of the previous derived generator graph, as outlined in \cref{Lemma:App:Subgraph0}.

\begin{lemma}
    \label{Lemma:App:Subgraph0}
    If $D$ is a family of $\Sigma$-defining relations and $D' \subseteq D$, then $\Gamma_D( D' )$ is a subgraph of $\Gamma_D( D )$.
\end{lemma}

\begin{proof}
    If $D$ is a family of $\Sigma$-defining relations for $\Pi$, then there exists $\Pi' \subseteq \Pi$ such that $D'$ is a family of $\Sigma$-defining relations for $\Pi'$.
    Then $\Gamma_D( D ) = ( \Pi, E )$ and $\Gamma_D( D' ) = ( \Pi', E' )$ for some $E \subseteq \Pi \times \Pi$ and $E' \subseteq \Pi' \times \Pi$.
    Let $( x, y ) \in E'$,
    Then there exists some $w, w' \in \Sigma^*$ such that $x \approx_{D'} w \cdot y \cdot w'$.
    Since $D' \subseteq D$, then $x \approx_{D} w \cdot y \cdot y'$.
    Then $( x, y ) \in E$.
    Since $( x, y )$ was arbitrary, then $E' \subseteq E$.
    Since $\Pi' \subseteq \Pi$ and $E' \subseteq E$, then $\Gamma_D( D' )$ is a subgraph of $\Gamma_D( D )$.
\end{proof}

\subsubsection{Introduction of Derived Generators}

The goal of this section is to prove \cref{Thm:App:IntroDGens}.
First, \cref{Lemma:App:DLeaf} is introduced to prove that every finite set of defining relations with an acyclic graph must contain at least one defining relation $x \approx w$ whose right-hand side consists only of primitive generators.
Since $w$ is consists only of primitive generators, then it may be introduced by a \textbf{Gen($+$)} transformation.
The proof then follows by induction on the number of defining relations, as outlined below.

\begin{lemma}
    \label{Lemma:App:DLeaf}
    If $D$ is a family of $\Sigma$-defining relations for $\Pi$ with $\Pi$ finite and $\Gamma_D( D )$ acyclic, then there exists a relation $x \approx w$ in $D$ such that $w \in ( \Sigma \setminus \Pi )^*$.
\end{lemma}

\begin{proof}
    Since $\Pi$ is finite and $\Gamma_D( D )$ is acyclic, the by \cref{Lemma:App:Acyc}, there exists some vertex $x \in \Pi$ such that $x$ has no children in $\Gamma_D( D )$.
    Let $x \approx_D w$ be the $\Sigma$-defining relation for $x$ in $D$ with $n = |w|$.
    Let $k \in [n]$.
    Assume for the intent of contradiction that $w_k \in \Pi$.
    Then there exists some $w', w'' \in \Sigma^*$ such that $w = w' \cdot w_k \cdot w''$.
    Then $w_k$ is a child of $x$ by definition.
    However, $x$ has no children by assumption.
    Then $w_k \not \in \Pi$ by contradiction.
    Since $k$ was arbitrary, then $w \in ( \Sigma \setminus \Pi )^*$.
\end{proof}

\begin{theorem}
    \label{Thm:App:IntroDGens}
    Let $\Sigma \subseteq \widehat{\Sigma}$ be an alphabet with $\Pi = \widehat{\Sigma} \setminus \Sigma$ finite and $R \subseteq \Sigma^* \times \Sigma^*$.
    If $D$ is a family of $\widehat{\Sigma}$-defining relations for $\Pi$ with $\Gamma_D( D )$ acyclic, then there exists a length $|\Pi|$ sequence of \textbf{Gen($+$)} transformations between $\langle \Sigma \mid R \rangle$ and $\langle \widehat{\Sigma} \mid R \cup D \rangle$.
\end{theorem}

\begin{proof}
    The proof follows by induction on $|\Pi|$.
    \begin{enumerate}
    \item[--] \textbf{Base Case}.
          Assume that $|\Pi| = 0$.
          Then $|D| = 0$ and $\langle \Sigma \mid R \rangle = \langle \Sigma' \mid R \cup D \rangle$.
          Then there exists a length $0$ sequence of \textbf{Gen($+$)} transformations between $\langle \Sigma \mid R \rangle$ and $\langle \widehat{\Sigma} \mid R \cup D \rangle$.
    \item [--] \textbf{Inductive Hypothesis}.
          Assume that for some $k \in \mathbb{N}$, if $|\Pi| = k$ and $\Pi_D( D )$ is acyclic, then there exists a length $k$ sequence of \textbf{Gen($+$)} transformations between $\langle \Sigma \mid R \rangle$ and $\langle \widehat{\Sigma} \mid R \cup D \rangle$.
    \item[--] \textbf{Inductive Step}.
          Assume that $|\Pi| = k + 1$ and $\Gamma_D( D )$ is acylic.
          Since $\Pi$ is finite and $\Gamma_D( D )$ is acyclic, then by \cref{Lemma:App:DLeaf} there exists a $\widehat{\Sigma}$-defining relation $x \approx w$ in $D$.
          Let $\Lambda = \Sigma \cup \{ x \}$ and $Q = R \cup \{ x \approx w \}$.
          Then $D \setminus \{ r_x \}$ is a family of $\widehat{\Sigma}$-defining relations for $\Pi \setminus \{ x \}$ and $Q \subseteq \Lambda^* \times \Lambda^*$.
          Since $\Gamma_D( D \setminus \{ r_x \} )$ is a subgraph of $\Gamma_D( D )$ by \cref{Lemma:App:Subgraph0} with $\Gamma_D( D )$ acyclic, then $\Gamma_D( D \setminus \{ r_x \} )$ is also acyclic.
          Since $|\Pi \setminus \{ x \}| = |\Pi| - 1 = k$ and $\Gamma_D( D \setminus \{ r_x \} )$ is acyclic, then by the inductive hypothesis, there exists a length $k$ sequence of \textbf{Gen($+$)} transformations from $\langle \Lambda \mid Q \rangle$ to $\langle \widehat{\Sigma} \mid R \cup D \rangle$.
           Furthermore, $\langle \Sigma \mid R \rangle \cong \langle \Lambda \mid Q \rangle$ by \textbf{Gen($+$)}.
           Then there exists a length $k + 1$ sequence of \textbf{Gen($+$)} transformations between $\langle \Sigma \mid D \rangle$ and $\langle \widehat{\Sigma} \mid R \cup D \rangle$.
    \end{enumerate}
    By the principle of induction, there exists a length $|\Pi|$ sequence of \textbf{Gen($+$)} transformations between $\langle \Sigma \mid D \rangle$ and $\langle \Sigma' \mid R \cup D \rangle$.
\end{proof}

\subsubsection{Tietze Transformations to Exchange Relations}

The goal of this section is to introduce a technique necessary to prove that derived generators can be eliminated via finite sequences of Tietze transformations.
Using the Tietze transformations discussed so far, it is possible to remove and introduce redundant relation.
In practice, one often wishes to replace a relation $q \approx w$ with a relation $q' \approx w'$, where neither $q \approx w$ nor $q' \approx w'$ is redundant without the other relation.
The following lemma gives a sufficient condition for when $q \approx w$ can be exchanged with $q' \approx w'$, and provides an upper bound on the number of Tietze transformations required to carry out the exchange.

\begin{lemma}
    \label{App:Lemma:TietzeEx}
    Let $G = \langle \Sigma \mid R \rangle$ be a presentation with $q \approx w$ in $R$ and $Q = R \setminus \{ q \approx w \}$.
    If there exists $q' \in \Sigma^*$ and $w' \in \Sigma^*$ such that $q \sim_Q q'$ and $w \sim_Q w'$, then there exists a finite sequence of Tietze transformations between $G$ and $\langle \Sigma \mid Q \cup \{ q' \approx w' \} \rangle$ of length between $1$ and $2$.
\end{lemma}

\begin{proof}
    Let $S = R \cup \{ q' \approx w' \}$ and $T = Q \cup \{ q' \approx w' \}$.
    If $q' \approx_R w'$, then $\langle \Sigma \mid R \rangle = \langle \Sigma \mid S \rangle$.
    Assume instead that $q \not \approx_R w$.
    Since $Q \subseteq R$, then $q \sim_R q'$ and $w \sim_R w'$.
    Then $q' \sim_R q$ by the symmetry of $( \sim_R )$.
    Since $q' \sim_R q$, $q \sim_R w$, and $w \sim_R w'$, then $q' \sim_R w'$ by transitivity.
    Then $\langle \Sigma \mid R \rangle \cong \langle \Sigma \mid S \rangle$ by a single \textbf{Rel($+$)} transformation.
    In either case, there exists a sequence of Tietze transformations between $\langle \Sigma \mid R \rangle$ and $\langle \Sigma \mid S \rangle$ of length at most one.
    Since $Q \subseteq T$, then $q \sim_T q'$ and $w \sim_T w'$.
    Then $w' \sim_T w$ by the symmetry of $( \sim_T )$.
    Since $q \sim_T q'$, $q' \sim_T w'$, and $w' \sim_T w$, then $q \sim_T w$ by the transitivity of $( \sim_S )$.
    Since $S \ne T$, then $\langle \Sigma \mid S \rangle \cong \langle \Sigma \mid T \rangle$ by a single \textbf{Rel($+$)} transformation.
    In conclusion, there exists a sequence of Tietze transformations between $\langle \Sigma \mid R \rangle$ and $\langle \Sigma \mid S \rangle$ of length at between $1$ and $2$.
\end{proof}

\subsubsection{Elimination of Derived Generators}

The goal of this section is to prove \cref{Thm:App:ElimDGens}.
The first step in this proof is to show that the set of relations can be replaced by $Q \cup D$ where $Q$ is a set of relations over the primitive generators and $D$ is the set of defining relations.
This follows by induction in \cref{Lemma:App:ElimDGen2}, where \cref{Lemma:App:ElimDGen1} is used to find new relations of the form, and then \cref{App:Lemma:TietzeEx} is used to exchange the relations.
The second step in this proof is to show that the derived generators can be eliminated through a finite sequence of Tietze transformations.
This follows by induction in \cref{Thm:App:ElimDGens}, where \cref{Lemma:App:DRoot} is used to find a derived generator upon which non other derived generator depends.
Such a generator is necessarily redundant, and may be eliminated via a \textbf{Rel($-$)} transformation.

\begin{lemma}
    \label{Lemma:App:ElimDGen1}
    Let $\Sigma \setminus \widehat{\Sigma}$ be an alphabet with $\Pi = \widehat{\Sigma} \setminus \Sigma$.
    If $D$ is a family of $\widehat{\Sigma}$-defining relations for $\Pi$ with $\Gamma_D( D )$ acyclic, then for each $w \in \Sigma^*$, there exists some $w' \in (\Sigma')^*$ such that $w \sim_D w'$.
\end{lemma}

\begin{proof}
    The proof follows by induction on the size of $\Pi$.
    \begin{enumerate}
    \item [--] \textbf{Base Case}.
          Assume that $|\Pi| = 0$ and $w \in \widehat{\Sigma}^*$.
          Then $\Pi = \varnothing$ and $w \in \Sigma^*$.
          Since $w \sim_D w$ by the symmetry of $(\sim_D)$, then there exists some $w' \in \Sigma^*$ such that $w \sim_D w'$.
    \item [--] \textbf{Inductive Hypothesis}.
          Assume that for some $k \in \mathbb{N}$, if $|\Pi| = k$, $\Gamma_D( D )$ is acyclic, and $w \in \widehat{\Sigma}^*$, then there exists some $w' \in \Sigma^*$ such that  $w \sim_D w'$.
    \item [--] \textbf{Inductive Step}.
          Assume that $|\Pi| = k + 1$ and $\Gamma_D( D )$ is acyclic.
          Since $\Pi$ is finite and $\Gamma_D( D )$ is acyclic, then by \cref{Lemma:App:DRoot}, there exists a $\widehat{\Sigma}$-defining relation $x \approx q$ in $D$ such that $D \setminus \{ r_x \}$ is a family of $(\widehat{\Sigma} \setminus \{ x \} )$-defining relations for $\Pi \setminus \{ x \}$.
          Then by \cref{Lemma:App:ElimDGen0}, there exists some $w' \in ( \widehat{\Sigma} \setminus \{ x \} )^*$ such that $w \sim_D w'$.
          Since $\Gamma_D( D \setminus \{ r_x \} )$ is a subgraph of $\Gamma_D( D )$ by \cref{Lemma:App:Subgraph0} with $\Gamma_D( D )$ acyclic, then $\Gamma_D( D \setminus \{ r_x \} )$ is also acyclic.
          Then $D \setminus \{ r_x \}$ is a family of $( \widehat{\Sigma} \setminus \{ x \} )$-defining relations for $\Pi \setminus \{ x \}$ with $\Gamma_D( D )$ acyclic and $w' \in ( \widehat{\Sigma} \setminus \{ x \} )^*$.
          Since $|\Pi \setminus \{ x \}| = |\Pi| - 1 = k$, then by the inductive hypothesis, there exists some $w'' \in \Sigma^*$ such that $w' \sim_{D'} w''$.
          Since $D' \subseteq D$, then $w' \sim_D w''$.
          Since $\sim_D$ is transitive, then $w \sim_D w''$.
          Then there exists a $w'' \in \Sigma^*$ such that $w \sim_D w''$.
    \end{enumerate}
    Since $\Gamma_D( D )$ is acyclic, then by the principle of induction, there exists a $w' \in \Sigma^*$ such that $w \sim_D w'$.
\end{proof}

\begin{lemma}
    \label{Lemma:App:ElimDGen2}
    Let $\Sigma \subseteq \widehat{\Sigma}$ be an alphabet with $\Pi = \widehat{\Sigma} \setminus \Sigma$ and $R \subseteq \widehat{\Sigma}^* \times \widehat{\Sigma}^*$ finite.
    If $D \subseteq R$ is a family of $\widehat{\Sigma}$-defining relations for $\Pi$ with $\Gamma_D( D )$ acyclic, then there exists a $Q \subseteq \Sigma^* \times \Sigma^*$ with $|Q| \le |R \setminus D|$ and a finite sequence of Tietze transformations between $\langle \widehat{\Sigma} \mid R \rangle$ and $\langle \widehat{\Sigma} \mid Q \cup D \rangle$ of length between $k$ and $2k$ where $k = |R| - |R \cap ( \Sigma^* \times \Sigma^* )| - |D|$.
\end{lemma}

\begin{proof}
    Let $S = R \cap ( \Sigma^* \times \Sigma^* )$ and $\overline{S} = R \setminus ( S \sqcup D )$.
    Then $R = S \sqcup D \sqcup \overline{S}$.
    The proof follows by induction on $|\overline{S}|$ in this decomposition.
    \begin{enumerate}
    \item[--] \textbf{Base Case}.
          Assume that $R$ decomposes into $S \sqcup D \sqcup \overline{S}$ with $|\overline{S}| = 0$.
          Then $\langle \widehat{\Sigma} \mid R \rangle = \langle 
          \widehat{\Sigma} \mid S \cup D \rangle$.
          Then there exists a length $0$ sequence of Tietze transformations between $\langle \widehat{\Sigma} \mid R \rangle$ and $\langle \widehat{\Sigma} \mid S \cup D \rangle$.
          Clearly $|S| = |R \setminus D|$.
    \item[--] \textbf{Inductive Hypothesis}.
          Assume that for some $k \in \mathbb{N}$, if $R$ decomposes as $S \sqcup D \sqcup \overline{S}$ with $|\overline{S}| = k$, then there exists some $Q \subseteq \Sigma^* \times \Sigma^*$ with $|Q| \le |R \setminus D|$ and a sequence of Tietze transformations between $\langle \widehat{\Sigma} \mid R \rangle$ and $\langle \widehat{\Sigma} \mid Q \cup D \rangle$ of length at most $2k$.
    \item[--] \textbf{Inductive Step}.
          Assume that $R$ decomposes into $S \sqcup D \sqcup \overline{S}$ with $|\overline{S}| = k + 1$.
          Let $q \approx w$ in $\overline{S}$.
          Then by \cref{Lemma:App:ElimDGen1}, there exists $q' \in \Sigma^*$ and $r' \in \Sigma^*$ such that $q \sim_D q'$ and $r \sim_D r'$.
          Then by \cref{App:Lemma:TietzeEx}, there exists a finite sequence of Tietze transformations between $\langle \widehat{\Sigma} \mid R \rangle$ and $\langle \widehat{\Sigma} \mid R' \rangle$ of length between $1$ and $2$, where $R' = R \cup \{ q' \approx w' \} \setminus \{ q \approx w \}$.
          Then $R'$ decomposes as $S' \sqcup D \sqcup \overline{S}'$ where $S' = S \cup \{ q' \approx r' \}$ and $\overline{S}' = \overline{S} \setminus \{ q \approx r \}$.
          Since $q \approx_S r$, then $|\overline{S}'| = |\overline{S}| - 1 = k$.
          Then the inductive hypothesis holds, and there exists some $Q \subseteq \Sigma^* \times \Sigma^*$ with $|Q| \le |R' \setminus D|$ and a sequence of Tietze transformations between $\langle \widehat{\Sigma} \mid R' \rangle$ and $\langle \widehat{\Sigma} \mid Q \cup D \rangle$ of length between $k$ and $2k$.
          Then there exists a sequence of Tietze transformations between $\langle \widehat{\Sigma} \mid R \rangle$ and $\langle \widehat{\Sigma} \mid Q \cup D \rangle$ of length between $k + 1$ and $2(k + 1)$.
          Since $|S'| \le |S|$ and $|\overline{S}'| = |\overline{S}|$, then $|Q| \le |S'| + |\overline{S}'| \le |S| + |\overline{S}| = |R \setminus D|$.
    \end{enumerate}
    By the principle of induction, there exists a $Q \subseteq \Sigma^* \times \Sigma^*$ with $|Q| \le |R \setminus D|$ and a sequence of Tietze transformations between $\langle \widehat{\Sigma} \mid R \rangle$ and $\langle \widehat{\Sigma} \mid Q \cup D \rangle$ of length between $|\overline{S}|$ and $2|\overline{S}|$.
\end{proof}

\begin{lemma}
    \label{Lemma:App:DRoot}
    If $D$ is a family of $\Sigma$-defining relations for $\Pi$ with $\Pi$ finite and $\Gamma_D( D )$ acyclic, then there exists a relation $x \approx w$ in $D$ such that $D \setminus \{ x \approx w \}$ is a family of $( \Sigma \setminus \{ x \} )$-defining relations for $\Pi \setminus \{ x \}$.
\end{lemma}

\begin{proof}
    Since $\Pi$ is finite and $\Gamma_D( D )$ is acyclic, then by \cref{Lemma:App:Acyc}, there exists some vertex $x \in \Pi$ such that $x$ has no parents in $\Gamma_D( D )$.
    Since $x$ is a vertex in $\Gamma_D( D )$, then there exists a $\Sigma$-defining relation $x \approx w$ in $D$.
    Let $D' = D \setminus \{ x \approx w \}$.
    Let $y \approx q$ be a $\Sigma$-defining relation in $D'$ with $n = |q|$.
    Let $k \in [n]$.
    Assume for the intent of contradiction that $q_k = x$.
    Then there exists some $q', q'' \in \Sigma^*$ such that $q = q' \cdot x \cdot q''$.
    Then $x$ is a child of $y$ by definition.
    However, $x$ has no parents by assumption.
    Then $q_k \ne x$ by contradiction.
    Since $k$ was arbitrary, then $q \in ( \Sigma \setminus \{ x \} )^*$.
    Since $y \approx q$ was arbitrary, then $D'$ is a family of $( \Sigma \setminus \{ x \} )$-defining relations for $\Pi \setminus \{ x \}$.
\end{proof}

\begin{theorem}
    \label{Thm:App:ElimDGens}
    Let $\Sigma' \subseteq \Sigma$ be an alphabet with $\Pi = \Sigma \setminus \Sigma'$ finite and $R \subseteq \Sigma^* \times \Sigma^*$ finite.
    If $D \subseteq R$ is a family of defining relations for $\Pi$ with $\Gamma_D( D )$ acyclic, then there exists a $Q \subseteq (\Sigma')^* \times (\Sigma')^*$ with $|Q| \le |R \setminus D|$ and a sequence of Tietze transformations between $\langle \Sigma \mid R \rangle$ and $\langle \Sigma' \mid Q \rangle$ of length between $n + k$ and $2n + k$ where $n = |R| - |R \cap ( ( \Sigma' )^* \times ( \Sigma' )^* )| - |D|$ and $k = |\Pi|$.
\end{theorem}

\begin{proof}
    By \cref{Lemma:App:ElimDGen2}, there exists a $Q \subseteq ( \Sigma' )^* \times ( \Sigma' )^*$ with $|Q| \le |R \setminus D|$ and a sequence of Tietze transformations between $\langle \Sigma \mid R \rangle$ and $\langle \Sigma \mid Q \cup D \rangle$ of length between $n$ and $2n$.
    The proof follows by induction on $|\Pi|$.
    \begin{enumerate}
    \item[--] \textbf{Base Case}.
          If $|\Pi| = 0$, then $|D| = 0$.
          Then $\langle \Sigma \mid Q \cup D \rangle = \langle \Sigma' \mid Q \rangle$.
          Then there exists a length $0$ sequence of Tietze transformations between $\langle \Sigma \mid Q \cup D \rangle$ and $\langle \Sigma' \mid Q \rangle$.
    \item[--] \textbf{Inductive Case}.
          Assume that for some $k \in \mathbb{N}$, if $|\Pi| = k$ and $\Pi_D( D )$ is acyclic, then there exists a length $k$ sequence of Tietze transformations between $\langle \Sigma \mid Q \cup D \rangle$ and $\langle \Sigma' \mid Q \rangle$.
    \item[--] \textbf{Inductive Step}.
          Assume that $|\Pi| = k + 1$ and $\Gamma_D( D )$ is acyclic.
          Since $\Pi$ is finite and $\Gamma_D( D )$ is acyclic, then by \cref{Lemma:App:Acyc} there exists some vertex $x \in \Pi$ such that $x$ has no parents in $\Gamma_D( D )$.
          Define $\Lambda = \Sigma \setminus \{ x \}$ and $D' = D \setminus \{ r_x \}$.
          Let $y \approx_{D'} w$ and assume for the intent of contradiction that $w \not \in \Lambda^*$.
          Then there exists $w', w'' \in \Sigma^*$ such that $w = w' \cdot x \cdot w''$.
          Then $x$ is a child of $y$ in $\Gamma_D( D )$.
          However, $x$ has no parents by assumption.
          Then $w \in \Lambda^*$ by contradiction.
          Since $y \approx_{D'} w$ was arbitrary, then $D' \subseteq \Lambda^* \times \Lambda^*$.
          Then $Q \cup D' \subseteq \Lambda^* \times \Lambda^*$.
          It follows that $\langle \Sigma \mid Q \cup D \rangle \cong \langle \Lambda \mid Q \cup D' \rangle$ by \textbf{Gen($-$)}.
          Furthermore, $D'$ is a family of defining relations for $\Pi \setminus \{ x \}$.
          Since $|\Pi \setminus \{ x \}| = k$ and $\Gamma_D( D' )$ is acyclic by \cref{Lemma:App:Subgraph0}, then there exists a length $k$ sequence of Tietze transformations between $\langle \Lambda \mid Q \cup D' \rangle$ and $\langle \Sigma' \mid Q \rangle$ by the inductive hypothesis.
          Then there exists a length $k + 1$ sequence of Tietze transformations between $\langle \Sigma \mid R \rangle$ and $\langle \Sigma' \mid Q \rangle$.
    \end{enumerate}
    Then by the principle of induction, there exists a length $k$ sequence of Tietze transformations between $\langle \Sigma \mid Q \cup D \rangle$ and $\langle \Sigma' \mid Q \rangle$.
    Then there exists a sequence of Tietze transformations between $\langle \Sigma \mid R \rangle$ and $\langle \Sigma' \mid Q \rangle$ of length between $n + k$ and $2n + k$.
\end{proof}

\subsubsection{The Derived Generator Graph for \texorpdfstring{$W( E_8 )$}{W(E8)}}

The derived generators for $W( E_8 )$ are $\Pi := \Sigma_D \setminus \{ \GateAt{X}{0}, \CGate{X}{0}{1},  \CCGate{X}{1}{2}, K_{1,2} \}$.
The defining relations $D \subseteq R_0$ for $\Pi \subseteq \Sigma_D$ are given by \cref{Rel:DGen:Start} through to \cref{Rel:DGen:End} in \cref{Fig:G0Rels}.
An illustration of the derived generator graph $\Gamma_D( D )$ can be found in \cref{Fig:App:Deps}.
Since this graph is acyclic, then \cref{Thm:App:IntroDGens} and \cref{Thm:App:ElimDGens} apply.
The derived generator graphs for $O(n,\D)$ have paths of length at most one, and are therefore trivially acyclic.

\subsection{Derivational Proofs and Tietze Transformation}
\label{Appendix:DerivationThm}

Assume that $G \cong \langle \Sigma \mid R \rangle$ with semantic interpretation $\interp{\cdot}_G$.
During proofs based on Tietze transformations, it is often necessary to find a sequence of \textbf{Rel($-$)} and \textbf{Rel($+$)} transformations between $\langle \Sigma \mid R \rangle$ and $\langle \Sigma \mid Q \rangle$.
For example, this case arises in \cref{Sec:E8}, where $R = R_{E8} \cup R_{D(E8)}$ and $Q = R_{E8(D)} \cup R_D$.
If $\interp{\cdot}_G$ induces an isomorphism and every relation $w \approx_Q w'$ satisfies $\interp{w}^* = \interp{w'}$, then using \cref{Thm:SemRewrite} there exists a sequence of Tietze transformations between $\langle \Sigma \mid R \rangle$ and $\langle \Sigma \mid R \rangle$.
Eliminating the relations in $R$ requires more care.
For example, if $r \in R$ is not derivable from $(R \cup Q) \setminus \{ r \}$, then $\langle \Sigma \mid R \cup Q \rangle$ is a proper quotient of $\langle \Sigma \mid Q \rangle$.
Instead, it must be shown that for each $w \approx_R w'$, it follows that $w \sim_Q w'$.
One way to approach this problem is to first derive some auxiliary relations $A$ from $Q$, and then use $Q \cup A$ to derive $R$.
However, transforming these derivations into a sequence of Tietze transformations is often tedious, and not well-aligned with the process of proof discovery.
On the other hand, if the derivations are not transformed into a valid sequence of Tietze transformations, then it is possible to obtain invalid proofs, such as those with cyclic derivations (see \cref{Ex:Cyclic}).

This section formalizes the ad-hoc proof technique described above,and identifies sufficient conditions for when such a family of derivations induces a valid sequence of Tietze transformations between $\langle \Sigma \mid R \rangle$ and $\langle \Sigma \mid Q \rangle$.
In the following definitions, $L( \Sigma ) = \mathbb{N} \times ( \Sigma^* \times \Sigma^* )$ will represent a set of indexed relations over $\Sigma$.
For example, let $( n, ( w, w' ) ) \in L(\Sigma)$.
The index $n$ in $( n, ( w, w' ) )$ indicates that $( n, ( w, w' ) )$ is a derived relation, and allows for multiple derivations of the same relation.
More concretely, if $n = 3$, $w = a \cdot b$, and $w' = x \cdot y \cdot z$, then $( n, ( w, w' ) )$ is the third derivation that yields $a \cdot b \approx x \cdot y \cdot z$.

\begin{definition}[Derivational Proof]
    A \emph{derivational proof in $\langle \Sigma \mid R \rangle$} is a subset $P \subseteq L( \Sigma ) \times (  L( \Sigma ) \cup R )^{*}$ which satisfies the following conditions.
    \begin{itemize}
    \item[--] \textbf{Indexed}.
              For all $( \ell, d ) \in P$ and $( \ell', d' ) \in P$ distinct, $\ell \ne \ell'$.
    \item[--] \textbf{Well-founded}.
          For all $( \ell, d ) \in P$ and $k \in \{ 1, 2, \ldots, |d| \}$, either $d_k \in R$ or $d_k \in L(\Sigma)$ and there exists some $d' \in ( L(\Sigma) \cup R)^*$ such that $( d_k, d' ) \in P$.
    \item[--] \textbf{Valid}.
          For all $( \ell, d ) \in P$ with $( n, ( w, w' ) ) = \ell$ and $m = |d|$, there exists some $v \in (\Sigma^*){}^{m+1}$ such that $v_1 = w$, $v_{m+1} = w'$, and for all $k \in [m]$ either $d_k \in R$ and $v_k \xrightarrow{d_k} v_{k+1}$ or $( n', r ) = d_k$ and $v_k \xrightarrow{r} v_{k+1}$.
    \end{itemize}
    A set $Q \subseteq \Sigma^{*} \times \Sigma^{*}$ is \emph{entailed by $P$}, written $P \models Q$, if $Q \subseteq R \cup \{ r \mid ( ( n, r ), d ) \in P \}$.
\end{definition}

\begin{definition}[Proof Substitution]
    Let $P$ be a proof in $\langle \Sigma \mid R \rangle$.
    If $( \ell, d ) \in P$, $d \in ( L( \Sigma ) \cup R )^{*}$, and $P' = ( P \setminus \{ ( \ell, d ) \} ) \cup \{ ( \ell, d' ) \}$ is a proof for $\langle \Sigma \mid R \rangle$, then we say that $P'$ is a \emph{substitution of $P$ by $d'$ at $\ell$}, written $P[ \ell \mapsto d' ]$.
\end{definition}

\begin{definition}[Derivation Graph]
    Let $P$ be a proof in $\langle \Sigma \mid R \rangle$.
    The \emph{derivation graph for $P$} is the digraph $\Gamma_D( P ) = ( V, E )$ such that $V = \{ \ell \mid ( \ell, d ) \in P \}$ and $( \ell, \ell' ) \in E$ if and only if there exists $( \ell, d ) \in P$  and $k \in \{ 1, 2, \ldots, |d| \}$ such that $d_k = \ell'$.
\end{definition}

\begin{example}[Derivations and Substitutions]
    \label{Ex:DerivsSubs}
    Consider $G = \langle x, y \mid x^2 \approx \epsilon, y^2 \approx \epsilon \rangle$.
    It is not hard to show $x \cdot y \cdot x \cdot y^2 \cdot x \cdot y \cdot x \sim \epsilon$.
    However, it helps to first prove that $x \cdot y^2 \cdot x \sim \epsilon$.
    This can be written as a derivation proof.
    There will be two derivations, with labels $\ell = ( 0, ( x \cdot y^2 \cdot x, \epsilon ) )$ and $\ell' = ( 0, ( x \cdot y \cdot x \cdot y^2 \cdot x \cdot y \cdot x, \epsilon ) )$ respectively.
    Associated with $\ell$ and $\ell'$ are two derivations $d$ and $d'$, defined as follows.
    Let $r = ( x^2, \epsilon )$ and $r' = ( y^2, \epsilon )$.
    \begin{align*}
        (d):&\;\; x \cdot y \cdot y \cdot x \xrightarrow{r'} x \cdot x \xrightarrow{r} \epsilon
        \\
        (d'):&\;\;  x \cdot y \cdot x \cdot y \cdot y \cdot x \cdot y \cdot x \xrightarrow{\ell} x \cdot y \cdot y \cdot x \xrightarrow{\ell} \epsilon
    \end{align*}
    These pieces can be assembled into a derivational proof $P = \{ ( \ell, d ), ( \ell', d' ) \}$.
    This proof is indexed, since $\ell$ and $\ell'$ each appear exactly once as labels.
    This proof is well-formed, since $r$, $r'$, $d$, and $d'$ are all in-scope.
    The proof is valid, since each step of each derivation follows.
    Of course, it is possible to expand out $d'$ using the steps of $d$.
    This yields a new derivation $d''$ defined as follows.
    \begin{equation*}
        (d''):\;\;
        x \cdot y \cdot x \cdot y \cdot y \cdot x \cdot y \cdot x
        \xrightarrow{r'}
        x \cdot y \cdot x \cdot x \cdot y \cdot x
        \xrightarrow{r}
        x \cdot y \cdot y \cdot x
        \xrightarrow{r'}
        x \cdot x
        \xrightarrow{r}
        \epsilon
    \end{equation*}
    This new derivational proof corresponds to the substitution $P[ \ell' \mapsto d'' ]$.
    In this new proof, the derivation of $\ell'$ no longer relies on the lemma $\ell$.
    Later in this section, lemma eliminating substitutions will be used to extract Tietze transformations from derivational proofs with acyclic derivation graphs.
    \qed
\end{example}

\begin{example}[Cyclic Derivations]
    \label{Ex:Cyclic}
    Consider $G = \langle x, y \mid x^2 \approx \epsilon, y^2 \approx \epsilon \rangle$.
    It is not hard to show that $G \cong \mathbb{Z}_2 \star \mathbb{Z}_2$, where $( \star )$ denotes the free product of groups.
    It follows that $G$ is non-abelian.
    In particular, $\pi_G( x \cdot y ) \ne \pi_G( y \cdot x )$.
    Now, consider the proof $P = \{ ( ( 0, ( x \cdot y, y \cdot x ) ), d ), ( ( 0, ( x \cdot y^3 , y \cdot x \cdot y ) ), d' ) \}$, where $d$ and $d'$ are defined as follows.
    Let $r = ( y^2, \epsilon )$, $\ell = ( 0, ( x \cdot y, y \cdot x ) )$ and $\ell' = ( 0, ( x \cdot y^3, y \cdot x \cdot y^2 ) )$. 
    \begin{align*}
        (d):&\;\; x \cdot y \xleftarrow{r} x \cdot y^3 \xrightarrow{\ell'} y \cdot x \cdot y^2 \xrightarrow{r} y \cdot x
        \\
        (d'):&\;\; x \cdot y^3 \xrightarrow{\ell} y \cdot x \cdot y^2 
    \end{align*}
    This proof is indexed, well-formed, and valid.
    However, these derivations suggest that $x \cdot y \sim y \cdot x$.
    The problem in this proof is that $d$ depends on $d'$ and $d'$ depends on $d$.
    In other words, the proof is self-referential.
    It will be shown later in this section that if $\Gamma_D(P)$ is acyclic, then $P$ is not self-referential.
    This motivates the requirement that $\Gamma_D(P)$ is acyclic throughout the rest of this section.
    \qed
\end{example}

\subsubsection{Substitutions and Derivation Graphs}

This section relates derivational proofs and substitutions to the derivation graphs they induce.
When a derivation $( \ell, d )$ depends on a derivation $( \ell', d' )$, we say that $\ell'$ is a lemma for $\ell$.
In \cref{Lemma:App:SucFreeDeriv}, lemma-free derivations are characterized by their vertices in a derivation graph.
Similarly, \cref{Lemma:App:Edgeless} characterizes lemma-free proofs in terms of their derivation graphs.
As expected, \cref{Lemma:App:Subgraph2} shows that if a substitution only eliminates the use of lemmas, such as in \cref{Ex:DerivsSubs}, then the resulting derivation graph is a subgraph of the original graph.
To this end, \cref{Lemma:App:DerivRewrite} provides sufficient conditions for a valid substitution.
Together, these four lemmas give a graph-theoretic characterization of the lemma substitution in \cref{Ex:DerivsSubs}.

\begin{lemma}
    \label{Lemma:App:SucFreeDeriv}
    Let $P$ be a proof for $\langle \Sigma \mid R \rangle$.
    If $( \ell, d ) \in P$ and $\ell$ has no children in $\Gamma_D( P )$, then $d \in R^*$.
\end{lemma}

\begin{proof}
    Let $k \in \{ 1, 2, \ldots, |d| \}$.
    Assume for the intent of contradiction that $d_k \not \in R^{*}$.
    Then $( \ell, d_k )$ is an edge in $\Gamma_D( P )$.
    However, $\ell$ has no children in $\Gamma_D( P )$.
    Then $d_k \in R$.
    Since $k$ was arbitrary, then $d \in R^{*}$.
\end{proof}

\begin{lemma}
    \label{Lemma:App:Edgeless}
    If $P$ is a proof for $\langle \Sigma \mid R \rangle$, then $\Gamma_D( P )$ is edgeless if and only if $d \in R^{*}$ for all $( \ell, d ) \in P$.
\end{lemma}

\begin{proof}
    Let $\Gamma_D( P ) = ( V, E )$.
    Consider the contrapositive statement.
    Then there exists an $( \ell, \ell' ) \in E$.
    This is true if and only if there exists a $( \ell, d ) \in P$ and $k \in \{ 1, 2, \ldots, |d| \}$ such that $d_k = \ell'$.
    This is true if and only if there exists an $( \ell, d ) \in P$ such that $d \not \in R^{*}$.
\end{proof}

\begin{lemma}
    \label{Lemma:App:Subgraph2}
    Let $P$ be a proof for $\langle \Sigma \mid R \rangle$ and $( V, E ) = \Gamma_D( P )$.
    If $( \ell, d ) \in P$ and there exists some $d' \in R^*$ such that $P' = P[ l \mapsto d' ]$ is also a proof for $\langle \Sigma \mid R \rangle$, then $\Gamma_D( P' ) = ( V, E_\ell )$ where $E_\ell = \{ ( \ell', \ell'' ) \in E \mid \ell \ne \ell' \}$.
\end{lemma}

\begin{proof}
    Let $\Gamma_D( P' ) = ( V', E' )$.
    \begin{enumerate}
    \item[--] \textbf{Vertices ($\subseteq$)}.
           Let $\ell' \in V'$.
           Then there exists some $\delta \in ( L(\Sigma) \cup R )^{*}$ such that $( \ell', \delta ) \in P'$.
           Then either $( \ell', \delta ) \in P$ or $( \ell', \delta ) = ( \ell, d' )$.
           If $( \ell', \delta ) \in P$, then $\ell \in V$.
           If $( \ell', \delta ) = ( \ell, d' )$, then $\ell' \in V$ since $( \ell, d ) \in P$.
           In either case $\ell' \in V$.
           Since $\ell'$ was arbitrary, then $V' \subseteq V$.
    \item[--] \textbf{Vertices ($\supseteq$)}.
           Let $\ell' \in V$.
           Then there exists some $\delta \in ( L(\Sigma) \cup R )^{*}$ such that $( \ell', \delta ) \in P$.
           Then either $\ell' = \ell$ or $\ell' \ne \ell$.
           If $\ell' = \ell$, then $\ell' \in V'$ since $( \ell', d' ) \in P'$.
           If $\ell' \ne \ell$, then $\ell' \in V'$ since $( \ell', \delta ) \in P \setminus \{ ( \ell, d ) \}$.
           In either case, $\ell' \in V$,
           Since $\ell'$ was arbitrary, then $V' \subseteq V$.
    \item[--] \textbf{Edges ($\subseteq$)}.
           Let $( \ell', \ell'' ) \in E'$.
           Then there exists some $\delta \in ( L(\Sigma) \cup R)^*$ and $k \in \{ 1, 2, \ldots, |\delta| \}$ such that $( \ell', \delta ) \in P'$ and $\delta_k = \ell''$.
           Then $\delta \not \in R^*$.
           Consequently, $\delta \ne d'$.
           Then $\ell \ne \ell'$, since $P'$ is indexed.
           Then $( \ell', \delta ) \in P$.
           Consequently, $( \ell', \ell'' ) \in E$.
           Since $\ell' \ne \ell$, then $( \ell', \ell'' ) \in E_\ell$.
           Since $( \ell', \ell'' )$ was arbitrary, then $E' \subseteq E_\ell$.
    \item[--] \textbf{Edges ($\supseteq$)}.
           Assume that $( \ell', \ell'') \in E_\ell$.
           Then $\ell' \ne \ell$ and $( \ell', \ell'' ) \in E$.
           Then there exists some $\delta \in ( L(\Sigma) \cup R)^*$ and $k \in \{ 1, 2, \ldots, |\delta| \}$ such that $( \ell', \delta ) \in P$ and $\delta_k = v$.
           Since $\ell' \ne \ell$, then $( \ell', \delta ) \in P'$.
           Then $( \ell', \ell'' ) \in E'$.
           Since $( \ell', \ell'' )$ was arbitrary, then $E_\ell \subseteq E'$.
    \end{enumerate}
    Then $\Gamma_D( P' ) = ( V, E_\ell )$.
\end{proof}

\begin{lemma}
    \label{Lemma:App:DerivRewrite}
    Let $P$ be a proof for $\langle \Sigma \mid R \rangle$ and $V = \{ \ell \mid ( \ell, d ) \in P \}$.
    If $( \ell, d ) \in P$, $d' \in ( V \cup R )^{*}$, and $( \ell, d' )$ is valid, then $P[ \ell \mapsto d' ]$ is a proof for $\langle \Sigma \mid R \rangle$.
\end{lemma}

\begin{proof}
    Let $P' = [ \ell \mapsto d' ]$.
    It must be shown that $P'$ is indexed, well-formed, and valid.
    \begin{enumerate}
        \item[--] \textbf{Indexed}.
                  Let $( \ell', \delta ) \in P'$ and $( \ell'', \delta' ) \in P'$ distinct.
                  Without loss of generality, assume $\ell' \ne \ell$.
                  Then $( \ell', \delta ) \in P$.
                  Now, these are two cases to consider, depending on whether $\ell'' = \ell$.
                  If $\ell'' = \ell$, then $( \ell'', d ) \in P$ and $\ell' \ne \ell''$ since $P$ is indexed.
                  If $\ell'' \ne \ell$, then $( \ell'', \delta' ) \in P$ and $\ell' \ne \ell''$ since $P$ is indexed.
                  In either case $\ell' \ne \ell'$.
                  Since $( \ell', \delta )$ and $( \ell'', \delta' )$ were arbitrary, then $P'$ is indexed.
        \item[--] \textbf{Well-formed}.
                  Let $( \ell', \delta ) \in P'$.
                  There are two cases to consider, depending on whether $\ell' = \ell$.
                  \begin{enumerate}
                  \item[--]
                        If $\ell' = \ell$, then $\delta = d'$, since $P'$ is indexed.
                        Then $\delta \in ( V \cup R )^*$.
                        Then for all $k \in \{ 1, 2, \ldots, |\delta| \}$, either $\delta_k \in R$ or $\delta_k \in V$.
                        If $\delta_k \in V$ and $\delta_k = \ell$, then $( \delta_k, d' ) \in P'$.
                        If $\delta_k \in V$ and $\delta_k \ne \ell$, then there exists some $\delta' \in ( L(\Sigma) \cup R )^*$ such that $( d_k, \delta' ) \in P$.
                        Since $\delta_k \ne \ell$, then $( \delta_k, \delta' ) \in P'$.
                        In either case, if $\delta_k \in V$, then there exists some $\delta' \in ( L(\Sigma) \cup R )^*$ such that $( d_k, d'' ) \in P'$
                        Since $k$ was arbitrary, then $( \ell', \delta )$ is well-formed.
                  \item[--]
                        If $\ell' \ne \ell$, then $( \ell', \delta ) \in P$.
                        Let $k \in \{ 1, 2, \ldots, |\delta| \}$.
                        Since $P$ is well-formed, then either $\delta_k \in R$ or there exists some $\delta' \in ( L(\Sigma) \cup R)^*$ such that $( \delta_k, \delta' ) \in P$.
                        There are two cases to consider, depending on whether $\delta_k = \ell$.
                        If $\delta_k = \ell$, then $( \ell, d' ) \in P'$.
                        If $\delta_k \ne \ell$, then $( \delta_k, \delta' ) \in P'$.
                        In either case, there exists some $\delta' \in ( L(\Sigma) \cup R)^*$ such that $( \delta_k, \delta' ) \in P'$.
                        Since $k$ was arbitrary, then $( \ell', \delta )$ is well-formed.
                  \end{enumerate}
                  In either case, $( \ell', \delta )$ is well-formed.
                  Since $( \ell', \delta )$ was arbitrary, then $P'$ is well-formed.
        \item[--] \textbf{Valid}.
                  Let $( \ell', \delta ) \in P'$.
                  Now these are two cases to consider, depending on whether $\ell' = \ell$.
                  If $\ell' = \ell$, then $\delta = d'$ since $P'$ is indexed, and consequently $\delta$ is valid by assumption.
                  If $\ell' \ne \ell$, then $( \ell', \delta ) \in P$ and $\delta$ is valid by the validity of $P$.
                  In either case, $( \ell', \delta )$ is valid.
                  Since $( \ell', \delta )$ was arbitrary, then $P'$ is valid.
    \end{enumerate}
    Then $P'$ is a proof for $\langle \Sigma \mid R \rangle$.
\end{proof}

\subsubsection{From Derivational Proofs to Tietze Transformations}

Let $\langle \Sigma \mid R \rangle$ be a monoid presentation.
The goal of this section is to prove \cref{Thm:App:Derivation}, which states that acyclic derivational proofs are sound for the isomorphism of finite monoid presentations.
The completeness of acyclic derivational proofs follows immediately from the fact that every statement of the form $w \sim_R w'$ corresponds to at least one finite derivation.
The proof proceeds as follows.
First, \cref{Lemma:App:Construct0} shows that the lemma substitutions outlined in \cref{Ex:DerivsSubs} preserve the structure of derivational proofs.
This is used in \cref{Lemma:App:Construct} to show that if a derivation $( \ell, d )$ in a proof $P$ depends on a lemma $( \ell', d' )$ which follows directly from $R$, then $d$ can be rewritten so that all dependencies on $\ell'$ are removed without introducing any new dependencies.
This is used repeated in \cref{Lemma:App:ProofRewrite} to show that any derivation $( \ell, d )$ which depends only on lemmas which follow directly from $R$, can be rewritten to also follow directly from $R$.
This is extended in \cref{Lemma:App:Derivation} to show that any finite and acyclic derivational proof $P$ can be written into a proof $P'$ such that every derivation depends only on $R$.
Then each derivation in $P'$ follows from $( \sim_R )$, and \cref{Thm:App:Derivation} follows immediately by induction.

\begin{lemma}
    \label{Lemma:App:Construct0}
    Let $P$ be a proof for $\langle \Sigma \mid R \rangle$.
    If $( \ell, d )$ and $( \ell', d' )$ are derivations in $P$ and there exists $\delta, \delta' \in ( L(\Sigma) \cup R )^*$ such that $d = \delta \cdot \ell' \cdot \delta'$, then $P[ \ell \mapsto \delta \cdot d' \cdot \delta' ]$ is a proof for $\langle \Sigma \mid R \rangle$.
\end{lemma}

\begin{proof}
    Let $L_P = \{ \ell \mid ( \ell, d ) \in P \}$, $d'' = \delta \cdot d' \cdot \delta'$, and $( x, ( q, r ) ) = \ell$.
    Since $\delta, d', \delta' \in (L_p \cup R)^{*}$, then $d'' \in (L_p \cup R)^{*}$.
    Since $P$ is valid, then there exists $v \in ( \Sigma^{*} )^{n+1}$ such that $v_1 = q$, $v_{n+1} = r$, and $v_k \xrightarrow{d_k} v_{k+1}$ for all $k \in[n]$, where $n = |d|$.
    Let $m = |\delta|$.
    Since $P$ is valid and $v_m \xrightarrow{\ell'} v_{m+1}$, then there exists some $u \in ( \Sigma^{*} )^{s+1}$ such that $u_1 = v_m$, $u_{s+1} = v_{m+1}$, and $u_k \xrightarrow{d'_k} u_{k+1}$ for all $k \in [s]$, where $s = |d'|$.
    Define $v' = ( v_1, \ldots, v_m, u_1, \ldots, u_s, v_{m+1}, \ldots, v_n )$.
    Let $k \in \{ 1, 2, \ldots, |d''| \}$.
    There are five cases to consider.
    \begin{enumerate}
        \item[--]
              If $k < m$, then $v'_k = v_k$, $v'_{k+1} = v_{k+1}$, and $d''_k = \delta_k = d_k$.
              Then $v'_k \xrightarrow{d''_k} v'_{k+1}$.
        \item[--]
              If $k = m$, then $v'_k = v_{m}$, $v'_{k+1} = u_1 = v_{m+1}$, and $d''_k = \delta_k$.
              Then $v'_k \xrightarrow{d''_k} v'_{k+1}$.
        \item[--]
              If $m < k < m + s$, then $v'_k = u_{k - m}$, $v'_{k+1} = u_{k - m + 1}$, and $d''_k = d'_{k - m}$.
              Then $v'_k \xrightarrow{d''_k} v'_{k+1}$.
        \item[--]
              If $k = m + s$, then $v'_k = u_s$, $v'_{k+1} = v_{m+1} = u_{s+1}$, and $d''_k = d'_s$.
              Then $v'_k \xrightarrow{d''_k} v'_{k+1}$.
        \item[--]
              If $k > m + s$, then $v'_k = u_{k - s}$, $v'_{k+1} = u_{k + 1 - s}$, and $d''_k = \delta'_{k - m - s}$.
              Then $v'_k \xrightarrow{d''_k} v'_{k+1}$.
    \end{enumerate}
    In each case, $v'_k \xrightarrow{d''_k} v'_{k+1}$.
    Since $k$ was arbitrary, then $( \ell, d'' )$ is valid.
    Then $P[ \ell \mapsto d'' ]$ is a proof by \cref{Lemma:App:DerivRewrite}.
\end{proof}

\begin{lemma}
    \label{Lemma:App:Construct}
    Let $P$ be a proof for $\langle \Sigma \mid R \rangle$.
    If $( \ell, \ell' )$ is a maximal path rooted at $\ell$ in $\Gamma_D( P )$, then there exists a $\hat{d} \in ( L(\Sigma) \cup R )^{*}$ such that $P' = P[ \ell \mapsto \hat{d} ]$ is a proof for $\langle \Sigma \mid R \rangle$ and $\ell$ has one less child in $\Gamma_D( P' )$.
\end{lemma}

\begin{proof}
    Let $f: ( L(\Sigma) \cup R )^{*} \to \mathbb{N}$ count the number of occurrences of $\ell'$ in a derivation.
    Given a proof $Q$ for $\langle \Sigma \mid R \rangle$, let $C_{Q}: L(\Sigma) \to \mathcal{P}( L )$ map each $\ell \in L(\Sigma)$ to its children in $\Gamma_D( Q )$.
    Since $( \ell, \ell' )$ is a maximal path, then $\ell'$ has no children in $\Gamma_D( P )$.
    The proof follows by induction on the number of occurrences of $\ell'$ in the derivation.
    \begin{enumerate}
    \item[--] \textbf{Base Case}.
          Let $\hat{d} \in ( L(\Sigma) \cup R)^{*}$ with $P' = P[ \ell \mapsto \hat{d} ]$ a proof for $\langle \Sigma \mid R \rangle$ and $C_P( \ell ) = C_{P'}( \ell ) \cup \{ \ell' \}$.
          Assume that $f( \hat{d} ) = 0$.
          Then $\hat{d}_j \ne \ell'$ for all $k \in [ n ]$, where $n = |\hat{d}|$.
          Then $( \ell, \ell' ) \not \in \Gamma_D( P' )$.
          Then $\ell' \not \in C_{P'}( \ell )$.
          Then $C_{P'}( \ell ) = C_P( \ell ) \setminus \{ \ell' \}$.
          Since $\ell' \in C_P( \ell )$, then $|C_{P'}( \ell )| = |C_{P}( \ell )| - 1$.
    \item[--] \textbf{Inductive Hypothesis}.
          Let $d' \in ( L(\Sigma) \cup R )^{*}$ such that $P' = P[ \ell \mapsto d' ]$ a proof for $\langle \Sigma \mid R \rangle$ and $C_P( \ell ) = C_{P'}( \ell ) \cup \{ \ell' \}$.
          Assume that for some $k \in \mathbb{N}$, if $f( d' ) = k$, then exists a $\hat{d} \in ( L(\Sigma) \cup R )^{*}$ such that $P'' = P[ \ell \mapsto \hat{d} ]$ is a proof for $\langle \Sigma \mid R \rangle$ and $|C_{P''}( \ell )| = |C_{P}( \ell )| - 1$.
    \item[--] \textbf{Inductive Step}.
          Under the conditions of the inductive hypothesis, assume that $f( d' ) = k + 1$.
          Then there exists some $\delta, \delta' \in ( L(\Sigma) \cup R )^{*}$ such that $d' = \delta \cdot \ell' \cdot \delta'$. 
          Since $\ell'$ is a vertex in $\Gamma_D( P )$, then there exists some $d'' \in ( L(\Sigma) \cup R)^*$ such that $( \ell, d'' ) \in P$.
          Define $\hat{d} = \delta \cdot d'' \cdot \delta'$.
          By \cref{Lemma:App:Construct0}, $P'' = P[ l \mapsto \hat{d} ]$ is a proof for $\langle \Sigma \mid R \rangle$.
          Let $\ell'' \in C_{P}( \ell ) \setminus \{ \ell' \}$.
          Then there exists some $k \in \{ 1, 2, \ldots, |\delta| \}$ such that $d_k = \ell''$.
          Since $\ell'' \ne \ell$, then without loss of generality $k \le |\delta|$ and $\hat{d}_k = \delta_k = \ell''$.
          Then $\ell'' \in C_{P''}( \ell )$.
          Since $\ell''$ was arbitrary, then $C_{P}( \ell ) \setminus \{ \ell' \} \subseteq C_{P''}( \ell )$.
          Next, let $\ell'' \in C_{P''}( \ell )$.
          Then there exists some $k \in \{ 1, 2, \ldots, |\hat{d}| \}$ such that $\hat{d}_k = \ell''$.
          Since $d \in ( R^{*} )$, then without loss of generality $k \le |\delta|$ and $\delta_k = \hat{h}_k = \ell''$.
          Then $\ell'' \in C_P( \ell )$.
          Since $\ell''$ was arbitrary, then $C_{P''}( \ell ) \subseteq C_{P}( \ell )$.
          Then $C_{P}( \ell ) = C_{P''}( \ell ) \cup \{ \ell' \}$.
          Since $k + 1 = f( d' ) = f( \delta ) + 1 + f( \delta' )$, then $f( \hat{d} ) = f( \delta ) + f( \delta' ) = k$.
          Then by the inductive hypothesis, there exists some $\hat{d}' \in ( L(\Sigma) \cup R )^{*}$ such that $Q = P'[ \ell \mapsto \hat{d}' ]$ is a proof for $\langle \Sigma \mid R \rangle$ and $|C_{Q}( \ell )| = |C_{P'}( \ell )| - 1 = |C_{P}( \ell )| - 1$.
          Since $Q = P[ \ell \mapsto \hat{d}' ]$ by definition, then the inductive step holds.
    \end{enumerate}
    It follows by definition that $d \in (L(\Sigma) \cup R)^*$ and $P = P[ \ell \mapsto d ]$.
    Then by the principle of induction, there exists a $\hat{d} \in ( L(\Sigma) \cup R )^{*}$ such that $P' = P[ \ell \mapsto \hat{d} ]$ is a proof for $\langle \Sigma \mid R \rangle$ and $|C_{P''}( \ell )| = |C_{P}( \ell )| - 1$.
\end{proof}

\begin{lemma}
    \label{Lemma:App:ProofRewrite}
    Let $P$ be a proof for $\langle \Sigma \mid R \rangle$.
    If $( \ell, d ) \in P$ and all paths rooted at $\ell$ in $\Gamma_D( P )$ have length at most one, then there exists a $\hat{d} \in R^{*}$ such that $P[ \ell \mapsto \hat{d} ]$ is a proof for $\langle \Sigma \mid R \rangle$.
\end{lemma}

\begin{proof}
    Let $\Gamma_D( P ) = ( V, E )$.
    Since $|V| = |P|$ and all paths rooted at $\ell$ have length at most one, then the number of paths rooted at $\ell$ in $\Gamma_D( P )$ is finite.
    The proof follows by induction on the number of paths rooted at $\ell$ in $\Gamma_D( P )$.
    \begin{enumerate}
    \item[--] \textbf{Base Case}.
          Assume that $\Gamma_D( P )$ has zero paths rooted at $\ell$.
          Then $d \in R^{*}$ by \cref{Lemma:App:SucFreeDeriv}.
          Then $P = P[ l \mapsto d ]$ with $d \in R^*$.
    \item[--] \textbf{Inductive Hypothesis}.
          Assume that for some $k \in \mathbb{N}$, if $\Gamma_D( P )$ has $k$ paths rooted at $\ell$, then there exists a $\hat{d} \in R^{*}$ such that $P[ \ell \mapsto \hat{d} ]$ is a proof for $\langle \Sigma \mid R \rangle$.
    \item[--] \textbf{Inductive Step}.
          Assume that $\Gamma_D( P )$ has $k + 1$ paths rooted at $\ell$.
          Then there exists at least one path rooted at $\ell$ in $\Gamma_D( P )$.
          Since all paths rooted at $\ell$ in $\Gamma_D( P )$ have length one, then there exists some path $( \ell, \ell' )$ in $\Gamma_D( P )$ such that $\ell'$ has no children in $\Gamma_D( P )$.
          Then by \cref{Lemma:App:Construct}, there exists some $d' \in ( L(\Sigma) \cup R )^*$ such that $P' = P[ \ell \mapsto d' ]$ is a proof for $\langle \Sigma \mid R \rangle$ and $\ell$ has $k$ children in $\Gamma_D( P' )$.
          Since all paths rooted at $\ell$ have length one, then $\Gamma_D( P' )$ has $k$ paths rooted at $\ell$.
          Then by the inductive hypothesis, then there exists a $\hat{d} \in R^{*}$ such that $P[ \ell \mapsto \hat{d} ]$ is a proof for $\langle \Sigma \mid R \rangle$.
          Then the inductive step holds.
    \end{enumerate}
    Then by the principle of induction, there exists a $\hat{d} \in R^{*}$ such that $P[ \ell \mapsto \hat{d} ]$ is a proof for $\langle \Sigma \mid R \rangle$.
\end{proof}

\begin{lemma}
    \label{Lemma:App:Derivation}
    Let $P$ be a finite proof for $\langle \Sigma \mid R' \rangle$ with $R' \subseteq R \subseteq \Sigma^{*} \times \Sigma^{*}$.
    If $P \models R$ and $\Gamma_D( P )$ is acyclic, then there exists a proof $P'$ for $\langle \Sigma \mid R' \rangle$ such that $|P'| = |P|$, $P' \models R$, and $d \in (R')^{*}$ for all $( \ell, d ) \in P'$.
\end{lemma}

\begin{proof}
    Let $f: L(\Sigma) \times ( L(\Sigma) \cup R' )^* \to \mathbb{N}$ count the number of vertices with children in the derivation graph of a proof.
    Since $\Gamma_D( P )$ has $|P|$ vertices, then $f( P ) \le |P|$.
    Since $P$ is finite, then $f( P )$ is also finite.
    The proof follows by induction on $f( P )$.
    \begin{enumerate}
        \item[--] \textbf{Base Case}.
              Assume that $f( P ) = 0$.
              Then there are no edges in $\Gamma_D( P )$.
              Then by \cref{Lemma:App:Edgeless}, $d \in ( R' )^{*}$ for all $( \ell, d ) \in P$.
        \item[--] \textbf{Inductive Hypothesis}
              Let $Q$ be a proof for $\langle \Sigma \mid R' \rangle$.
              Assume that for some $k \in \mathbb{N}$, if $f( Q ) = k$ with $P \models R$ and $\Gamma_D( Q )$ is acyclic, then there exists a proof $P'$ for $\langle \Sigma \mid R \rangle$ such that $|P'| = |Q|$, $P' \models R$, and $d \in (R')^{*}$ for all $( \ell, d ) \in P'$.
        \item[--] \textbf{Inductive Step}.
              Let $Q$ be a proof for $\langle \Sigma \mid R' \rangle$.
              Assume that $f( Q ) = k + 1$ with $Q \models R$ and $\Gamma_D( Q )$ acyclic.
              Since $f( Q ) > 0$, then there exists at least one edge $( \ell, \ell' )$ in $\Gamma_D( Q )$.
              By \cref{Lemma:App:Acyc}, there exists some path $( \ell_0, \ldots, \ell_n )$ in $\Gamma_D( Q )$ such that $\ell_0 = \ell$ every path rooted at $\ell$ has length at most $n$.
              Since $( \ell, \ell' )$ is a path of length one in $\Gamma_D( Q )$ rooted at $\ell$, then $n \ge 1$.
              Assume for the intent of contradiction that there exists a path of length at least $2$ rooted at $\ell_{n-1}$.
              Then there exists a path $( \ell_{n-1}, x, y )$ in $\Gamma_D( Q )$.
              Then $( \ell_0, \ldots, \ell_{n-1}, x, y )$ is a path of length $n+1$ in $\Gamma_D( Q )$ rooted at $\ell$.
              However, all paths rooted at $\ell$ have length at most $n$.
              Therefore, all paths rooted at $\ell_{n-1}$ have length at most one.
              Since $\ell_{n-1}$ is a vertex in $\Gamma_D( Q )$, then there exists some $d \in ( L(\Sigma) \cup R ')^*$ such that $( \ell_{n-1}, d ) \in Q$.
              Then by \cref{Lemma:App:ProofRewrite}, there exists some $d' \in ( R' )^*$ such that $P' = Q[ \ell_{n-1} \mapsto d' ]$ is a proof for $\langle \Sigma \mid R' \rangle$.
              Let $( z_{n-1}, r_{n-1} ) = \ell_{n-1}$.
              Since $Q \models R$ with respect to $\langle \Sigma \mid R' \rangle$, it follows that $R \subseteq R' \cup \{ r \mid ( ( m, r ), d ) \in P \}$.
              Then,
              \[
              R \subseteq R' \cup \{ r \mid ( ( z, r ), d ) \in P \setminus \{ \ell_{n-1}, d ) \} \cup \{ r_{n-1} \} \subseteq R' \cup \{ r \mid ( ( m, r ), d ) \in P' \}.
              \]
              Then $P' \models R$ with respect to $\langle \Sigma \mid R' \rangle$.
              By \cref{Lemma:App:Subgraph2}, $\Gamma_D( P' )$ is also a subgraph of $\Gamma_D( Q )$ with $f( P' ) = f( Q ) - 1 = k$.
              Since $\Gamma_D( P' )$ is a subgraph of $\Gamma_D( Q )$ with $\Gamma_D( Q )$ acyclic, then $\Gamma_D( P' )$ is also acyclic.
              Then by the inductive hypothesis, there exists a proof $P''$ for $\langle \Sigma \mid R \rangle$ such that $|P''| = |P'|$, $P'' \models R$, and $d \in (R')^{*}$ for all $( \ell, d ) \in P''$.
              Then $|P''| = (|Q|-1)+1 = |Q|$, since $Q$ is indexed.
              Then the inductive step holds.
    \end{enumerate}
    Then by the principle of induction, there exists a proof $P'$ for $\langle \Sigma \mid R \rangle$ such that $|P'| = |P|$, $P' \models R$, and $d \in (R')^{*}$ for all $( \ell, d ) \in P$.
\end{proof}

\begin{theorem}
    \label{Thm:App:Derivation}
    Let $P$ be a finite proof for $\langle \Sigma \mid R' \rangle$ with $R' \subseteq R \subseteq \Sigma^{*} \times \Sigma^{*}$.
    If $P \models R$ and $\Gamma_D( P )$ is acyclic, then there exists a length $|R \setminus R'|$ sequence of \textbf{Rel($+$)} transformations between $\langle \Sigma \mid R' \rangle$ and $\langle \Sigma \mid R \rangle$.
\end{theorem}

\begin{proof}
    By \cref{Lemma:App:Derivation}, there exists a proof $P'$ for $\langle \Sigma \mid R' \rangle$ with $P' \models R$ and $d \in (R')^{*}$ for all $( \ell, d ) \in P$.
    The proof follows by induction on $|R \setminus R'|$.
    \begin{enumerate}
    \item[--] \textbf{Base Case}.
        If $|R \setminus R'| = 0$, then $R \subseteq R' \subseteq R$.
        Then $R = R'$ and $\langle \Sigma \mid R \rangle = \langle \Sigma \mid R' \rangle$.
        Then there exists a length zero sequence of \textbf{Rel($+$)} transformations between $\langle \Sigma \mid R' \rangle$ and $\langle \Sigma \mid R \rangle$.
    \item[--] \textbf{Inductive Hypothesis}.
        Let $P'$ be a proof with respect to $\langle \Sigma \mid Q \rangle$ with $Q \subseteq R$ and $P' \models R$.
        Assume that for some $k \in \mathbb{N}$, if $|R \setminus R'| = k$, then there exists a length $k$ sequence of \textbf{Rel($+$)} transformations between $\langle \Sigma \mid Q \rangle$ and $\langle \Sigma \mid R \rangle$.
    \item[--] \textbf{Inductive Step}.
        Under the conditions of the inductive hypothesis, assume that $|R \setminus R'| = k + 1$.
        Then there exists some $r \in R \setminus R'$, say $( w, w' ) = r$.
        Since $P' \models R$ and $r \not \in R'$, then there exists some $x \in \mathbb{N}$ and $d \in ( L(\Sigma) \cup R' )^*$ such that $( ( x, r ), d ) \in P'$.
        Since $d \in ( R' )^{*}$, then $w \sim_{R'} w'$ by validity of $P'$.
        Let $Q = R' \cup \{ r \}$.
        Then $\langle \Sigma \mid R' \rangle \cong \langle \Sigma \mid Q \rangle$ by \textbf{Rel($+$)}.
        Since $R' \subseteq Q$ and $P' \models R$ with respect to $\langle \Sigma \mid R' \rangle$, then $$R \subseteq R' \cup \{ r \mid ( ( n, r ), d ) \in P' \} \subseteq Q \cup \{ r \mid ( ( n, r ), d ) \in P' \}.$$
        Then $P' \models R$ with respect to $\langle \Sigma \mid Q \rangle$.
        Since $r \in R \setminus R'$, then $|R \setminus Q| = |R \setminus R'| - 1 = k$.
        Then by the inductive hypothesis, there exists a length $k$ sequence of \textbf{Rel($+$)} transformations between $\langle \Sigma \mid Q \rangle$ and $\langle \Sigma \mid R \rangle$.
        Then there exists a length $k + 1$ sequence of \textbf{Rel($+$)} transformations between $\langle \Sigma \mid R' \rangle$ and $\langle \Sigma \mid R \rangle$.
        Then the inductive step holds.
    \end{enumerate}
    Then by the principle of induction, there exists a length $|R \setminus R'|$ sequence of \textbf{Rel($+$)} transformations between $\langle \Sigma \mid R' \rangle$ and $\langle \Sigma \mid R \rangle$.
\end{proof}

\newpage
\section{Circuit Decompositions of Coxeter Generators}
\label{Appendix:E8Decomp}

In \cref{Sec:E8:DIntro}, the Coxeter generator $r_3$ was decomposed into a circuit over $\Sigma_D$.
In this section, the remaining $7$ Coxeter generators are decomposed into circuits over $\Sigma_D$.
Scalar multiples of the normal vectors are used freely.
Recall that $\CCGate{X}{0}{1}$ is a reflection about the normal vector $\ket{\hat{b}} = \ket{1} \otimes \ket{1} \otimes \ket{-}$.
Similarly, $\CGate{X}{1}{2}$ is a reflection about the normal vector $\ket{\overline{b}} = \ket{-} \otimes \ket{1} \otimes \ket{1}$.
\begin{enumerate}
\item[$r_1$.]
     This generator is defined by the normal vector $\ket{b_1} = \ket{0} \otimes \ket{0} \otimes \ket{-}$.
     Since $(\GateAt{X}{0} \circ \GateAt{X}{1}) \ket{\hat{b}} = \ket{b_1}$ with $(\GateAt{X}{0} \circ \GateAt{X}{1})^{-1} = \GateAt{X}{1} \circ \GateAt{X}{0}$, then $r_1 = \GateAt{X}{0} \circ \GateAt{X}{1} \circ \CCGate{X}{0}{1} \circ \GateAt{X}{1} \circ \GateAt{X}{0}$.
\item[$r_2$.]
     This generator is defined by the normal vector $\ket{b_2} = \ket{0} \otimes \left( \ket{1} \otimes \ket{0} - \ket{0} \otimes \ket{1} \right) / \sqrt{2}$.
     Since $\ket{b_2}$ and $-\ket{b_2}$ define the same hyperplane, then $-\ket{b_2}$ also defines the same generator.
     Recall that $r_3$ is a reflection about the normal vector $\ket{b_3} = \ket{0} \otimes \ket{1} \otimes \ket{-}$.
     Then $\left( \CGate{X}{2}{1} \right) \ket{b_3} = -\ket{b_2}$.
     Since $\CGate{X}{2}{1}$ is self-inverse, then $r_2 = \CGate{X}{2}{1} \circ r_3 \circ \CGate{X}{2}{1}$.
     Since $r_3 = \GateAt{X}{0} \circ \CCGate{X}{0}{1} \circ \GateAt{X}{0}$ with $\GateAt{X}{0}$ and $\CGate{X}{2}{1}$ commuting, then $r_2 = \GateAt{X}{0} \circ \CGate{X}{2}{1} \circ \CCGate{X}{0}{1} \circ \CGate{X}{2}{1} \circ \GateAt{X}{0}$.
\item[$r_4$.]
     This generator is defined by the normal vector $\ket{b_4} = \left( \ket{0} \otimes \ket{1} \otimes \ket{1} - \ket{1} \otimes \ket{0} \otimes \ket{0} \right) / \sqrt{2}$.
     Since $\ket{\overline{b}} = \left( \ket{0} \otimes \ket{1} \otimes \ket{1} - \ket{1} \otimes \ket{1} \otimes \ket{1} \right) / \sqrt{2}$, then $\left( \CGate{X}{0}{1} \circ \CGate{X}{0}{2} \right) \ket{\overline{b}} = \ket{b_4}$.
     Furthermore, since $\left( \CGate{X}{0}{1} \circ \CGate{X}{0}{2} \right)^{-1} = \CGate{X}{0}{2} \circ \CGate{X}{0}{1}$, then $r_4 = \CGate{X}{0}{1} \circ \CGate{X}{0}{2} \circ \CCGate{X}{1}{2} \circ \CGate{X}{0}{2} \circ \CGate{X}{0}{1}$.
\item[$r_5$.]
     This generator is defined by the normal vector $\ket{b_5} = \ket{1} \otimes \ket{0} \otimes \ket{-}$.
     Since $(\GateAt{X}{1}) \ket{\hat{b}} = \ket{b_5}$ with $\GateAt{X}{1}$ self-inverse, then $r_5 = \GateAt{X}{1} \circ \CCGate{X}{0}{1} \circ \GateAt{X}{1}$.
\item[$r_6$.]
     This generator is defined by the normal vector $\ket{b_6} = \ket{1} \otimes \left( \ket{0} \otimes \ket{1} - \ket{1} \otimes \ket{0} \right) / \sqrt{2}$.
     Since $\left( \CGate{X}{2}{1} \right) \ket{\hat{b}} = \ket{b_6}$ with $\CGate{X}{2}{1}$ self-inverse, then $r_6 = \CGate{X}{2}{1} \circ \CCGate{X}{0}{1} \circ \CGate{X}{2}{1}$.
\item[$r_7$.]
     This generator is defined by the normal vector $\ket{b_7} = \ket{1} \otimes \left( \ket{0} \otimes \ket{1} + \ket{1} \otimes \ket{0} \right) / \sqrt{2}$.
     Recall that $r_6$ is a reflection about the normal vector $\ket{b_6} = \ket{1} \otimes \left( \ket{0} \otimes \ket{1} - \ket{1} \otimes \ket{0} \right) / \sqrt{2}$.
     It follows that $\left( \CGate{Z}{0}{1} \right) \ket{b_6} = \ket{b_7}$.
     Since $\CGate{Z}{0}{1}$ is self-inverse, then $r_7 = \CGate{Z}{0}{1} \circ r_6 \circ \CGate{Z}{0}{1}$.
     Furthermore, since $r_6 = \CGate{X}{2}{1} \circ \CCGate{X}{0}{1} \circ \CGate{X}{2}{1}$, then $r_7 = \CGate{Z}{0}{1} \circ \CGate{X}{2}{1} \circ \CCGate{X}{0}{1} \circ \CGate{X}{2}{1} \circ \CGate{Z}{0}{1}$.
\item[$r_8$.]
     The generator is defined by the normal vector $\ket{b_8} = \ket{+} \otimes \ket{+} \otimes \ket{+}$.
     First, define the operator $M = K_{1,2} \circ \GateAt{X}{1} \circ \GateAt{X}{2} \circ \CGate{Z}{0}{2}$.
     Clearly, $M^{-1} = \CGate{Z}{0}{2} \circ \GateAt{X}{2} \circ \GateAt{X}{1} \circ K_{1,2}$.
     Furthermore,
     \begin{equation*}
         M \ket{\overline{b}}
         =
         \left( K_{1,2} \circ \GateAt{X}{1} \circ \GateAt{X}{2} \right) \left( \ket{+} \otimes \ket{1} \otimes \ket{1} \right)
         =
         K_{1,2} \left( \ket{+} \otimes \ket{0} \otimes \ket{0} \right)
         =
         \ket{b_8}.
     \end{equation*}
     Therefore, $r_8 = K_{1,2} \circ \GateAt{X}{1} \circ \GateAt{X}{2} \circ \CGate{Z}{0}{2} \circ \CCGate{X}{1}{2} \circ \CGate{Z}{0}{2} \circ \GateAt{X}{2} \circ \GateAt{X}{1} \circ K_{1,2}$.
\end{enumerate}
This establishes all decompositions of the $\Sigma_{E8}$ in terms of $\Sigma_{D}$.

\newpage
\section{Constructing the Generators for \texorpdfstring{$W( E_8 )$}{W(E8}}
\label{Appendix:E8Construct}

The section walks through the construction of $\GateAt{X}{0}$, $\CGate{X}{0}{1}$, $\CCGate{X}{1}{2}$, and $\GateAt{K}{1,2}$ using the Coxeter generators for $W( E_8 )$.
As suggested in \cref{Sec:E8:DIntro}, this construction begins by deriving several diagonal matrices over $(\pm 1)$.
{\small\begin{align*}
    w_1 &= r_6 \cdot r_7
    &
    w_2 &= r_6 \cdot r_5 \cdot w_1 \cdot r_5 \cdot r_6
    &
    w_3 &= r_5 \cdot r_4 \cdot w_2 \cdot r_4 \cdot r_5
    \\
    w_4 &= r_4 \cdot r_3 \cdot w_3 \cdot r_3 \cdot r_4
    &
    w_5 &= r_3 \cdot r_2 \cdot w_4 \cdot r_2 \cdot r_3
    &
    w_6 &= r_2 \cdot r_1 \cdot w_5 \cdot r_1 \cdot r_2
\end{align*}}%
For example, $\interp{w_1}^{*}_{E8} = \CGate{Z}{0}{1} \circ \CGate{Z}{0}{2}$.
It is then possible to derive $\CCGate{X}{0}{1}$ and $\GateAt{X}{2}$.
{\small\begin{align*}
    w_7 &= r_7 \cdot r_8 \cdot r_6 \cdot w_6 \cdot w_4 \cdot w_2 \cdot r_8 \cdot w_6 \cdot w_4 \cdot w_2 \cdot r_6 \cdot r_8 \cdot r_7
    &
    w_8 &= r_1 \cdot r_3 \cdot r_5 \cdot w_7
\end{align*}}%
Then $\interp{w_7}^{*}_{E8} = \CCGate{X}{0}{1}$ and $\interp{w_8}^{*}_{E8} = \GateAt{X}{2}$.
Using $\CCGate{X}{0}{1}$, it is then possible to derive $K_{1,2}$.
{\small\begin{align*}
    w_9 &= r_6 \cdot w_7 \cdot w_1 \cdot w_7 \cdot r_6
    &
    w_{10} &= r_2 \cdot r_6 \cdot w_5 \cdot w_3 \cdot w_2 \cdot r_8 \cdot w_9 \cdot r_8 \cdot w_5 \cdot w_3 \cdot w_2 \cdot w_9
\end{align*}}%
Then $\interp{w_9}^{*}_{E8}$ is a diagonal matrix over $(\pm 1)$ and $\interp{w_{10}}^{*}_{E8} = K_{1,2}$.
Next, the permutations are derived.
{\small\begin{align*}
    w_{11} &= w_{10} \cdot r_4 \cdot w_8 \cdot r_4 \cdot w_{10} \cdot w_8
    &
    w_{12} &= r_2 \cdot r_6
    &
    w_{13} &= w_{11} \cdot w_{12} \cdot w_{11}
\end{align*}}%
It can be validated that $\interp{w_{11}}^{*}_{E8} = \sigma_{1,2}$, $\interp{w_{12}}^{*}_{E8} = \sigma_{0,1}$, and $\interp{w_{13}}^{*}_{E8} = \sigma_{0,2}$.
As an immediate consequence, $\interp{w_{13} \cdot w_7 \cdot w_{13}}^{*}_{E8} = \CCGate{X}{1}{2}$ and $\interp{w_{13} \cdot w_8 \cdot w_{13}}^{*}_{E8} = \GateAt{X}{0}$.
Then by three applications of \textbf{Gen($+$)}, the generators $K_{1,2}$, $\CCGate{X}{1}{2}$, and $\GateAt{X}{0}$ are introduced, alongside the following relations.
{\small\begin{align*}
    K_{1,2} &\approx w_{10}
    &
    \CCGate{X}{1}{2} &\approx w_{13} \cdot w_7 \cdot w_{13}
    &
    \GateAt{X}{0} &\approx w_{13} \cdot w_8 \cdot w_{13}
\end{align*}}%
Next, define $w_{14} = w_{12} \cdot \GateAt{X}{0} \cdot w_{10} \cdot \GateAt{X}{0} \cdot w_{12}$.
It can be validated directly that $\interp{w_{14}}^{*}_{E8} = \CGate{X}{0}{1}$.
Then by application of \textbf{Gen($+$)}, the generator $\CGate{X}{0}{1}$ is introduced, alongside the relation $\CGate{X}{0}{1} \approx w_{14}$.

\newpage
\section{Establishing the Minimality of \texorpdfstring{$W( E_8 )$}{W(E8)} and \texorpdfstring{$\mathrm{O}(8,\D)$}{O(8,D)} Generators}
\label{Appendix:E8MinGens}

This section establishes the minimality of certain generating sets for $W( E_8 )$ and $\mathrm{O}(8,\D)$.
First, a general result about minimal generating sets is established.
This result is then applied to the generating sets of interest, to prove their minimality.

\subsection{Two Results on Minimal Generating Sets}

\begin{theorem}
    \label{Lemma:App:MinGens}
    Let $G$ be a group with $\Sigma' \subseteq \Sigma \subseteq G$.
    If there exists a $g \in G$ such that $g$ commutes with the elements of $\Sigma'$ and $g$ does not commute with the elements of $\Sigma$, then $\langle \Sigma' \rangle$ is a proper subgroup of $\langle \Sigma \rangle$.
\end{theorem}

\begin{proof}
    Assume that $g \in G$, $g$ commutes with every element of $\Sigma'$, and $\langle \Sigma' \rangle = \langle \Sigma \rangle$.
    It follows by induction on the length of an element in $\langle \Sigma' \rangle$, that $g$ commutes with every element in $\langle \Sigma' \rangle$.
    As a base case, if $h \in \langle \Sigma' \rangle$ corresponds to a word of length $0$, then $h$ is the identity and $g \circ h = g = h \circ g$.
    As an inductive hypothesis, assume that for some $n \in \mathbb{N}$, if $h_1, h_2, \ldots, h_n \in \Sigma'$ and $h = h_1 \circ h_2 \circ \cdots \circ h_n$, then $g \circ h = h \circ g$.
    To show that the inductive step holds, let $h_1, h_2, \ldots, h_n, h_{n+1} \in \Sigma'$ and $h = h_1 \circ h_2 \circ \cdots h_{n+1}$.
    By the inductive hypothesis, $g \circ h' = h' \circ g$ where $h' = h_1 \circ h_2 \circ \cdots \circ h_n$.
    Then $g \circ h = g \circ h' \circ h_{n+1} = h' \circ g \circ h_{n+1} = h' \circ h_{n+1} \circ g = h \circ g$.
    Then the inductive step holds, and $g$ commutes with every element of $\langle \Sigma' \rangle$.
    In particular, $g$ commutes with $\Sigma$.
    By the contrapositive, if $g$ does not commute with $\Sigma$, then $\langle \Sigma' \rangle \ne \langle \Sigma \rangle$.
    However, $\langle \Sigma' \rangle \le \langle \Sigma \rangle$ since $\Sigma' \subseteq \Sigma$.
    Therefore, $\langle \Sigma' \rangle$ is a proper subgroup of $\langle \Sigma \rangle$.
\end{proof}

\begin{lemma}
    Let $G$ be a group with $\Sigma \subseteq G$.
    If for every maximal proper subset $\Sigma'$ of $\Sigma$, $\langle \Sigma' \rangle$ is a proper subgroup of $\langle \Sigma \rangle$, then $\Sigma$ is a minimal generating set for $\langle \Sigma \rangle$.
\end{lemma}

\begin{proof}
    Let $\Sigma'$ be a proper subset of $\Sigma$.
    Then there exists some maximal proper subset $\Pi$ of $\Sigma$ such that $\Sigma' \subseteq \Pi \subseteq \Sigma$.
    Then $\langle \Sigma' \rangle \le \langle \Pi \rangle \le \langle \Sigma \rangle$.
    Since $\Pi$ is maximal, then by assumption, $\langle \Pi \rangle$ is a proper subgroup of $\langle \Sigma \rangle$.
    Consequently, $\langle \Sigma' \rangle$ is a proper subgroup of $\langle \Sigma \rangle$.
    Since $\Sigma'$ was arbitrary, then $\Sigma$ is a minimal generating set for $\langle \Sigma \rangle$.
\end{proof}

\subsection{Minimality for \texorpdfstring{$W(E_8)$}{W(E8)}}
\label{Appendix:E8MinGens:E8}

It must be shown that for every maximal proper subset $\Sigma'$ of $\Sigma_0$, there exists some $8 \times 8$ dyadic matrix $M$ such that $M$ commutes with $\Sigma'$ but does not commute with $\Sigma_0$.
The first three cases can be solved by inspection.
In fact, these matrices follow from well-known circuit relations.
\begin{enumerate}
    \item $\GateAt{Z}{2}$ commutes with $\{ \GateAt{X}{0}, \CGate{X}{0}{1}, \CCGate{X}{1}{2} \}$ but does not commute with $K_{1,2}$.
    \item $\GateAt{H}{2}$ commutes with $\{ \GateAt{X}{0}, \CGate{X}{0}{1}, K_{1,2} \}$ but does not commute with $\CCGate{X}{1}{2}$.
    \item $\GateAt{X}{0}$ commutes with $\{ \GateAt{X}{0}, \CCGate{X}{1}{2}, K_{1,2} \}$ but does not commute with $\CGate{X}{0}{1}$.
\end{enumerate}
The final case is less obvious, but can be reduced to solving a linear integer program.
Assume that there exists such a matrix $M$.
Since $M$ is dyadic, then there exists some integer matrix $N$ and integer $k$ such that $M = N / 2^k$.
Clearly, $M$ and $N$ commute with the same matrices.
Then $M$ is characterized by the following four equations.
{\small\begin{align*}
    \GateAt{X}{0} \circ N &\ne N \circ \GateAt{X}{0}
    &
    \CGate{X}{0}{1} \circ N &= N \circ \CGate{X}{0}{1}
    &
    \CCGate{X}{1}{2} \circ N &= N \circ \CCGate{X}{1}{2}
    &
    K_{1,2} \circ N &= N \circ K_{1,2}
\end{align*}}%
Without loss of generality, $K_{1,2}$ can be replaced by its integral scalar multiple $2 \cdot K_{1,2}$.
Then the entries of $N$ can be thought of as $64$ integer variables, with each equation yielding $64$ linear constraints.
Using Z3~\cite{MouraBjorner2008} as a solver, the following solution is obtained.
{\small\begin{equation*}
    N_{0,0} = \begin{bmatrix}
  4 & 2 & 2 & 0 \\
  2 & 1 & 1 & 0 \\
  2 & 1 & 1 & 0 \\
  0 & 0 & 0 & 0
\end{bmatrix}

    \qquad
    N_{0,1} = N_{1,0} = N_{1,1} = \begin{bmatrix}
    0 & 0 & 0 & 0 \\
    0 & 0 & 0 & 0 \\
    0 & 0 & 0 & 0 \\
    0 & 0 & 0 & 0
\end{bmatrix}
    \qquad
    N = \begin{bmatrix}
        N_{0,0} & M_{0,1} \\
        N_{1,0} & M_{1,1}
    \end{bmatrix}
\end{equation*}}
This establishes \cref{Thm:MinGens}.

\subsection{Minimality of \texorpdfstring{$\Sigma_K$}{Sigma(K)} for \texorpdfstring{$\mathrm{O}(8,\D)$}{O(8,D)}}

First, it will be shown that $\Sigma_K$ is a minimal generating set for $\mathrm{O}(8, \D)$.
To see that $\Sigma_K$ generates $\mathrm{O}(8, \D)$ simply note that $\Sigma_K \cup \{ \GateAt{K}{1,2} \}$ generates $\mathrm{O}(8, \D)$ with $\GateAt{K}{1,2} = \TLO{K}{0,1,2,3} \circ \GateAt{X}{0} \circ \TLO{K}{0,1,2,3} \circ \GateAt{X}{0}$.
It remains to be shown that for every maximal proper subset $\Sigma'$ of $\Sigma_K$, there exists some $8 \times 8$ dyadic matrix $M$ such that $M$ commutes with $\Sigma'$ but does not commute with $\Sigma_K$.
The first three cases are also solved by inspection, using well-known circuit relations.
\begin{enumerate}
\item $\GateAt{Z}{2}$ commutes with $\left\{ \GateAt{X}{0}, \CGate{X}{0}{1}, \CCGate{X}{1}{2} \right\}$ but does not commute with $\TLO{K}{0,1,2,3}$.
    \item $\GateAt{H}{2}$ commutes with $\left\{ \GateAt{X}{0}, \CGate{X}{0}{1}, \TLO{K}{0,1,2,3} \right\}$ but does not commute with $\CCGate{X}{1}{2}$.
\item $\GateAt{X}{2} \circ \CGate{Z}{0}{2} \circ \GateAt{X}{2}$ commutes with $\left\{ \CGate{X}{0}{1}, \CCGate{X}{1}{2}, \TLO{K}{0,1,2,3} \right\}$ but does not commute with $\GateAt{X}{0}$.
\end{enumerate}
Using Z3 as in the $W( E_8 )$, it is then possible to find an integer matrix $L$ such that $L$ commutes with $\left\{ \GateAt{X}{0}, \CCGate{X}{1}{2}, \TLO{K}{0,1,2,3} \right\}$ but does not commute with $\CGate{X}{0}{1}$.
The solution is as follows.
{\small\begin{equation*}
    L_0 = \begin{bmatrix}
        1 &  2 &  2 &  0 \\
        2 &  0 & -1 &  0 \\
        2 & -1 &  0 &  0 \\
        0 &  0 &  0 & -3 \\
    \end{bmatrix}
    \qquad
    L_1 = \begin{bmatrix}
    0 & 0 & 0 & 0 \\
    0 & 0 & 0 & 0 \\
    0 & 0 & 0 & 0 \\
    0 & 0 & 0 & 0
\end{bmatrix}
    \qquad
    L = \begin{bmatrix}
        L_0 & L_1 \\
        L_1 & L_0
    \end{bmatrix}
\end{equation*}}%
This establishes the first claim of~\cref{Thm:O8DMin}.

\subsection{Minimality of \texorpdfstring{$\Sigma_Z$}{Sigma(Z)} for \texorpdfstring{$\mathrm{O}(8,\D)$}{O(8,D)}}

The proof the $\Sigma_Z$ is minimal proves more challenging.
However, the minimality of $\Sigma_K$ can be used to simplify this argument significantly.
Clearly, $\Sigma_D = \Sigma_Z \setminus \{ CCZ \}$ does not generate $\mathrm{O}(8,\D)$, since $\langle \Sigma_D \rangle = W( E_8 ) < \mathrm{O}(8,\D)$.
Three of the remaining four cases can be solved by inspection.
\begin{enumerate}
\item Recall $N$ from \cref{Appendix:E8MinGens:E8}.
      By construction, this matrix commutes with $\left\{ \CGate{X}{0}{1}, \CCGate{X}{1}{2}, \GateAt{K}{1,2} \right\}$ and does not commute with $\GateAt{X}{0}$.
      Furthermore, $CCZ \circ N = N \circ CCZ$, since the $8$-th row and $8$-th column of $N$ contain only zeros.
\item $\sigma_{1,2}$ commutes with $\left\{ \GateAt{X}{0}, \CCGate{X}{1}{2}, \GateAt{K}{1,2}, CCZ \right\}$ but does not commute with $\CGate{X}{0}{1}$.
\item $\GateAt{Z}{2}$ commutes with $\left\{ \GateAt{X}{0}, \CGate{X}{0}{1}, \CCGate{X}{1}{2}, CCZ \right\}$ but does not commute with $\GateAt{K}{1,2}$.
\end{enumerate}
The case of $\CCGate{X}{1}{2}$ requires more care.
There is no obvious operation which commutes with all generators except for $\CCGate{X}{1}{2}$.
Furthermore, Z3 fails to find an solution to the corresponding integer program.
Instead, consider the automorphism $f: M \mapsto \GateAt{H}{2} \circ M \circ \GateAt{H}{2}$ of $\mathrm{O}(8,\D)$.
Since $f$ fixes $\left\{ \GateAt{X}{0}, \CGate{X}{0}{1}, \GateAt{K}{1,2,} \right\}$ and maps $CCZ$ to $\CCGate{X}{0}{1}$, then $f$ induces an isomorphism between the subgroups $\langle \GateAt{X}{0}, \CGate{X}{0}{1}, \GateAt{K}{1,2}, CCZ \rangle$ and $\langle \GateAt{X}{0}, \CGate{X}{0}{1}, \GateAt{K}{1,2}, \CCGate{X}{0}{1} \rangle$ of $\mathrm{O}(8,\D)$.
Since $\GateAt{X}{0}, \CGate{X}{0}{1}$, $\GateAt{K}{1,2,}, \CCGate{X}{0}{1} \in W( E_8 )$ with $W( E_8 )$ finite, then
\begin{equation*}
    \left| \langle \GateAt{X}{0}, \CGate{X}{0}{1}, \GateAt{K}{1,2}, CCZ \rangle \right|
    =
    \left| \langle \GateAt{X}{0}, \CGate{X}{0}{1}, \GateAt{K}{1,2}, \CCGate{X}{1}{2} \rangle \right|
    \le
    |W( E_8 )|
    <
    \infty.
\end{equation*}
Since $\mathrm{O}(8,\D)$ is an infinite group, then $\langle \GateAt{X}{0}, \CGate{X}{0}{1}, \GateAt{K}{1,2}, CCZ \rangle < \mathrm{O}(8,\D)$.
This establishes the second claim of~\cref{Thm:O8DMin}.

\newpage
\section{Proof Details for a Presentation of \texorpdfstring{$\mathrm{O}(8,\D)$}{O(u,D)}}
\label{Appendix:O8D}

In \cref{Sec:O8D}, many informal claims were made about the relations in $\TLOGen{n}{}$, and the derivations that are possible using these relations.
This section restates each claim as a lemma or theorem, provides a proof for each claim, and then explains how these claims establish the lemmas and theorems in \cref{Sec:O8D}.

\subsection{Counting the Relations in \texorpdfstring{$\TLORel{n}{}$}{R(n)}}
\label{Appendix:RelCount}

This section validates the claim that $\TLOGen{8}{}$ contains $2039$ relations.
We say that two relations $( q, r ) \in \TLOGen{8}{}$ and $( q', r' ) \in \TLOGen{8}{}$ are distinct if $q \ne q'$ or $r \ne r'$.
This means, for example, that the relations $\TLO{K}{0,1,2,3} \cdot \TLO{K}{4,5,6,7} \approx \TLO{K}{4,5,6,7} \cdot \TLO{K}{4,5,6,7}$ and $\TLO{K}{4,5,6,7} \cdot \TLO{K}{0,1,2,3} \approx \TLO{K}{0,1,2,3} \cdot \TLO{K}{4,5,6,7}$ are distinct
The techniques used in this section can be generalized to count the number of relations in $\TLOGen{n}{}$.

First, consider the relations whose parameters are linearly ordered.
If a relation schema $r$ has $m$ linearly ordered parameters, then each choice of $m$ distinct numbers in $[n]$ corresponds to a unique instance of $r$.
It follows that a relation schema with $m$ linearly ordered parameters corresponds to $\binom{n}{m}$ unique relations.
For each choice of $m$, we compute $\binom{8}{m}$ and count the number of relations with $m$ linearly ordered parameters.
\begin{enumerate}
\item[--]
      If $m = 1$, then there are $\binom{8}{1} = 8$ instances.
      The only relation with a single parameter is \cref{EQ:O8DRels:Orig:Rep1}.
      Then this case contributes $8$ relations.
\item[--]
      If $m = 2$, then there are $\binom{8}{2} = 28$ instances.
      The only relations with two parameters, all linearly ordered, are \cref{Eq:O8DRels:Orig:Perm1,Eq:O8DRels:Orig:Perm7}.
      Then this case contributes $56$ relations.
\item[--]
      If $m = 3$, then there are $\binom{8}{3} = 56$ instances.
      The only relations with three parameters, all linearly ordered, are \cref{Eq:O8DRels:Orig:Perm5,Eq:O8DRels:Orig:Perm6}.
      Then this case contributes $102$ relations.
\item[--]
      If $m = 4$, then there are $\binom{8}{4} = 70$ instances.
      The relations \cref{EQ:O8DRels:Orig:Rep2,EQ:O8DRels:Orig:Rep5a,EQ:O8DRels:Orig:Rep6,EQ:O8DRels:Orig:Rep7} all have exactly four parameters, which are all linearly ordered.
      Then this case contributes $280$ relations.
\item[--]
      If $m = 5$, then there are $\binom{8}{5} = 56$ instances.
      The relations \cref{Eq:O8DRels:Orig:Perm8,Eq:O8DRels:Orig:Perm9,Eq:O8DRels:Orig:Perm10,Eq:O8DRels:Orig:Perm11} all have exactly five parameters, which are all linearly ordered.
      This this case contributes $224$ relations.
\item[--]
      If $m = 6$, then there are $\binom{8}{6} = 28$ instances.
      The only relation with a six parameters is \cref{EQ:O8DRels:Orig:Rep8}.
      Then this case contributes $28$ relations.
\item[--]
      If $m = 8$, then there are $\binom{8}{8} = 1$ instances.
      The only relation with eight parameters is \cref{EQ:O8DRels:Orig:Rep9}.
      Then this case contributes $1$ relations.
\end{enumerate}
In total, the relation schemata with linearly ordered parameters contribute $699$ instances.

The remaining six schemata induce a partial order on the parameters.
For example, in \cref{Eq:O8DRels:Orig:KCom1}, the term $\TLO{K}{a,b,c,d} \cdot \TLO{K}{e,f,g,h}$ indicates that $a < b < c < d$ and $e < f < g < h$.
However, the choices of $( a, b, c, d )$ are independent from the choices of $( e, f, g, h )$, except that all choices must be distinct.
In this example, there are $\binom{n}{4}$ ways to select the four indices in the first order.
Then $n - m$ indices remain, from which there are $\binom{n - 4}{4}$ choices.
In general, for two independent linear orders with $m$ and $k$ parameters respectively, there will be $\binom{n}{m} \cdot \binom{n - m}{k}$ choices.
The six schemata are described below.
\begin{enumerate}
\item[--]
      In \cref{Eq:O8DRels:Orig:ZCom1}, $m = 1$ and $k = 1$, resulting in $\binom{8}{1} \cdot \binom{7}{1} = 56$ choices.
\item[--]
      In \cref{Eq:O8DRels:Orig:Perm3}, $m = 1$ and $k = 2$, resulting in $\binom{8}{1} \cdot \binom{7}{2} = 168$ choices.
\item[--]
      In \cref{EQ:O8DRels:Orig:Rep3}, $m = 1$ and $k = 4$, resulting in $\binom{8}{1} \cdot \binom{7}{4} = 280$ choices.
\item[--]
      In \cref{Eq:O8DRels:Orig:Perm2}, $m = 2$ and $k = 2$, resulting in $\binom{8}{2} \cdot \binom{6}{2} = 420$ choices.
\item[--]
      In \cref{Eq:O8DRels:Orig:Perm4}, $m = 2$ and $k = 4$, resulting in $\binom{8}{2} \cdot \binom{6}{4} = 420$ choices.
\item[--]
      In \cref{Eq:O8DRels:Orig:KCom1}, $m = 4$ and $k = 4$, resulting in $\binom{8}{4} \cdot \binom{4}{4} = 70$ choices
\end{enumerate}
In total, the relations partially ordered parameters contribute $1414$ instances.
Then $|\TLORel{8}{}| = 2113$.

\subsection{Correctness of Relation Reindexing}
\label{Appendix:O8D:Reindex}

This section justifies the reindexing of relations via permutations.
First, recall that every permutation on $[n]$ can be represented by a permutation of the basis vectors in $\mathbb{R}^8$, with $\tau_{a,b}$ corresponding to $\TLO{X}{a,b}$.
The intuition is that every $\sigma$ can be represented by a word $w$ over generators of type $X$, and that conjugation by $w$ corresponds to formal reindexing when $\sigma$ is valid.

First, a subset $\TLORel{n}{B}$ of $\TLORel{n}{}$ is identified, for which all of the order and braiding relations for generators of type $X$ hold.
Consequently, $\TLORel{n}{B}$ is complete for words over generators of type $X$.
Then $\TLORel{n}{B}$ is extended to a subset $\TLORel{n}{\sigma}$ of $\TLORel{n}{}$ for which all formal reindexings are derivable.
The result is proven first, for individual generators, and then extended to entire words.

\subsubsection{Deriving the Braiding Relations}

First, define the set of relations,
\begin{equation*}
    \TLORel{n}{\tau} = \left\{ \TLO{X}{a,a+1}{}^2 \approx \epsilon \mid a \in \mathbb{Z} \right\}
        \cup
        \left\{ \TLO{X}{a,a+1} \cdot \TLO{X}{a,b} \approx \TLO{X}{a+1,b} \cdot \TLO{X}{a,a+1} \mid a, b \in \mathbb{Z} \text{ with } a + 1 < b  \right\}.
\end{equation*}
This set is sufficient to decompose all swaps into transpositions, as proven in \cref{Lemma:BraidDecomp}.
Of interest in this section is the following extension of $\TLORel{n}{\tau}$,
\begin{equation*}
    \TLORel{n}{B} = \TLORel{n}{\tau}
        \cup
        \left\{ \TLO{X}{a+1,a+2} \cdot \TLO{X}{a,a+1} \approx \TLO{X}{a,a+2} \cdot \TLO{X}{a+1,a+2} \mid a \in \mathbb{Z} \right\}.
\end{equation*}
The relations in $\TLORel{n}{B}$ entail the braiding and order relations for $S(n)$, as shown in \cref{Lemma:TypeXBraids}, and are therefore complete for equality of words over generators of type $X$.
Of important note is that $\TLORel{n}{B} \subseteq \TLORel{n}{}$.

\begin{lemma}
    \label{Lemma:BraidDecomp}
    Let $R$ be a set of relations over $\TLOGen{n}{}$ which contains all well-formed relations in $\TLORel{n}{\tau}$.
    If $v$ is a word over two-level operators of type $X$, then there exists a word $u$ over transpositions such that $v \sim_{R} u$.
\end{lemma}

\begin{proof}
    Consider a valid two-level operator $\TLO{X}{a,b}$.
    The proof follows by induction on $b - a > 0$.
    \begin{enumerate}
    \item[--] \textbf{Base Case}.
          If $b - a = 1$, then $\TLO{X}{a,b} = \TLO{X}{a,a+1}$.
    \item[--] \textbf{Inductive Hypothesis}.
          Assume that for some $k \in \mathbb{N}_{>0}$, if $b - a = k$, then there exists a word $w$
    \item[--] \textbf{Inductive Step}.
          Assume that $b - a = k + 1$.
          Since $k > 0$, then $a + 1 \ne b$, and the following derivation holds.
          {\scriptsize\begin{equation*}
              \TLO{X}{a,b}
              \leftarrow
              \TLO{X}{a,a+1}{}^2 \cdot \TLO{X}{a,b}
              \rightarrow
              \TLO{X}{a,a+1} \cdot \TLO{X}{a+1,b} \cdot \TLO{X}{a,a+1}
          \end{equation*}}
          Since $b - (a + 1) = k$, then by the inductive hypothesis, there exists a word $w$ over transpositions such that $\TLO{X}{a+1,b} \sim_R w$.
          Then $\TLO{X}{a,b} \sim_R \TLO{X}{a,a+1} \cdot w \cdot \TLO{X}{a,a+1}$.
          Since $\TLO{X}{a,a+1} \cdot w \cdot \TLO{X}{a,a+1}$ is a word over transpositions, then the inductive hypothesis holds.
    \end{enumerate}
    Then for each symbol $M$ in $w$, there exists a decomposition of $M$ into transpositions.
    Then by \cref{Appendix:DerivedGens}, there exists a word $v$ over transpositions such that $w \sim_R v$.
\end{proof}

\begin{lemma}
    \label{Lemma:TypeXBraids}
    Let $R$ be a set of relations over $\TLOGen{n}{}$ which contains all well-formed relations in $\TLORel{n}{B}$.
    If $v$ and $w$ are words over two-level operators of type $X$ and $\interp{v}_O^* = \interp{w}_O^*$, then $v \sim_{\TLORel{\sigma}{}} w$.
\end{lemma}

\begin{proof}
    Since $v$ and $w$ are words over two-level operators of type $X$ with $\TLORel{n}{\tau} \subseteq \TLORel{n}{B}$, then by \cref{Lemma:BraidDecomp} there exists words $\hat{v}$ and $\hat{w}$ over transpositions such that $\hat{v} \sim \hat{w}$.
    Then $\hat{v}$ and $\hat{w}$ are words in the braid representation of $\mathrm{S}(n)$.
    If $R$ contains all order and braiding relations for the transpositions in $\TLOGen{n}{}$, then $R$ is complete for words over the transpositions in $\TLOGen{n}{}$.
    Let $a \in [n-2]$.
    Then the following derivation holds using the relations in $\TLORel{n}{B}$.
    {\scriptsize\begin{align*}
        \TLO{X}{a,a+1} \cdot \TLO{X}{a+1,a+2} \cdot \TLO{X}{a,a+1}
        \leftarrow
        \TLO{X}{a,a+1}{}^2 \cdot \TLO{X}{a,a+2}
        \rightarrow
        \TLO{X}{a,a+2}
        \leftarrow
        \TLO{X}{a,a+2} \cdot \TLO{X}{a+1,a+2}{}^2
        \leftarrow
        \TLO{X}{a+1,a+2}{} \cdot \TLO{X}{a,a+1} \cdot \TLO{X}{a+1,a+2}
    \end{align*}}%
    Then $\TLO{X}{a,a+1} \cdot \TLO{X}{a+1,a+2} \cdot \TLO{X}{a,a+1} \sim_{R} \TLO{X}{a+1,a+2} \cdot \TLO{X}{a,a+1} \cdot \TLO{X}{a+1,a+2}$.
    Since $a$ was arbitrary, then $R$ is complete for $\mathrm{S}(n)$.
    Since $v \sim_R \hat{v}$ and $w \sim_R \hat{w}$, then $\interp{v}_O^* = \interp{\hat{v}}_O^*$ and $\interp{w}_O^* = \interp{\hat{w}}_O^*$.
    Then $\interp{\hat{v}}_O^* = \interp{\hat{w}}_O^*$.
    Then $\hat{v} \sim_R \hat{w}$ by the completeness of $R$.
    Then $v \sim_R w$ by the transitivity of $( \sim_R )$.
\end{proof}

\subsubsection{Formal Inverses for Self-Inverse Generators}

For each $w = w_1 \cdot w_2 \cdots w_n$ over $\TLOGen{n}{}$, define $\overline{w} = w_n \cdots w_2 \cdot w_1$.
Since each element in $\TLOGen{n}{}$ is self-inverse, then $\interp{\overline{w}}_{O}^*$ is the inverse of $\interp{w}_{O}^*$ in $\mathrm{O}( 8, \D )$.
One can prove that given a complete set of relations, both $u \cdot \overline{u}$ and $\overline{u} \cdot u$ always derive to $\epsilon$.
For the purposes of this proof, only the case for $X$-type generators is necessary.

\begin{lemma}
    \label{Lemma:FormalInverse}
    Let $R$ be a set of relations over $\TLOGen{n}{}$ which contains all well-formed relations in the set below.
    {\scriptsize\begin{align*}
        \left\{ \TLO{X}{a,a+1}{}^2 \approx \epsilon \mid a \in \mathbb{Z} \right\}
        \cup
        \left\{ \TLO{X}{a,a+1} \cdot \TLO{X}{a,b} \approx \TLO{X}{a+1,b} \cdot \TLO{X}{a,a+1} \mid a, b \in \mathbb{Z} \right\}
        \cup
        \left\{ \TLO{X}{a+1,a+2} \cdot \TLO{X}{a,a+1} \approx \TLO{X}{a,a+@} \cdot \TLO{X}{a+1,a+2} \mid a \in \mathbb{Z} \right\}
    \end{align*}}%
    If $u$ is a word over two-level operators of type $X$, then $u \cdot \overline{u} \sim_{R} \epsilon$ and $\overline{u} \cdot u \sim_{R} \epsilon$.
    Furthermore, if $v$ is a word over two-level operators of type $X$ and $u \sim_R v$, then $\overline{u} \sim_R \overline{v}$.
\end{lemma}

\begin{proof}
    Since $\interp{\cdot}_O^*$ maps each generator in $\TLOGen{n}{}$ to a self-inverse matrix, then $\interp{\overline{u}}_O^*$ is the inverse to $\interp{u}_O^*$.
    Then $\interp{\overline{u} \cdot u}_O^* = \interp{\epsilon}_O^* = \interp{u \cdot \overline{u}}_O^*$.
    Since $u$ is a word over two-level operators of type $X$, then $u \cdot \overline{u} \sim_{R} \epsilon$ and $\overline{u} \cdot u \sim_{R} \epsilon$ by \cref{Lemma:TypeXBraids}.
    Now assume that $v$ is a word over two-level operators of type $X$ with $u \sim_R v$.
    Then $\interp{u}_O^* = \interp{v}_O^*$.
    Since $\overline{u}$ is the inverse to $u$ and $\overline{v}$ is the inverse to $v$, then $\interp{\overline{u}}_O^* = \interp{\overline{v}}_O^*$.
    Then $\overline{u} \sim_R \overline{v}$ by \cref{Lemma:TypeXBraids}.
\end{proof}

\subsubsection{Permuting the Indices in Multi-Level Operators}

\begin{theorem}
    \label{Thm:PermuteX}
    If $\sigma \in \mathrm{S}(n)$ is a valid reindexing for a two-level operator $M$ of type $X$, then there exists a word $v$ over the two-level operators of type $X$, such that $\interp{v}_O^* = \interp{\sigma}_S$ and $\sigma(M) \sim_{\TLORel{\sigma}{}} v \cdot M \cdot \overline{v}$.
\end{theorem}

\begin{proof}
    Since $M$ is a two-level operator of type $X$, then there exists an increasing sequence $(a,b)$ over $[n]$ such that $M = \TLO{X}{a,b}$.
    Let $\sigma_1 \circ \sigma_2 \circ \cdots \circ \sigma_m$ be the decomposition of $\sigma$ into a sequence of transpositions.
    Then define $v = \interp{\sigma_1}_S \cdot \interp{\sigma_2}_S \cdots \interp{\sigma_m}_S$.
    Clearly $v$ is a word over two-level operators of type $X$ satisfying $\interp{v}_O^* = \interp{\sigma}_S$.
    Furthermore, $\interp{\sigma(M)}_O^* = \interp{\tau_{\sigma(a),\sigma(b)}}_S = \interp{\sigma \cdot \tau_{a,b} \cdot \sigma^{-1}}_S^* = \interp{v \cdot M \cdot \overline{v}}_O^*$.
    Since $\sigma(M)$ and $\overline{v} \cdot M \cdot v$ are words over two-level operators of type $X$, then $\sigma(M) \sim_{\TLORel{\sigma}{}} v \cdot M \cdot \overline{v}$ by \cref{Lemma:TypeXBraids}.
\end{proof}

\begin{theorem}
    \label{Thm:PermuteZ}
    If $\sigma \in \mathrm{S}(n)$ and $M$ is a one-level operator of type $(-1)$, then there exists a word $v$ of transpositions, such that $\interp{v}_O^* = \interp{\sigma}_S$ and $\sigma(M) \sim_{\TLORel{\sigma}{}} v \cdot M \cdot \overline{v}$.
\end{theorem}

\begin{proof}
    Since $M$ is a one-level operator of type $(-1)$, then there exists an $a \in [a]$ such that $M = \TLO{(-1)}{a}$.
    Let $v $
    Let $\sigma_1 \circ \sigma_2 \circ \cdots \circ \sigma_m$ be the decomposition of $\sigma$ into a sequence of transpositions.
    Then define $v = \interp{\sigma_1}_S \cdot \interp{\sigma_2}_S \cdots \interp{\sigma_m}_S$.
    It follows by induction on $m$ that $\sigma(M) \sim_{\TLORel{\sigma}{}} v \cdot M \cdot \overline{v}$.
    \begin{enumerate}
    \item[--] \textbf{Base Case}.
          Assume that $m = 0$.
          Then $v = \overline{v} = \epsilon$ and $\sigma(M) = M$.
          Then $\sigma(M) \sim_{\TLORel{\sigma}{}} v \cdot M \cdot \overline{v}$ by the reflexivity of $(\sim_{\TLORel{\sigma}{}})$.
    \item[--] \textbf{Inductive Hypothesis}.
          Assume that for some $k \in \mathbb{N}$, if $m = k$, then $\sigma(M) \sim_{\TLORel{\sigma}{}} v \cdot M \cdot \overline{v}$.
    \item[--] \textbf{Inductive Step}.
          Assume that $m = k + 1$ and define $u = \interp{\sigma_1}_S \cdot \interp{\sigma_2}_S \cdots \interp{\sigma_k}_S$.
          Then by the inductive hypothesis $\sigma(M) \sim_{\TLORel{\sigma}{}} u \cdot \sigma_m(M) \cdot \overline{u}$.
          Since $\sigma_m$ is a transposition, then there exists some $j \in [n-1]$ such that $\sigma_m = \tau_{j,j+1}$.
          Furthermore, $\interp{\sigma_m}_S = \TLO{X}{j,j+1}$.
          If $j = a$, then $\sigma_m(M) = \TLO{(-1)}{a+1}$ then the following derivation holds using only \cref{Eq:O8DRels:Orig:Perm1,Eq:O8DRels:Orig:Perm7}.
          {\scriptsize\begin{equation*}
              \TLO{(-1)}{a+1} \leftarrow \TLO{(-1)}{a+1} \cdot \TLO{X}{j,j+1}{}^2 \leftarrow \TLO{X}{j,j+1} \cdot \TLO{(-1)}{a} \cdot \TLO{X}{j+1}
          \end{equation*}}%
          The case when $j + 1 = a$ follows symmetrically.
          When $j \ne a$ and $j + 1 \ne a$, then $\sigma_m(M) = M$ and the following derivation holds using only \cref{Eq:O8DRels:Orig:Perm1,Eq:O8DRels:Orig:Perm3}.
          {\scriptsize\begin{equation*}
              \TLO{(-1)}{a} \leftarrow \TLO{(-1)}{a} \cdot \TLO{X}{j,j+1}{}^2 \leftarrow \TLO{X}{j,j+1} \cdot \TLO{(-1)}{a} \cdot \TLO{X}{j+1}
          \end{equation*}}%
          In either case, $\sigma_m(M) \sim_{\TLORel{\sigma}{}} \interp{\sigma_m}_S \cdot M \cdot \interp{\sigma_m}_S$.
          Then $\sigma_m(M) \sim_{\TLORel{\sigma}{}} u \cdot \sigma_m(M) \cdot \overline{u} \sim_{\TLORel{\sigma}{}} v \cdot M \cdot \overline{v}$ and the inductive step is established.
    \end{enumerate}
    Then by the principle of induction,  $\sigma(M) \sim_{\TLORel{\sigma}{}} v \cdot M \cdot \overline{v}$.
    Clearly $\interp{v}_O^* = \interp{\sigma}_S$.
\end{proof}

\begin{lemma}
    \label{Lemma:PermuteK:1}
    For each four-level operator $M$ of type $K$, there exists a valid reindexing $\sigma$ for $M$ and a word $v$ over two-level operators of type $X$, such that $\interp{v}_O^* = \interp{\sigma}_S$ and $\TLO{K}{0,1,2,3} \sim_{\TLORel{\sigma}{}} v \cdot M \cdot \overline{v}$.
\end{lemma}

\begin{proof}
    Since $M$ is a four-level operator of type $K$, then there exists an increasing sequence $(a_0,a_1,a_2,a_3)$ over $[n]$ such that $M 
= \TLO{(-1)}{a_0,a_1,a_2,a_3}$.
    Since $(a_0,a_1,a_3,a_4)$ is increasing, then $k \le a_k$ for $k \in [4]$.
    Then for each $k \in [4]$, define $\sigma_k$ to be $\tau_{k,a_k}$ if $k \ne a_k$, or identity otherwise.
    Then the following equations hold.
    {\scriptsize\begin{align*}
        \sigma_0(M) &= \TLO{(-1)}{0,a_1,a_2,a_3}
        &
        \sigma_1(\sigma_0(M)) &= \TLO{(-1)}{0,1,a_2,a_3}
        \\
        \sigma_2(\sigma_1(\sigma_0(M))) &= \TLO{(-1)}{0,1,2,a_3}
        &
        \sigma_3(\sigma_2(\sigma_1(\sigma_0(M)))) &= \TLO{(-1)}{0,1,2,3}
    \end{align*}}%
    Let $v = \interp{\sigma_3}_S \cdot \interp{\sigma_2}_S \cdot \interp{\sigma_1}_S \cdot \interp{\sigma_0}_S$.
    Then the following derivations hold by \cref{Eq:O8DRels:Orig:Perm1,Eq:O8DRels:Orig:Perm8,Eq:O8DRels:Orig:Perm9,Eq:O8DRels:Orig:Perm10,Eq:O8DRels:Orig:Perm11}.
    We assume that each $\sigma_k$ is not the identity, else the derivation is trivial.
    {\scriptsize\begin{align*}
        \sigma_0(M) &\leftarrow \TLO{X}{0,a_0}{}^2 \cdot \TLO{K}{0,a_1,a_2,a_3}  \rightarrow \TLO{X}{0,a_0} \cdot \TLO{K}{a_0,a_1,a_2,a_3} \cdot \TLO{X}{0,a_0}
        \\
        \sigma_1(\sigma_0(M)) &\leftarrow \TLO{X}{1,a_1}{}^2 \cdot \TLO{K}{0,1,a_2,a_3} \rightarrow \TLO{X}{1,a_1} \cdot \TLO{K}{0,a_1,a_2,a_3} \cdot \TLO{X}{1,a_1} = \TLO{X}{1,a_1} \cdot \sigma_0(M) \cdot \TLO{X}{1,a_1}
        \\
        \sigma_2(\sigma_1(\sigma_0(M))) &\leftarrow \TLO{X}{2,a_2}{}^2 \cdot \TLO{K}{0,1,2,a_3} \rightarrow \TLO{X}{2,a_2} \cdot \TLO{K}{0,1,a_2,a_3} \cdot \TLO{X}{2,a_2} = \TLO{X}{2,a_2} \cdot \sigma_1(\sigma_0(M)) \cdot \TLO{X}{2,a_2}
        \\
        \sigma_3(\sigma_2(\sigma_1(\sigma_0(M)))) &\leftarrow \TLO{X}{3,a_3}{}^2 \cdot \TLO{K}{0,1,2,3} \rightarrow \TLO{X}{3,a_3} \cdot \TLO{K}{0,1,2,a_3} \cdot \TLO{X}{3,a_3} = \TLO{X}{3,a_3} \cdot \sigma_1(\sigma_0(M)) \cdot \TLO{X}{3,a_3}
    \end{align*}}%
    It follows that $\sigma(M) \sim_{\TLORel{\sigma}{}} v \cdot M \cdot \overline{v}$ where $\sigma = \sigma_3 \cdot \sigma_2 \cdot \sigma_1 \cdot \sigma_0$.
    Clearly $\interp{v}_O^* = \interp{\sigma}_S$.
\end{proof}

\begin{lemma}
    \label{Lemma:PermuteK:2}
    Let $M = \TLO{K}{0,1,2,3}$ be a four-level operator of dimension $n$.
    For any increasing sequence $(a_0,a_1,a_2,a_3)$ over $[n]$, there exists a valid reindexing $\sigma$ for $M$ and a word $v$ over two-level operators of type $X$, such that $\sigma(M) = \TLO{K}{a_0,a_1,a_2,a_3}$, $\interp{v}_O^* = \interp{\sigma}_S$, and $\TLO{K}{a_0,a_1,a_2,a_3} \sim_{\TLORel{\sigma}{}} v \cdot M \cdot \overline{v}$.
\end{lemma}

\begin{proof}
    Since $(a_0,a_1,a_3,a_4)$ is increasing, then $k \le a_k$ for $k \in [4]$.
    Then for each $k \in [4]$, define $\sigma_k$ to be $\tau_{k,a_k}$ if $k \ne a_k$, or identity otherwise.
    Then the following equations hold.
    {\scriptsize\begin{align*}
        \sigma_3(M) &= \TLO{(-1)}{0,1,2,a_3}
        &
        \sigma_2(\sigma_3(M)) &= \TLO{(-1)}{0,1,a_2,a_3}
        \\
        \sigma_1(\sigma_2(\sigma_3(M))) &= \TLO{(-1)}{0,a_1,a_2,a_3}
        &
        \sigma_0(\sigma_1(\sigma_2(\sigma_3(M)))) &= \TLO{(-1)}{a_0,a_2,a_3,a_4}
    \end{align*}}%
    Let $v = \interp{\sigma_0}_S \cdot \interp{\sigma_1}_S \cdot \interp{\sigma_2}_S \cdot \interp{\sigma_3}_S$.
    Then the following derivations hold by \cref{Eq:O8DRels:Orig:Perm1,Eq:O8DRels:Orig:Perm8,Eq:O8DRels:Orig:Perm9,Eq:O8DRels:Orig:Perm10,Eq:O8DRels:Orig:Perm11}.
    We assume that each $\sigma_k$ is not the identity, else the derivation is trivial.
    {\scriptsize\begin{align*}
        \sigma_3(M) &\leftarrow \TLO{K}{0,1,2,3} \cdot \TLO{X}{3,a_3}{}^2 \leftarrow \TLO{X}{3,a_3} \cdot \TLO{K}{0,1,2,a_3} \cdot \TLO{X}{3,a_3}
        \\
        \sigma_2(\sigma_3(M)) &\leftarrow \TLO{K}{0,1,a_2,a_3} \cdot \TLO{X}{2,a_2}{}^2 \leftarrow \TLO{X}{2,a_2} \cdot \TLO{K}{0,1,2,a_3} \cdot \TLO{X}{2,a_2} = \TLO{X}{2,a_2} \cdot \sigma_3(M) \cdot \TLO{X}{2,a_2}
        \\
        \sigma_1(\sigma_2(\sigma_3(M))) &\leftarrow \TLO{K}{0,a_1,a_2,a_3} \cdot \TLO{X}{1,a_1}{}^2 \leftarrow \TLO{X}{1,a_1} \cdot \TLO{K}{0,1,a_2,a_3} \cdot \TLO{X}{1,a_1} = \TLO{X}{1,a_1} \cdot \sigma_2(\sigma_3(M)) \cdot \TLO{X}{1,a_1}
        \\
        \sigma_0(\sigma_1(\sigma_2(\sigma_3(M)))) &\leftarrow \TLO{K}{a_0,a_1,a_2,a_3} \cdot \TLO{X}{0,a_0}{}^2 \leftarrow \TLO{X}{0,a_0} \cdot \TLO{K}{0,a_1,a_2,a_3} \cdot \TLO{X}{0,a_0} = \TLO{X}{0,a_0} \cdot \sigma_1(\sigma_2(\sigma_3(M))) \cdot \TLO{X}{0,a_0}
    \end{align*}}%
    It follows that $\sigma(M) \sim_{\TLORel{\sigma}{}} v \cdot M \cdot \overline{v}$ where $\sigma = \sigma_0 \cdot \sigma_1 \cdot \sigma_2 \cdot \sigma_3$.
    Clearly $\interp{v}_O^* = \interp{\sigma}_S$.
\end{proof}

\begin{lemma}
    \label{Lemma:PermuteK:3}
    If $\sigma \in \mathrm{S}(n)$ is a valid reindexing for a four-level operator $M$ of type $K$ and $\sigma(M) = M$, then there exists a word $v$ over the two-level operators of type $X$, such that $\interp{v}_O^* = \interp{\sigma}_S$ and $M \sim_{\TLORel{\sigma}{}} v \cdot M \cdot \overline{v}$.
\end{lemma}

\begin{proof}
    Since $M$ is a four-level operator of type $K$, then there exists an increasing sequence $(a,b,c,d)$ over $[n]$ such that $M 
= \TLO{(-1)}{a,b,c,d}$.
    Since $\sigma(M) = M$, then $\sigma$ fixes $\{ a, b, c, d \}$.
    Then $\sigma$ restricts to a permutation on $[n] \setminus \{ a, b, c, d \}$.
    Decompose this restriction of $\sigma$ into a sequence of transpositions $\sigma_1 \circ \sigma_2 \circ \cdots \circ \sigma_m$ on $[n] \setminus \{ a, b, c, d \}$.
    Since $\sigma$ fixes $\{ a, b, c, d \}$, then $\sigma = \sigma_1 \circ \sigma_2 \circ \cdots \circ \sigma_m$ when viewing each $\sigma_j$ as a permutation on $[n]$.
    Define $v = \interp{\sigma_1}_S \cdot \interp{\sigma_2}_S \cdots \interp{\sigma_m}_S$.
    It follows by induction on $m$ that $M \sim_{\TLORel{\sigma}{}} v \cdot M \cdot \overline{v}$.
    \begin{enumerate}
    \item[--] \textbf{Base Case}.
          If $m = 0$, then $v = \overline{v} = \epsilon$.
          Then $M \sim_{\TLORel{\sigma}{}} v \cdot M \cdot \overline{v}$ by the reflexivity of $( \sim_{\TLORel{\sigma}{}} )$.
    \item[--] \textbf{Inductive Hypothesis}.
          Assume that for some $k \in \mathbb{N}$, if $m = k$, then $M \sim_{\TLORel{\sigma}{}} v \cdot M \cdot \overline{v}$.
    \item[--] \textbf{Inductive Step}.
          Assume that $m = k + 1$ and define $u = \interp{\sigma_1}_S \cdot \interp{\sigma_2}_S \cdots \interp{\sigma_k}_S$.
          Then by the inductive hypothesis $M \sim_{\TLORel{\sigma}{}} u \cdot M \cdot \overline{u}$.
          Since $\sigma_m$ is a transposition of elements in $[n] \setminus \{ a, b, c, d \}$, then there exists some $j, l \in [n] \setminus \{ a, b, c, d \}$ such that $\sigma_m = \tau_{j,l}$.
          Furthermore, $\interp{\sigma_m}_S = \TLO{X}{j,l}$.
          Since $j, l \not \in \{ a, b, c, d \}$, then the following derivation holds by \cref{Eq:O8DRels:Orig:Perm1,Eq:O8DRels:Orig:Perm4}.
          {\scriptsize\begin{equation*}
              \TLO{K}{a,b,c,d} \leftarrow \TLO{X}{j,k}^2 \cdot \TLO{K}{a,b,c,d} \rightarrow \TLO{X}{j,l} \cdot \TLO{K}{a,b,c,d} \cdot \TLO{X}{j,l}
          \end{equation*}}%
          Then $M \sim_{\TLORel{\sigma}{}} \interp{\sigma_m}_S \cdot M \cdot \interp{\sigma_m}_S$.
          Then $M \ \sim_{\TLORel{\sigma}{}} u \cdot M \cdot \overline{u} \sim_{\TLORel{\sigma}{}} v \cdot M \cdot \overline{v}$.
    \end{enumerate}
    Then by the principle of induction, $M \sim_{\TLORel{\sigma}{}} v \cdot M \cdot \overline{v}$.
    Clearly $\interp{v}_O^* = \interp{\sigma}_S$.
\end{proof}

\begin{theorem}
    \label{Thm:PermuteK}
    If $\sigma \in \mathrm{S}(n)$ is a valid reindexing for a four-level operator $M$ of type $K$, then there exists a word $v$ over the two-level operators of type $X$, such that $\interp{v}_O^* = \interp{\sigma}_S$ and $\sigma(M) \sim_{\TLORel{\sigma}{}} v \cdot M \cdot \overline{v}$.
\end{theorem}

\begin{proof}
    Since $M$ is a four-level operator of type $K$, then there exists an increasing sequence $(a_0,a_1,a_2,a_3)$ over $[n]$ such that $M 
= \TLO{(-1)}{a_0,a_1,a_2,a_3}$.
    By \cref{Lemma:PermuteK:1}, there exists a word $u$ over two-level operators of type $X$, and a permutation $\sigma_1$ such that $\sigma_1(M) = \TLO{(-1)}{0,1,2,3}$, $\interp{\sigma_1}_S = \interp{u}_S^*$, and $\sigma_1(M) \sim_{\TLORel{\sigma}{}} u \cdot M \cdot \overline{u}$.
    By \cref{Lemma:PermuteK:2}, there exists a word $v$ over two-level operators of type $X$, and a permutation $\sigma_2$ such that $\sigma_2(\sigma_1(M)) = \TLO{(-1)}{\sigma(a),\sigma(b),\sigma(b),\sigma(c)} = \sigma(M)$, $\interp{\sigma_2}_S = \interp{v}_S^*$, and $\sigma_2(\sigma_1(M)) \sim_{\TLORel{\sigma}{}} v \cdot \sigma_1(M) \cdot \overline{v}$.
    Then define $\sigma_3 = \sigma \circ \sigma_1^{-1} \circ \sigma_2^{-1}$.
    Then $\sigma_3(\sigma(a_k)) = \sigma(\sigma_1^{-1}(\sigma_2^{-1}(\sigma(a_k)))) = \sigma(\sigma_1^{-1}(k)) = \sigma(a_k)$ for all $k \in [4]$.
    Then $\sigma_3$ is a valid reindexing for $M$ with $\sigma_3(\sigma_2(\sigma_1(M))) = M$.
    By \cref{Lemma:PermuteK:3}, there exists a word $w$ over two-level operators of type $X$, such that $\interp{\sigma_3}_S = \interp{w}_S^*$ and $\sigma(M) \sim_{\TLORel{\sigma}{}} w \cdot \sigma_2(\sigma_1(M)) \cdot \overline{w}$.
    It follows that $\sigma(M) \sim_{\TLORel{\sigma}{}} w \cdot \sigma_2(\sigma_1(M)) \overline{w} \sim_{\TLORel{\sigma}{}} w \cdot u \cdot \sigma_1(M) \cdot \overline{w \cdot u} \sim_{\TLORel{\sigma}{}} w \cdot u \cdot v \cdot M \cdot \overline{w \cdot u \cdot v}$.
    Moreover, $\interp{w \cdot u \cdot v}_O^* = \interp{w}_O^* \circ \interp{u}_O^* \circ \interp{v} = \interp{\sigma_3}_S \circ \interp{\sigma_2}_S \circ \interp{\sigma_1}_S = \interp{\sigma}_S$.
\end{proof}

\subsubsection{Permuting the Indices in Relations Over Multi-Level Operators}

\begin{theorem}
    \label{Lemma:PermRelElim}
    For each $\sigma \in \mathrm{S}( n )$, there exits a set of words $L_{\sigma}$ with the following properties.
    \begin{enumerate}
        \item If $v_1 \in L_{\sigma}$ and $v_2 \in L_{\sigma}$, then $v_1 \sim_{\TLORel{\sigma}{}} v_2$.
        \item If $\sigma$ is a valid reindexing for $w$, then there exists a $v \in L_{\sigma}$ such that $\sigma(w) \sim_{\TLOGen{\sigma}{}} v \cdot w \cdot \overline{v}$.
    \end{enumerate}
\end{theorem}

\begin{proof}
    Let $L_{\sigma}$ be the set of all words $v$ over the two-level operators of type $X$, such that $\interp{v}_O^* = \interp{\sigma}_S$.
    Let $v_1 \in L_{\sigma}$ and $v_2 \in L_{\sigma}$.
    Then $\interp{v_1}_O^* = \interp{\sigma}_S = \interp{v_2}_O^*$.
    Then $v_1 \sim_{\TLOGen{\sigma}{}} v_2$.
    Since $v_1$ and $v_2$ were arbitrary, then Property (1) holds.
    Now assume that $\sigma$ is a valid reindexing for $w$.
    Property (2) follows by induction on the length of $w$.
    \begin{enumerate}
        \item[--] \textbf{Base Case}.
              If $|w| = 0$, then $\sigma(w) = w$.
              Let $\sigma_1 \circ \sigma_2 \circ \cdots \circ \sigma_m$ be a decomposition of $\sigma$ into transpositions.
              Define $v = \interp{\sigma_1}_S \cdot \interp{\sigma_2}_S \cdots \interp{\sigma_m}$.
              Then $\interp{\overline{v}}_O^* = \interp{\sigma^{-1}}_S$ since $\overline{v}$ is the inverse to $v$.
              Since $v \cdot w \cdot \overline{v}$, then $\sigma(w) \sim_{\TLORel{\sigma}{}} v \cdot w \cdot \overline{v}$ by \cref{Lemma:FormalInverse}.
        \item[--] \textbf{Inductive Hypothesis}.
              Assume that for some $k \in \mathbb{N}$, if $|w| = k$, then there exists a $v \in L_{\sigma}$ such that $\sigma(w) \sim_{\TLOGen{\sigma}{}} v \cdot w \cdot \overline{v}$.
        \item[--] \textbf{Inductive Step}.
              Assume that $|w| = k + 1$.
              Then there exists some word $\hat{w}$ over $\TLOGen{n}{}$ and some $M \in \TLOGen{n}{}$ such that $w = \hat{w} \cdot M$ with $|\hat{w}| = k$.
              Clearly, $M$ is either of type $X$, type $(-1)$, or type $K$.
              In any case, there exists a word $v$ over the two-level operators of type $X$ such that $\interp{v}_O^* = \interp{\sigma}_S$ and $\sigma(M) \sim_{\TLORel{\sigma}{}} v \cdot M \cdot \overline{v}$.
              Then $v \in L_{\sigma}$.
              By the inductive hypothesis, there exists a $u \in L_{\sigma}$ such that $\sigma(\hat{w}) \sim_{\TLORel{\sigma}{}} u \cdot \hat{w} \cdot \overline{u}$.
              Since $u \in L_{\sigma}$ and $v \in L_{\sigma}$, then $u \sim_{\TLORel{\sigma}{}} v$ by Property (1).
              Furthermore $\overline{v} \cdot v \sim_{\TLORel{\sigma}{}} \epsilon$ by \cref{Lemma:FormalInverse}.
              Then $\overline{v} \cdot u  \sim_{\TLORel{\sigma}{}} \overline{v} \cdot v  \sim_{\TLORel{\sigma}{}} \epsilon$.
              Since $\sigma(w) = \sigma(\hat{w}) \cdot \sigma(M)$, then it follows $\sigma(w) \sim_{\TLORel{\sigma}{}} v \cdot \hat{w} \cdot \overline{v} \cdot \sigma(M) \sim_{\TLORel{\sigma}{}} v \cdot \hat{w} \cdot \overline{v} \cdot u \cdot M \cdot \overline{u} \sim_{\TLORel{\sigma}{}} v \cdot \hat{w} \cdot M \cdot \overline{u} \sim_{\TLORel{\sigma}{}} u \cdot w \cdot \overline{u}$ and the inductive step is established.
    \end{enumerate}
    Then by the principle of induction, Property (2) holds.
\end{proof}

\begin{corollary}
    \label{Thm:PermRelElim}
    Let $v$ and $w$ be words over $\TLOGen{n}{}$.
    If $\sigma$ is a valid reindexing for $u$ and $w$, then $\sigma(w)$ is derivable from $\sigma(u)$ using $\TLORel{\sigma}{} \cup \{ u \approx w \}$.
\end{corollary}

\begin{proof}
    Let $Q = \TLORel{\sigma}{} \cup \{ u \approx w \}$.
    By \cref{Lemma:PermRelElim}, there exists words $v_1$ and $v_2$ over two-level operators of type $X$, such that $v_1 \sim_Q v_2$,  $\sigma(u) \sim_Q v_1 \cdot u \cdot \overline{v}_1$ and $\sigma(w) \sim_Q v_2 \cdot w \cdot \overline{v}_2$.
    Since $v_1 \sim_Q v_2$, then $\interp{v_1}_O^* = \interp{v_2}_O^*$.
    Then $\interp{\overline{v}_1}_O^* = \interp{\overline{v}_2}_O^*$.
    Since $\overline{v}_1$ and $\overline{v}_2$ are words over two-level operators of type $X$, then $\overline{v}_1 \sim_{Q} \overline{v}_2$ by \cref{Lemma:TypeXBraids}.
    Then the following derivation holds over $Q$.
    {\scriptsize\begin{align*}
        \sigma(u)
        \rightarrow
        v_1 \cdot u \cdot \overline{v}_1
        \rightarrow
        v_2 \cdot u \cdot \overline{v}_1
        \rightarrow
        v_2 \cdot w \cdot \overline{v}_1
        \rightarrow
        v_2 \cdot w \cdot \overline{v}_2
        \rightarrow
        \sigma(w)
    \end{align*}}%
    Then $\sigma(w)$ is derivable from $\sigma(w)$ using $Q$.
\end{proof}

\subsection{The Set of Representative Relations}
\label{Appendix:O8D:Reps}

\begin{figure}[t]
  {\scriptsize\begin{minipage}[t]{0.35\textwidth}
    \begin{align}
    \TLO{X}{a,b}{}^2 &\,\approx \, \epsilon \\
    \TLO{(-1)}{0}{}^2 &\,\approx \, \epsilon \\
    \TLO{K}{0,1,2,3}{}^2 &\,\approx \, \epsilon \\
    \TLO{X}{a,b} \cdot \TLO{X}{c,d} &\,\approx \, \TLO{X}{c,d} \cdot \TLO{X}{a,b} \\
    \TLO{X}{a,b} \cdot \TLO{(-1)}{c} &\,\approx \, \TLO{(-1)}{c} \cdot \TLO{X}{a,b} \\
    \TLO{X}{a,b} \cdot \TLO{K}{c,d,e,f} &\,\approx \, \TLO{K}{c,d,e,f} \cdot \TLO{X}{a,b} \\
    \TLO{(-1)}{4} \cdot \TLO{K}{0,1,2,3} &\,\approx \, \TLO{K}{0,1,2,3} \cdot \TLO{(-1)}{4} \\
    \TLO{(-1)}{0} \cdot \TLO{(-1)}{4} &\,\approx \, \TLO{(-1)}{4} \cdot \TLO{(-1)}{0} \\
    \TLO{K}{0,1,2,3} \cdot \TLO{K}{4,5,6,7} &\,\approx \, \TLO{K}{4,5,6,7} \cdot \TLO{K}{0,1,2,3} \\
    \TLO{X}{a,a+2} \cdot \TLO{X}{a,a+1} &\,\approx \, \TLO{X}{a+1,a+2} \cdot \TLO{X}{a,a+2}
    \end{align}
  \end{minipage}
  \begin{minipage}[t]{0.63\textwidth}
    \begin{align}
    \TLO{X}{a+1,a+2} \cdot \TLO{X}{a,a+1} &\,\approx \, \TLO{X}{a,a+2} \cdot \TLO{X}{a+1,a+2} \\
    \TLO{X}{a,b} \cdot \TLO{(-1)}{a} &\,\approx \, \TLO{(-1)}{b} \cdot \TLO{X}{a,b} \\
    \TLO{X}{a,e} \cdot \TLO{K}{a,b,c,d} &\,\approx \, \TLO{K}{e,b,c,d} \cdot \TLO{X}{a,e} \\
    \TLO{X}{b,e} \cdot \TLO{K}{a,b,c,d} &\,\approx \, \TLO{K}{a,e,c,d} \cdot \TLO{X}{b,e} \\
    \TLO{X}{c,e} \cdot \TLO{K}{a,b,c,d} &\,\approx \, \TLO{K}{a,b,e,d} \cdot \TLO{X}{c,e} \\
    \TLO{X}{d,e} \cdot \TLO{K}{a,b,c,d} &\,\approx \, \TLO{K}{a,b,c,e} \cdot \TLO{X}{d,e} \\
    \TLO{X}{0,1} \cdot \TLO{K}{0,1,2,3} &\,\approx \, \TLO{K}{0,1,2,3} \cdot \TLO{X}{0,1} \cdot \TLO{(-1)}{1} \cdot \TLO{(-1)}{3} \\
    \TLO{X}{1,2} \cdot \TLO{K}{0,1,2,3} &\,\approx \, \TLO{(-1)}{0} \cdot \TLO{K}{0,1,2,3} \cdot \TLO{(-1)}{0} \cdot \TLO{K}{0,1,2,3} \cdot \TLO{(-1)}{0} \\
    \TLO{X}{2,3} \cdot \TLO{K}{0,1,2,3} &\,\approx \, \TLO{K}{0,1,2,3} \cdot \TLO{X}{1,3} \\
    \TLO{K}{0,1,2,3} \cdot \TLO{K}{1,3,4,5} &\,\approx \, \TLO{K}{1,3,4,5} \cdot \TLO{K}{0,1,2,3} \\
    \TLO{(-1)}{0} \cdot \TLO{(-1)}{4} \cdot \TLO{X}{0,4} \cdot \rho &\,\approx\, \rho \cdot \TLO{X}{0,4} \cdot \TLO{(-1)}{4} \cdot \TLO{(-1)}{0}
    \end{align}
  \end{minipage}}
  \caption{The representative relations in $\TLORel{n}{1}$, for all valid choices of $a, b, c, d, e, f \in \mathbb{Z}$. We write $\rho$ for the substring $\TLO{K}{4,5,6,7} \cdot \TLO{K}{0,1,2,3} \cdot \TLO{X}{3,4} \cdot \TLO{K}{0,1,2,3} \cdot \TLO{K}{4,5,6,7} \cdot \TLO{X}{0,4}$.}
  \label{Fig:O8DRels:Reps}
  \hfill\rule{0.95\textwidth}{0.1mm}\hfill\hbox{}
\end{figure}

In \cref{Sec:O8D:RepRels}, a set of representative relations were selected from $\TLORel{n}{}$.
These relations are illustrated in \cref{Fig:O8DRels:Reps}.
In some sense, the choice of representative relations were arbitrary, since all choices are equivalent up to permutation.
However, preference was given to the parameters $[0]$, $[4]$, $[0,1,2,3]$, $[4,5,6,7]$, since these correspond well to controlled qubit operators.

\subsection{Proving the Redundant Relations are Derivable}
\label{Appendix:O8D:Reduce}

This section makes use of the braiding relations and the inverse relations, to derive several bifunctoriality and commutator relations.
Each proof follows the same structure.
First, the special case is proven where all generators of type $(-1)$ or $K$ have consecutive indices starting from $0$.
In all other cases, there is a generator of type $(-1)$ or $K$ conjugated by a permutation.
The braiding relations are used to obtain a convenient decomposition for each permutation.
The commutativity and bifunctoriality follow immediately from these decompositions.

\begin{lemma}
    \label{Lemma:ZTxPull}
    If $( 0, a, b )$ is an increasing sequence over $[n]$, then $\TLO{X}{a,b} \cdot \TLO{(-1)}{0} \sim_{\TLORel{n}{3}} \TLO{(-1)}{0} \cdot \TLO{X}{a,b}$.
\end{lemma}

\begin{proof}
    Let $\sigma = \tau_{a,b}$.
    Since $3 < a < b$, then there exists a decomposition $\tau_{c_1,c_1+1} \circ \tau_{c_2,c_2+1} \circ \cdots \tau_{c_m,c_m+1}$ of $\sigma$ into transpositions such that $0 < c_k$ for all $k \in [n]$.
    Define $u = \TLO{X}{c_1,c_1+1} \cdot \TLO{X}{c_2,c_2+1} \cdots \TLO{X}{c_m,c_m+1}$.
    Then $\interp{\TLO{X}{a,b}}_O^* = \interp{\sigma}_S = \interp{u}_O^*$.
    Then $\TLO{X}{a,b} \sim_{\TLORel{n}{3}} u$ by \cref{Lemma:TypeXBraids}.
    The proof follows by induction on $m$.
    \begin{enumerate}
    \item[--] \textbf{Base Case}.
          If $|u| = 0$, then $u \cdot \TLO{(-1)}{0} = \TLO{(-1)}{0} \cdot u$.
          Then $u \cdot \TLO{(-1)}{0} \sim_{\TLORel{n}{3}} \TLO{(-1)}{0} \cdot u$ by reflexivity.
    \item[--] \textbf{Inductive Hypothesis}.
          Assume that for some $k \in \mathbb{N}$, if $|u| = k$, then $u \cdot \TLO{(-1)}{0} \sim_{\TLORel{n}{3}} \TLO{(-1)}{0} \cdot u$.
    \item[--] \textbf{Inductive Step}.
          Assume that $m = k + 1$.
          Define $v = \TLO{X}{c_1,c_1+1} \cdot \TLO{X}{c_2,c_2+1} \cdots \TLO{X}{c_k,c_k+1}$.
          Then by definition $u = v \cdot \TLO{X}{c_m,c_m+1}$.
          Since $c_m > 0$, then $u \cdot \TLO{(-1)}{0} \sim_{\TLORel{n}{3}} v \cdot \TLO{(-1)}{0} \cdot \TLO{X}{c_m,c_m+1}$ by the relation $\TLO{X}{c_m,c_m+1} \cdot \TLO{(-1)}{0} \approx_{\TLORel{n}{3}} \TLO{(-1)}{0} \cdot \TLO{X}{c_m,c_m+1}$.
          Furthermore, since $|v| = k$, then by the inductive hypothesis $v \cdot \TLO{(-1)}{0} \sim_{\TLORel{n}{3}} \TLO{(-1)}{0} \cdot v$.
          Then $v \cdot \TLO{(-1)}{0} \cdot \TLO{X}{c_m,c_m+1} \sim_{\TLORel{n}{3}} \TLO{(-1)}{0} \cdot u$.
          Then by the transitivity of $( \sim_{\TLORel{n}{3}} )$, $u \cdot \TLO{(-1)}{0} \sim_{\TLORel{n}{3}} \TLO{(-1)}{0} \cdot u$ and the inductive step is established.
    \end{enumerate}
    Then $u \cdot \TLO{(-1)}{0} \sim_{\TLORel{n}{3}} \TLO{(-1)}{0} \cdot u$ by the principle of induction.
    Since $\TLO{X}{a,b} \cdot \TLO{(-1)}{0} \sim_{\TLORel{n}{3}} u \cdot \TLO{(-1)}{0}$ and $\TLO{(-1)}{0} \cdot u \sim_{\TLORel{n}{3}} \TLO{(-1)}{0} \cdot \TLO{X}{a,b}$, then $\TLO{X}{a,b} \cdot \TLO{(-1)}{0} \sim_{\TLORel{n}{3}} \TLO{(-1)}{0} \cdot \TLO{X}{a,b}$ by the transitivity of $( \sim_{\TLORel{n}{3}} )$.
\end{proof}

\begin{theorem}
    All instances of \cref{Eq:O8DRels:Orig:Perm3} are derivable from $\TLORel{n}{3}$.
\end{theorem}

\begin{proof}
    Let $\{ a, b, c \} \in [n]$.
    Define $\sigma \in \mathrm{S}(n)$ such that $\sigma$ is $\tau_{k,c_k}$ if $c > 0$, or identity otherwise.
    Likewise, define $u$ to be $\TLO{X}{0,c}$ if $c > 0$, or $\epsilon$ otherwise.
    Clearly $\tau_{a,b} \circ \sigma = \sigma \circ \tau_{\sigma(a),\sigma(b)}$.
    Then, $$\interp{\TLO{X}{a,b} \cdot u}_O^* = \interp{\tau_{a,b} \circ \sigma}_S = \interp{\sigma \circ \tau_{\sigma(a),\sigma(b)}}_S = \interp{u \cdot \TLO{X}{\sigma(a),\sigma(b)}}_O^*.$$
    Then $\TLO{X}{a,b} \cdot u \sim_{\TLORel{n}{3}} u \cdot \TLO{X}{\sigma(a),\sigma(b)}$ by \cref{Lemma:TypeXBraids}.
    Likewise, $\TLO{X}{\sigma(a), \sigma(b)} \cdot \overline{u} \sim_{\TLORel{n}{3}} \overline{u} \cdot \TLO{X}{a,b}$ by \cref{Lemma:FormalInverse}.
    Since $\{ a, b, c \}$ are distinct, then $\sigma(a) > 0$ and $\sigma(b) > 0$.
    Then $\TLO{X}{\sigma{a},\sigma{b}} \cdot \TLO{(-1)}{0} \sim_{\TLORel{n}{3}} \TLO{(-1)}{0} \cdot \TLO{X}{\sigma{a},\sigma{b}}$ by \cref{Lemma:ZTxPull}.
    Then the following derivation holds.
    {\scriptsize\begin{equation*}
        \TLO{X}{a,b} \cdot u \cdot \TLO{(-1)}{0} \cdot \overline{u}
        \sim_{\TLORel{n}{3}}
        u \cdot \TLO{X}{\sigma(a),\sigma(b)} \cdot \TLO{(-1)}{0} \cdot \overline{u}
        \sim_{\TLORel{n}{3}}
        u \cdot \TLO{(-1)}{0} \cdot \TLO{X}{\sigma(a),\sigma(b)} \cdot \overline{u}
        \sim_{\TLORel{n}{3}}
        u \cdot \TLO{(-1)}{0} \cdot \overline{u} \cdot \TLO{X}{a,b}
    \end{equation*}}
    Since $\{ a, b, c \}$ were arbitrary, then all instances of \cref{Eq:O8DRels:Orig:Perm3} are derivable from $\TLORel{n}{3}$.
\end{proof}

\begin{lemma}
    \label{Lemma:KTxPull}
    If $( 3, a, b )$ is an increasing sequence over $[n]$, then $\TLO{X}{a,b} \cdot \TLO{K}{0,1,2,3} \sim_{\TLORel{n}{3}} \TLO{K}{0,1,2,3} \cdot \TLO{X}{a,b}$.
\end{lemma}

\begin{proof}
    Let $\sigma = \tau_{a,b}$.
    Since $3 < a < b$, then there exists a decomposition $\tau_{c_1,c_1+1} \circ \tau_{c_2,c_2+1} \circ \cdots \tau_{c_m,c_m+1}$ of $\sigma$ into transpositions such that $3 < c_k$ for all $k \in [n]$.
    Define $u = \TLO{X}{c_1,c_1+1} \cdot \TLO{X}{c_2,c_2+1} \cdots \TLO{X}{c_m,c_m+1}$.
    Then $\interp{\TLO{X}{a,b}}_O^* = \interp{\sigma}_S = \interp{u}_O^*$.
    Then $\TLO{X}{a,b} \sim_{\TLORel{n}{3}} u$ by \cref{Lemma:TypeXBraids}.
    The proof follows by induction on $m$.
    \begin{enumerate}
    \item[--] \textbf{Base Case}.
          If $|u| = 0$, then $u \cdot \TLO{K}{0,1,2,3} \sim_{\TLORel{n}{3}} \TLO{K}{0,1,2,3} \cdot u$ by reflexivity.
    \item[--] \textbf{Inductive Hypothesis}.
          Assume that for some $k \in \mathbb{N}$, if $|u| = k$, then $u \cdot \TLO{K}{0,1,2,3} \sim_{\TLORel{n}{3}} \TLO{K}{0,1,2,3} \cdot u$.
    \item[--] \textbf{Inductive Step}.
          Assume that $m = k + 1$.
          Define $v = \TLO{X}{c_1,c_1+1} \cdot \TLO{X}{c_2,c_2+1} \cdots \TLO{X}{c_k,c_k+1}$.
          Then by definition $u = v \cdot \TLO{X}{c_m,c_m+1}$.
          Since $c_m > 3$, then $u \cdot \TLO{K}{0,1,2,3} \sim_{\TLORel{n}{3}} v \cdot \TLO{K}{0,1,2,3} \cdot \TLO{X}{c_m,c_m+1}$ by the relation $\TLO{X}{c_m,c_m+1} \cdot \TLO{K}{0,1,2,3} \approx_{\TLORel{n}{3}} \TLO{K}{0,1,2,3} \cdot \TLO{X}{c_m,c_m+1}$.
          Since $|v| = k$, then by the inductive hypothesis $v \cdot \TLO{K}{0,1,2,3} \sim_{\TLORel{n}{3}} \TLO{K}{0,1,2,3} \cdot v$.
          Then $v \cdot \TLO{K}{0,1,2,3} \cdot \TLO{X}{c_m,c_m+1} \sim_{\TLORel{n}{3}} \TLO{K}{0,1,2,3} \cdot u$.
          Then by the transitivity of $( \sim_{\TLORel{n}{3}} )$, $u \cdot \TLO{K}{0,1,2,3} \sim_{\TLORel{n}{3}} \TLO{K}{0,1,2,3} \cdot u$ and the inductive step is established.
    \end{enumerate}
    Then $u \cdot \TLO{K}{0,1,2,3} \sim_{\TLORel{n}{3}} \TLO{K}{0,1,2,3} \cdot u$ by the principle of induction.
    Since $\TLO{X}{a,b} \cdot \TLO{K}{0,1,2,3} \sim_{\TLORel{n}{3}} u \cdot \TLO{K}{0,1,2,3}$ and $\TLO{K}{0,1,2,3} \cdot u \sim_{\TLORel{n}{3}} \TLO{K}{0,1,2,3} \cdot \TLO{X}{a,b}$, then $\TLO{X}{a,b} \cdot \TLO{K}{0,1,2,3} \sim_{\TLORel{n}{3}} \TLO{K}{0,1,2,3} \cdot \TLO{X}{a,b}$ by the transitivity of $( \sim_{\TLORel{n}{3}} )$.
\end{proof}

\begin{theorem}
    All instances of \cref{Eq:O8DRels:Orig:Perm4} are derivable from $\TLORel{n}{3}$.
\end{theorem}

\begin{proof}
    Let $( c_0, c_1, c_2, c_3 )$ an increasing sequence over $[n]$.
    Since $(c_0,c_1,c_2,c_3)$ is increasing, then $k \le c_k$ for $k \in [4]$.
    Then for each $k \in [4]$, define $\sigma_k$ to be $\tau_{k,c_k}$ if $k \ne c_k$, or identity otherwise, and define $\sigma = \sigma_0 \circ \sigma_1 \circ \sigma_2 \circ \sigma_3$.
    Likewise, for each $k \in [4]$, define $u_k$ to be $\TLO{X}{k,c_k}$ if $k \ne c_k$, or $\epsilon$ otherwise, and let $u = u_0 \cdot u_1 \cdot u_2 \cdot u_3$.
    Clearly $\tau_{a,b} \circ \sigma = \sigma \circ \tau_{\sigma(a),\sigma(b)}$.
    Then,
    {\scriptsize\begin{equation*}
        \interp{\TLO{X}{a,b} \cdot u}_O^*
        =
        \interp{\tau_{a,b} \circ \sigma}_S
        =
        \interp{\sigma \circ \tau_{\sigma(a),\sigma(b)}}_S = \interp{u \cdot \TLO{X}{\sigma(a),\sigma(b)}}_O^*.
    \end{equation*}}%
    Then $\TLO{X}{a,b} \cdot u \sim_{\TLORel{n}{3}} u \cdot \TLO{X}{\sigma(a),\sigma(b)}$ by \cref{Lemma:TypeXBraids}.
    Likewise, $\TLO{X}{\sigma(a),\sigma(b)} \cdot \overline{u} \sim_{\TLORel{n}{3}} \overline{u} \cdot \TLO{X}{a,b}$ by \cref{Lemma:FormalInverse}.
    Since $a,b \not \in \{ c_0, c_1, c_2, c_3 \}$, then $\sigma(a) > 3$ and $\sigma(b) > 3$.
    Then $\TLO{X}{\sigma{a},\sigma{b}} \cdot \TLO{K}{0,1,2,3} \sim_{\TLORel{n}{3}} \TLO{K}{0,1,2,3} \cdot \TLO{X}{\sigma{a},\sigma{b}}$ by \cref{Lemma:KTxPull}.
    Then the following derivation holds.
    {\scriptsize\begin{equation*}
        \TLO{X}{a,b} \cdot u \cdot \TLO{K}{0,1,2,3} \cdot \overline{u}
        \sim_{\TLORel{n}{3}}
        u \cdot \TLO{X}{\sigma(a),\sigma(b)} \cdot \TLO{K}{0,1,2,3} \cdot \overline{u}
        \sim_{\TLORel{n}{3}}
        u \cdot \TLO{K}{0,1,2,3} \cdot \TLO{X}{\sigma(a),\sigma(b)} \cdot \overline{u}
        \sim_{\TLORel{n}{3}}
        u \cdot \TLO{K}{0,1,2,3} \cdot \overline{u} \cdot \TLO{X}{a,b}
    \end{equation*}}%
    Since $\{ a, b, c_0, c_1, c_2, c_3 \}$ were arbitrary, then all instances of \cref{Eq:O8DRels:Orig:Perm4} are derivable from $\TLORel{n}{3}$.
\end{proof}

\begin{lemma}
    \label{Lemma:KPullPerm}
    Let $( c_0, a, c_1, c_2, c_3 )$ be an increasing sequence over $[n]$.
    For each $k \in [4]$, define $\sigma_k$ to be $\tau_{k,c_k}$ if $k \ne c_k$, or identity otherwise.
    If $\sigma = \sigma_0 \cdot \sigma_1 \cdot \sigma_2 \cdot \sigma_3$ and $\rho = \tau_{0,a} \cdot \sigma_1 \cdot \sigma_2 \cdot \sigma_3$, then there exists an $\alpha \in \mathrm{S}(n)$ such that $\tau_{c_0,a} \circ \sigma = \rho \circ \alpha$ with $\alpha$ fixing $[4]$.
\end{lemma}

\begin{proof}
    By definition, $\rho(0) = a = \tau_{c_0,a}(\sigma(0))$, $\rho(1) = c_1 = \tau_{c_0,a}(\sigma(0))$, $\rho(1) = c_2 = \tau_{c_0,a}(\sigma(0))$, and $\rho(2) = c_2 = \tau_{c_0,a}(\sigma(0))$.
    Since $\mathrm{S}(n)$ is a group, then there exists an $\alpha \in \mathrm{S}(n)$ such that $\tau_{c_0,a} \circ \sigma = \rho \circ \alpha$.
    Assume that there exists a $k \in [4]$ such that $\alpha(k) \ne k$.
    Then $(\tau_{c_0,a} \circ \sigma)(\alpha(k)) \ne (\tau_{c_0,a} \circ \sigma)(k) = \rho(k)$.
    Then by contradiction, $\alpha$ fixes $[4]$.
    Then $\alpha$ decomposes into a sequence of transpositions over $[n] \setminus [4]$.
\end{proof}

\begin{lemma}
    \label{Lemma:KPull}
    Let $w = \TLO{X}{a_0,b_0} \cdot \TLO{X}{a_1,b_1} \cdots \TLO{X}{a_m,b_m}$ such that $a_k > 3$ and $b_k > 3$ for all $k \in [m+1]$.
    Then $w \cdot \TLO{K}{0,1,2,3} \sim_{\TLORel{n}{3}} \TLO{K}{0,1,2,3} \cdot w$.
\end{lemma}

\begin{proof}
    Let proof follows by induction on $|w|$.
    \begin{enumerate}
    \item[--] \textbf{Base Case}.
          If $|w| = 0$, then $w = \epsilon$ and $w \cdot \TLO{K}{0,1,2,3} \sim_{\TLORel{n}{3}} \TLO{K}{0,1,2,3} \cdot w$ by the transitivity of $( \sim_{\TLORel{n}{3}} )$.
    \item[--] \textbf{Inductive Hypothesis}.
          Assume that for some $k \in \mathbb{N}$, if $|w| = k$, then $w \cdot \TLO{K}{0,1,2,3} \sim_{\TLORel{n}{3}} \TLO{K}{0,1,2,3} \cdot w$
    \item[--] \textbf{Inductive Step}.
          Assume that $|w| = k + 1$.
          Define $v = \TLO{X}{a_0,b_0} \cdot \TLO{X}{a_1,b_1} \cdots \TLO{X}{a_k,b_k}$.
          Then by definition $w = v \cdot \TLO{X}{a_m,b_m}$.
          Since $a_m > 3$ and $b_m > 3$, then $\TLO{X}{a_m,b_m} \cdot \TLO{K}{0,1,2,3} \sim_{\TLORel{n}{3}} \TLO{K}{0,1,2,3} \cdot \TLO{X}{a_m,b_m}$ by \cref{Lemma:KTxPull}.
          As a result, $w \cdot \TLO{K}{0,1,2,3} \sim_{\TLORel{n}{3}} v \cdot \TLO{K}{0,1,2,3} \cdot \TLO{X}{a_m,b_m}$.
          Then by the inductive hypothesis, $v \cdot \TLO{K}{0,1,2,3} \sim_{\TLORel{n}{3}} \TLO{K}{0,1,2,3} \cdot v$.
          As a result, $v \cdot \TLO{K}{0,1,2,3} \cdot \TLO{X}{a_m,b_m} \sim_{\TLORel{n}{3}} \TLO{K}{0,1,2,3} \cdot w$.
          Then $w \cdot \TLO{K}{0,1,2,3} \sim \TLO{K}{0,1,2,3} \cdot w$ by the transitivity of $( \sim_{\TLORel{n}{3}} )$ and the inductive step is established.
    \end{enumerate}
    Then by the principle of induction, $w \cdot \TLO{K}{0,1,2,3} \sim_{\TLORel{n}{3}} \TLO{K}{0,1,2,3} \cdot w$.
\end{proof}

\begin{theorem}
    All instances of \cref{Eq:O8DRels:Orig:Perm8,Eq:O8DRels:Orig:Perm9,Eq:O8DRels:Orig:Perm10,Eq:O8DRels:Orig:Perm11} are derivable from $\TLORel{n}{3}$.
\end{theorem}

\begin{proof}
    Let be $( c_0, a, c_1, c_2, c_3 )$ an increasing sequence.
    For each $k \in [4]$, define $\sigma_k$ to be $\tau_{k,c_k}$ if $k \ne c_k$, or identity otherwise.
    Then define $\sigma = \sigma_0 \circ \sigma_1 \circ \sigma_2 \circ \sigma_3$ and $\rho = \tau_{0,a} \cdot \sigma_1 \cdot \sigma_2 \cdot \sigma_3$.
    By \cref{Lemma:KPullPerm}, there exists a sequence of transpositions $\alpha = \alpha_0 \circ \alpha_1 \circ \cdots \alpha_m$ such that $\tau_{c_0,a} \circ \sigma = \rho \circ \alpha$.
    Then define $w = \interp{\alpha_0}_S \cdot \interp{\alpha_1}_S \cdots \interp{\alpha_m}_S$.
    Next, for each $k \in [4]$, define $u_k$ to be $\TLO{X}{k,c_k}$ if $k \ne c_k$, or $\epsilon$ otherwise.
    Then define $u = u_0 \cdot u_1 \cdot u_2 \cdot u_3$, $v = \TLO{X}{c_0,a} \cdot u_1 \cdot u_2 \cdot u_3$, and $w = \interp{\alpha_0}_S \cdot \interp{\alpha_1}_S \cdots \interp{\alpha_m}_S$.
    It follows that $\interp{u}_O^* = \interp{\tau_{c_0,a} \circ \sigma}_S = \interp{\rho \circ \alpha}_S = \interp{v \cdot w}$.
    Since $u$, $v$, and $w$ are words over two-level operators of type $X$, then $\TLO{X}{c_0,a} \cdot u \sim_{\TLORel{n}{3}} v \cdot w$.
    Likewise, by \cref{Lemma:FormalInverse}, $\overline{u} \cdot \TLO{X}{c_0,a} \sim_{\TLORel{n}{3}} \overline{w} \cdot \overline{v}$.
    Then $\TLO{X}{c_0,a} \cdot u$ acts by conjugation on $\TLO{K}{0,1,2,3}$ as follows.
    {\scriptsize\begin{equation*}
        \TLO{X}{c_0,a} \cdot u \cdot \TLO{K}{0,1,2,3} \cdot \overline{u} \cdot \TLO{X}{c_0,1}
        \sim_{\TLORel{n}{3}}
        v \cdot w \cdot \TLO{K}{0,1,2,3} \overline{u} \cdot \TLO{X}{c_0,1}
        \sim_{\TLORel{n}{3}}
        v \cdot w \cdot \TLO{K}{0,1,2,3} \cdot \overline{w} \cdot \overline{v}
    \end{equation*}}%
    Then by \cref{Lemma:KPull}, $w \cdot \TLO{K}{0,1,2,3} \sim_{\TLORel{n}{3}} \TLO{K}{0,1,2,3} \cdot w$.
    Furthermore, $w \cdot \overline{w} \sim_{\TLORel{n}{3}} \epsilon$ \cref{Lemma:FormalInverse}.
    Then $w$ acts by conjugation on $\TLO{K}{0,1,2,3}$ as follows.
    {\scriptsize\begin{equation*}
        w \cdot \TLO{K}{0,1,2,3} \cdot \overline{w}
        \sim_{\TLORel{n}{3}}
        \TLO{K}{0,1,2,3} \cdot w \cdot \overline{w}
        \sim_{\TLORel{n}{3}}
        \TLO{K}{0,1,2,3}
    \end{equation*}}
    Since $\TLO{X}{c_0,a}{}^2 \sim_{\TLORel{n}{3}} \epsilon$ by \cref{Lemma:TypeXBraids}, then the following derivation also holds.
    {\scriptsize\begin{equation*}
        \TLO{X}{c_0,a} \cdot u \cdot \TLO{K}{0,1,2,3} \cdot \overline{u}
        \sim_{\TLORel{n}{3}}
        \TLO{X}{c_0,a} \cdot u \cdot \TLO{K}{0,1,2,3} \cdot \overline{u} \cdot \TLO{X}{c_0,1}{}^2
        \sim_{\TLORel{n}{3}}
        v \cdot w \cdot \TLO{K}{0,1,2,3} \cdot \overline{w} \cdot \overline{v} \cdot \TLO{X}{c_0,1}
        \sim_{\TLORel{n}{3}}
        v \cdot \TLO{K}{0,1,2,3} \cdot \overline{v} \cdot \TLO{X}{c_0,1}
    \end{equation*}}%
    Since $\{ a, b, c_0, c_1, c_2, c_3 \}$ were all arbitrary, then all instances of \cref{Eq:O8DRels:Orig:Perm8} holds.
    The cases of \cref{Eq:O8DRels:Orig:Perm9,Eq:O8DRels:Orig:Perm10,Eq:O8DRels:Orig:Perm11} follow symmetrically.
\end{proof}

\end{document}